\documentclass[12pt]{article}
\usepackage[utf8]{inputenc}
\usepackage[letterpaper,margin=1in]{geometry}
\usepackage[margin=1in]{geometry}
\usepackage[fleqn]{amsmath}
\usepackage{amssymb}
\usepackage{amsthm}
\usepackage{thmtools}
\usepackage{extarrows}
\usepackage{enumerate}
\usepackage{graphicx}
\usepackage{indentfirst}
\usepackage{natbib}
\usepackage{verbatim}
\usepackage{booktabs}
\usepackage{xcolor}
\usepackage{titlesec}
\usepackage{longtable}
\usepackage{bookmark}
\usepackage{lscape}
\usepackage{microtype}

\newtheorem{theorem}{Theorem}
\newtheorem{lemma}{Lemma}

\newtheorem{assumption}{Assumption}
\newtheorem{definition}{Definition}

\newtheorem{corollary}{Corollary}
\let\oldproofname=\proofname
\renewcommand{\proofname}{\rm\bf{\oldproofname}}
\allowdisplaybreaks
\emergencystretch=\maxdimen

\title{Identification and Estimation of Nonseparable Triangular Equations with Mismeasured Instruments}
\author{Shaomin Wu}

\begin{document}
\abovedisplayskip=15pt
\belowdisplayskip=15pt
\maketitle
\centerline{\textbf{Abstract}}
In this paper, I study the nonparametric identification and estimation of the marginal effect of an endogenous variable $X$ on the outcome variable $Y$, given a potentially mismeasured instrument variable $W^*$, without assuming linearity or separability of the functions governing the relationship between observables and unobservables. To address the challenges arising from the co-existence of measurement error and nonseparability, I first employ the deconvolution technique from the measurement error literature to identify the joint distribution of $Y, X, W^*$ using two error-laden measurements of $W^*$. I then recover the structural derivative of the function of interest and the ``Local Average Response'' (LAR) from the joint distribution via the ``unobserved instrument'' approach in \cite{matzkin2016independence}. I also propose nonparametric estimators for these parameters and derive their uniform rates of convergence. Monte Carlo exercises show evidence that the estimators I propose have good finite sample performance. 

\section{Introduction}
This paper studies the nonparametric identification and estimation of the marginal effect of an endogenous variable $X$ on the outcome variable $Y$, given a potentially mismeasured instrument variable $W^*$, without assuming linearity or separability of the functions governing the relationship between observables and unobservables. Without measurement error, this type of model is referred to as ``nonseparable triangular equations''
and its identification was studied by \cite{chesher2003identification}, \cite{imbens2009identification}, and \cite{shaikh2011partial}. Measurement error on the instrument variable poses additional challenges to the identification and estimation of the model primitives. Because of nonseparability, simply using an error-laden measurement as the instrument leads to inconsistent results. Because of measurement error, the true value of $W^*$ is unobserved, making it impossible to proceed using existing methods in \cite{imbens2009identification}. In this paper, I propose a way to deal with the two difficulties mentioned above and show the identification of average and individual-level marginal effects. I also propose estimators for these parameters. To illustrate ideas, I study the following model: 
 \begin{align*}
    Y &= m(X, \epsilon)\\
    X &= h(W^*,\eta),
\end{align*}
where $X$ is endogenous in the sense that it's correlated with $\epsilon$, and $W^*$ is the instrument variable independent with both of the error terms $\epsilon, \eta$ but cannot be measured accurately. Following the measurement error literature, I assume that there are two error-laden measurements of $W^*$, denoted as $W_1, W_2$. Denote the corresponding measurement errors as $\Delta W_1, \Delta W_2$, i.e. $W^*=W_1+\Delta W_1$, $W^*=W_2+\Delta W_2$. I could allow a vector of observed exogenous control variables $Z$ to enter both equations and all the arguments will carry over by conditioning on $Z$, so they are omitted here for simplicity. \\
\indent To give an example of where this model can be used, consider the Engel curve estimation studied in \cite{blundell2007semi}. They used Sieve Minimum Distance to estimate a semi-nonparametric model, but the curve can be estimated under less restrictive modeling assumptions as done in \cite{imbens2009identification}. $Y$ here is the share of expenditure on a commodity of a household. $X$ is the log of the household’s total expenditure. $Z$ is a vector capturing the household's demographic composition and $\epsilon$ captures the household's unobserved heterogeneity. Researchers might be interested in the average response of households' expenditure on food $Y$, to changes in household's total expenditure $X$, holding the distribution of unobserved household heterogeneity $F_{\epsilon\mid X=x}$ fixed. This parameter is called ``Local Average Response'' (LAR) by \cite{altonji2005cross} and if the function $m$ is differentiable, it can be written as
\begin{align}
    \mathbf{LAR}(x) = \int \frac{\partial m(x, e)}{\partial x}f_{\epsilon\mid X=x}(e)de. \label{introLAR}
\end{align}
In addition, researchers might also be interested in the structural derivative of the function $m$. This is a more disaggregate level parameter. It stands for the marginal response of $Y$ to changes in $X$ for a specific household, say the household with log total expenditure equal $\bar x$, and share of food expenditure equal $\bar y$. Denote the structural derivative of this household as $\rho(\bar y, \bar x)$. If $m$ is differentiable and strictly monotone in its second argument for all values of $X$, $\rho(\bar y, \bar x)$ can be written as 
\begin{align}
    \rho(\bar y, \bar x) = \left.\frac{\partial m(x, \epsilon)}{\partial x}\right|_{x=\bar x, \epsilon = m^{-1}(\bar y, \bar x)}, \label{introdmdx}
\end{align}
where $m^{-1}$ is the inverse of $m$ with respect to its second argument, and $m^{-1}(\bar y, \bar x)$ is the value of $\epsilon$ of the specific household one is interested in. To estimate these parameters, under the assumption that heterogeneity in earnings is not correlated with households’ preferences over consumption, one can use the income of the head of the household as the IV $W^*$. The income variable is likely to suffer from measurement errors. Researchers could obtain multiple measurements of it  from panel data, for example.\\
\indent To see how nonseparability makes it harder to identify parameters like the LAR and the structural derivative when the IV $W^*$ is mismeasured, consider the case when both equations are linear:
 \begin{align*}
    Y &= \alpha_0+\alpha_1 X+\epsilon\\
    X &= \beta_0+\beta_1 W^*+\eta,
\end{align*}
where $E[\epsilon W^*]=0$ (exclusion restriction) and $E[XW^*]\neq E[X]E[W^*]$ (relevance condition). Suppose I have $W_2=W^*+\Delta W_2$ as an error-laden measurement of $W^*$, it's not hard to verify that $E[\epsilon W_2]=0$ and $X\not\perp W_2 $ still hold under mild assumptions on $\Delta W_2$ (e.g. $E\left[\Delta W_2|Y,X\right]=0$). This means in a linear model, even if the instrument variable suffers from measurement errors, one can still use the error-laden measurement $W_2$ as an IV and proceed as usual. However, this is not the case in nonseparable models. Plugging $W_2$ into the second equation yields $X=h(W_2-\Delta W_2, \eta)$, where both $\Delta W_2$ and $\eta$ are unobervable and importantly, $W_2\not\perp\Delta W_2$. This means if I were to use $W_2$ as an IV, both of the two equations in the triangular system would contain endogenous variables and nothing can be done without additional IVs outside of this system. This also means that if one simply uses $W_2$ as the instrument variable and proceeds with the standard techniques in \cite{imbens2009identification}, they would get inconsistent results.\\
\indent Given that nonseparability makes the problem much harder to solve, a natural question is why one wants to deal with nonseparable models instead of an additive separable or linear model. There are multiple reasons why researchers might prefer a nonseparable model. First, nonseparable models allow the observed variable $X$ and the unobserved variable $\epsilon$ to interact in a flexible way. For example, a recent paper by \cite{brancaccio2020geography} estimated the matching function $m(s,e)$ between ships (of number $s$) and exporters (of number $e$ which is unobserved) at a seaport. Not imposing functional form assumptions (including separability) is important. It allows the authors to remain agnostic about the nature of the meeting process. In addition, the flexibility of functional forms can be key when deriving welfare and policy implications (see \cite{brancaccio2020geography}). Second, nonseparability is also important when $X$ and $\epsilon$ are correlated, since in many cases the source of endogeneity is the nonseparable nature of the model. For example, $X$ could come from the optimization problem of maximizing the expected value of $Y$ minus the production cost, given an exogenous variable $W$, and some noisy information about $\epsilon$. Think of $X$ as an individual's education level, or a firm's input level, and $Y$ as the individual's lifetime earnings, or the firm's output. Then $X$ could be the solution of $\max_{x}\left\{E[m(x,\epsilon)|\eta,W]-c(x,W)\right\}$, leading to $X=h(W,\eta)$, where $\eta$ is some noisy proxy of $\epsilon$, and $c(x,W)$ is the cost function. If the function $m$ were additively separable in $\epsilon$, the optimal choice of $x$ would not even depend on $\eta$.\\
\indent When there is no measurement error, \cite{imbens2009identification} proposed a way to identify the model primitives in a nonseparable triangular system, making use of the fact that $X\perp\epsilon|\eta$. They first estimate the control variable $\eta$ (or a strictly monotone function of $\eta$, denoted as $V$ in their paper) from the second equation, and then estimate the model primitives in the first equation by first conditioning on the estimated $\eta$, and then integrate it out. Their method cannot directly apply when the instrument variable $W^*$ in the second equation is mismeasured, because the control variable (or control function) cannot be observed from error-laden measurements of $W^*$. A recent paper by \cite{aradillas2022inference} studies inference in models where control functions are unobserved. The setup of his paper is different from this paper in many aspects. He requires the availability of observable or estimable bounds for the unobserved control functions, which is not required by the model in this paper. Instead, this paper requires the availability of error-laden measurements and builds upon the measurement error literature. Also, his focus is on constructing confidence sets for finite-dimensional parameters, while my focus is on point identification and estimation of infinite-dimensional parameters.\\
\indent In this paper, I propose a method that makes use of the same intuition as in \cite{imbens2009identification}, but can deal with mismeasured $W^*$. Same as \cite{imbens2009identification}, I utilize the fact that the correlation between $X$ and $\epsilon$ is merely coming from $\eta$, but instead of estimating and conditioning on $\eta$, I separate out this correlation by writing $\epsilon$ as a function of $\eta$ and a uniformly distributed random variable which is independent with $X$ and $W^*$. In this way, I am able to write the model primitives like the structural derivative and LAR as functionals of the joint distribution of $Y,X,W^*$, which can be recovered using the two error-laden measurements and the deconvolution technique developed in the measurement error literature (e.g. \cite{fan1991asymptotic},  \cite{fan1993nonparametric}, \cite{schennach2004estimation}, \cite{schennach2004nonparametric}). I can thus identify the model primitives in a constructive way and estimate them using plug-in estimators. I also derive uniform rates of convergence of the estimators.\\
\indent This paper is most related to \cite{schennach2012local} (SWC, hereafter), where the authors also consider a triangular simultaneous equations model with a mismeasured exogenous instrument. This paper differs from their paper in the modeling assumptions, parameters that can be identified, and also theoretical methods. SWC shows that under separability assumption on the second equation, the instrument conditioned marginal response\footnote{$m_x$ denotes the partial derivative of the $m$ function with respect to its first argument.} $E\left[m_x(X, \epsilon)|W^*=w^*\right]$ can be written as a ratio of the derivative of the conditional mean of $Y$ given $W^*$ over the derivative of the conditional mean of $X$ given $W^*$, both of which can be recovered from the data given two error-laden measurements $W_1, W_2$. Integrating out $W^*$, they can also recover the average response $E\left[m_x(X, \epsilon)\right]$. However, under the nonseparability of both of the equations, their method cannot recover either of the two parameters. Different from SWC, this paper shows that it's possible to identify not only the instrument-conditioned marginal response and the average marginal response but also the structural derivative and LAR, even when both of the equations are nonseparable. This conclusion, however, comes at the expense of more assumptions on the function $m$ and unobservables compared with SWC. In particular, this paper assumes strict monotonicity of the function $m$ on its second argument, and scalar unobservables, which are not required in SWC. Regarding estimation, both SWC and this paper employ plug-in estimators, but the asymptotic analysis in this paper is a bit more complex than in SWC, because of the observed variable components $Y,X$ of the joint density $f_{Y,X,W^*}$. They have non-trivial implications on the asymptotic treatment (including the convergence rates) so that SWC's asymptotic analysis cannot directly apply here. \\
\indent This paper is also closely related to \cite{song2015estimating}. In their paper, the coauthors study a nonseparable model with mismeasured endogenous variable, assuming a correctly measured control variable is available. The setup of this paper is different from theirs. This paper also studies nonseparable models with endogeneity, but instead of mismeasured endogenous variable, this paper studies the case when the endogenous variable is correctly measured, and instead of assuming the existence of a correctly measured control variable, this paper assumes that a potentially mismeasured instrument variable is available. The two papers are complementary depending on the availability of data. Regarding the parameters of interest, in addition to the various parameters studied in \cite{song2015estimating}, including the average marginal response conditional on the control variable, the average marginal response, the LAR, and the weighted average version of these variables, this paper additionally show identification of the structural derivative, which is more disaggregate and could be helpful when researchers are interested heterogeneous effects.\\
\indent The rest of the paper consists of six parts. Section 2 introduces the model and assumptions. Section 3 talks about the identification of the model primitives. Section 4 proposes plug-in estimators of the model primitives identified in Section 3, and talks about their asymptotic properties. Section 5 conducts limited Monte Carlo studies to show the finite sample performance of the estimator. Section 6 concludes.
\section{The Model}
I consider the following triangular model in this paper:
\begin{align*}
    Y &= m(X,Z,\epsilon)\\
    X &= h(W^*,Z,\eta)
\end{align*}
where $X$ is an observed endogenous variable that is correlated with $\epsilon$. In the returns to education example, $X$ stands for years of education, which is correlated with $\epsilon$ since it's chosen by the agent as an equilibrium outcome. Z is a vector of observed exogenous variables which are independent with $\epsilon$ and $\eta$.  $W^*$ is an instrument for $X$ and satisfies $W^*\perp (\epsilon, \eta)$. Researchers cannot measure $W^*$ exactly but have two error-laden measurements of it:
 \begin{align*}
        W_1 &= W^*+\Delta W_1\\
        W_2 &= W^*+\Delta W_2.
\end{align*}
The measurement errors satisfy $E[\Delta W_1\mid W^*, \Delta W_2]=0$ and $\Delta W_2\perp W^*, Y, X, Z$. As in the literature studying the identification of nonseparable models (\cite{chesher2003identification}, \cite{matzkin2003nonparametric},\cite{altonji2005cross},\cite{matzkin2015estimation}), I impose monotonicity assumptions on the structural functions. I assume that $m$ is strictly increasing in $\epsilon$ and that $h$ is strictly increasing in $\eta$. Since the vector of covariates $Z$ is correctly measured and exogenous, all the analysis can be done conditional on $Z$. For brevity, from now on I omit the vector $Z$ in my notations and work on the simplified model below, while other assumptions on the structural functions and distributions of variables remain unchanged.
\begin{align}
    Y &= m(X,\epsilon)\\
    X &= h(W^*,\eta)
\end{align}
The assumptions mentioned above are stated formally below:
\begin{assumption} \label{aspm&h}
Function $m$ and $h$ are continuously differentiable with respect to both of their arguments and are strictly increasing in their respective second argument.  
\end{assumption}
\begin{assumption}
$W^*\perp (\epsilon,\eta)$.
\end{assumption}
\begin{assumption}\label{repeatedmeas}
$W_1 = W^*+\Delta W_1$ and $W_2 = W^*+\Delta W_2$.
\end{assumption}
\begin{assumption}\label{measerrors}
 $E[\Delta W_1\mid W^*, \Delta W_2]=0$, $\Delta W_2\perp W^*, Y, X$.
\end{assumption}
Following the measurement error literature, I also impose:
\begin{assumption}\label{posdenom}
For any finite $t\in\mathbb{R}$, $\left|E[\exp\left(\mathbf{i}tW_2\right)]\right|>0$. 
\end{assumption}
\indent Note that for the measurement error $\Delta W_1$, only the mean independence assumption is imposed, which is weaker than the assumption on $\Delta W_2$. This weaker assumption is sufficient to identify the distribution of $W^*$ (\cite{schennach2004nonparametric}). On the other hand, although $\Delta W_2$ satisfies strong independence assumptions, making $W_2$ a measurement with classical measurement error, one still cannot use $W_2$ as the instrument because of the nonseparability of the model. More specifically, plugging $W_2$ into the second equation in the triangular system yields
\begin{align*}
    X=h(W_2-\Delta W_2, \eta).
\end{align*}
This is a nonseparable model with two unobservables $\Delta W_2$ and $\eta$, and $W_2$ is not independent with $\Delta W_2$. This means both of the two equations in the system contain endogenous variables and thus cannot be identified without further assumptions.\\
\indent In addition, I also impose the following assumptions on the distribution of $Y,X,W^*$:
\begin{assumption}\label{cpctcontdiff} The distribution of $(Y,X,W^*)$ has compact support, denoted as $\mathbb{S}_{(Y,X,W^*)}$ and has a joint density $f_{Y,X,W^*}$ which is twice continuously differentiable on $R^3$. 
\end{assumption}
\begin{assumption}\label{pospartialh}
For any $x,w^*$ belonging to the support of $X,W^*$, $\left|\frac{\partial F_{X\mid W^*=w^*}(x)}{\partial w^*}\right|>0$.
\end{assumption}
Assumption \ref{pospartialh} is a sufficient condition to ensure $\left.\frac{\partial h(w^*,\eta)}{\partial w^*}\right|_{\eta = r(x,w^*)}\neq 0$ for all $x,w^*$ on the support of $X,W^*$. This assumption is like the rank condition imposed in the linear instrument variable models. To illustrate this idea, suppose $h$ is a linear function such that $X= \pi W^*+\eta$, then $F_{X\mid W^*=w^*}(x) = F_{\eta}(x-\pi w^*)$, and $\pi\neq 0$ is a neccessary condition for Assumption \ref{pospartialh}.
\section{Identification}
In this section, I show the identification of the model parameters in two steps. First, assuming that the joint distribution of $Y, X, W^*$ is known, I show identification of the derivative of the structural function $m$ with respect to $x$, when the value of $\epsilon$ is fixed, and the identification of other model parameters including the LAR and AR. Then I show how to identify the joint distribution of $Y, X, W^*$ from the two measurements $W_1$ and $W_2$.
\subsection{Identification of parameters assuming the joint distribution of $Y, X, W^*$ is known}
\indent The independence between $W^*$ and $\epsilon,\eta$ implies that conditional on $\eta$, $W^*$ and $\epsilon$ are independent. The next lemma uses this fact to show that one can write $\epsilon$ as a function of $\eta$ and another random variable which is independent with $X, W^*$. This lemma is similar to Proposition 5.1 in \cite{matzkin2016independence}. Before stating the lemma, I first impose the following assumption on the conditional distribution of $\epsilon$ given $\eta$:
\begin{assumption}\label{Fepsilon|eta}
$F_{\epsilon\mid \eta=\bar\eta}(\bar\epsilon)$ is strictly increasing in $\bar\epsilon$, given any values of $\bar\eta$.
\end{assumption}
\begin{lemma}\label{deltaindependent}
Under the model setup, suppose Assumption \ref{Fepsilon|eta} is satisfied. There exists a function $s:\mathbb{R}^2\rightarrow\mathbb{R}$ strictly increasing in its second argument and an unobservable random term $\delta$ such that 
\begin{align}
    \epsilon = s(\eta, \delta), \label{epsilondecomposition}
\end{align}
and $\delta$ is independent of $(X,W^*)$ and is $U(0,1)$.
\end{lemma}
\begin{proof}
See Appendix B.
\end{proof}
Next, I plug (\ref{epsilondecomposition}) into the first structural equation. The assumption that $h$ is strictly increasing in $\eta$ implies that one can write the inverse of $h(W^*,\eta)$ w.r.t. $\eta$ as $r(X, W^*)$. Then I have
\begin{align}\label{vfunction}
    Y &= m(X,\epsilon)=m(X,s(\eta,\delta))=m(X,s(r(X,W^*),\delta))\equiv v(X,W^*,\delta).
\end{align}
Equation (\ref{vfunction}) builds a bridge between the structural function $m$ and the reduced form function $v$. Note that $v$ is a function of observable variables $X,W^*$ and an unobservable variable $\delta$ which is independent of the observables. The derivatives of $v$ can be identified by applying identification techniques in standard nonseparable models (\cite{matzkin2003nonparametric}). To identify the main parameter of interest, the derivative of the structural function $m$, one can utilize the last equality in \ref{vfunction}: $m(X,s(r(X,W^*),\delta))\equiv v(X,W^*,\delta)$. Taking derivative w.r.t $X$ and $W^*$ yields:
\begin{footnotesize}
  \begin{align}
    &\left.\frac{\partial m(x,\epsilon)}{\partial x}\right|_{\epsilon =s(r(x, w^*),\delta)} + \left.\frac{\partial m(x,\epsilon)}{\partial \epsilon}\right|_{\epsilon =s(r(x, w^*),\delta)}\cdot\left.\frac{\partial s(\eta, \delta)}{\partial \eta}\right|_{\eta = r(x,w^*)}\cdot \frac{\partial r(x,w^*)}{\partial x} = \frac{\partial v(x,w^*,\delta)}{\partial x}\label{pvpx}\\
    &\left.\frac{\partial m(x,\epsilon)}{\partial \epsilon}\right|_{\epsilon =s(r(x, w^*),\delta)}\cdot\left.\frac{\partial s(\eta, \delta)}{\partial \eta}\right|_{\eta = r(x,w^*)}\cdot \frac{\partial r(x,w^*)}{\partial w^*} = \frac{\partial v(x,w^*,\delta)}{\partial w^*} \label{pvpw^*}
 \end{align}  
\end{footnotesize}
Plug (\ref{pvpw^*}) into (\ref{pvpx}) to cancel $\left.\frac{\partial m(x,\epsilon)}{\partial \epsilon}\right|_{\epsilon =s(r(x, w^*),\delta)}\cdot\left.\frac{\partial s(\eta, \delta)}{\partial \eta}\right|_{\eta = r(x,w^*)}$ yields
\begin{align*}
    \left.\frac{\partial m(x,\epsilon)}{\partial x}\right|_{\epsilon =s(r(x, w^*),\delta)}=\frac{\partial v(x, w^*, \delta)}{\partial x}-\frac{\partial v(x, w^*, \delta)}{\partial w^*}\frac{\frac{\partial r(x,w^*)}{\partial x}}{\frac{\partial r(x,w^*)}{\partial w^*}}.
\end{align*}
Taking derivative w.r.t. $W^*$ on both sides of $X\equiv h(W^*,r(X,W^*))$ and cancelling out the unobserved $\left.\frac{\partial h(w^*,\eta)}{\partial \eta}\right|_{\eta=r(x,w^*)}$, one can get $\left.\frac{\partial h(w^*,\eta)}{\partial w^*}\right|_{\eta=r(x,w^*)} = -\frac{\frac{\partial r(x,w^*)}{\partial w^*}}{\frac{\partial r(x,w^*)}{\partial x}}$. Then one can write 
\begin{align}
    \left.\frac{\partial m(x,\epsilon)}{\partial x}\right|_{\epsilon =s(r(x, w^*),\delta)}=\frac{\partial v(x, w^*, \delta)}{\partial x}+\frac{\partial v(x, w^*, \delta)}{\partial w^*}\frac{1}{\left.\frac{\partial h(w^*,\eta)}{\partial w^*}\right|_{\eta=r(x,w^*)}}.\label{pmpx}
\end{align}
The right-hand side of the equation (\ref{pmpx}) can be identified from the data. To show this, note that 
\begin{align}
    &F_{Y\mid X=x, W^*=w^*}(v(x,w^*,\delta))=F_{\delta}(\delta)=\delta\label{F_Y|XW^*}\\
    \Rightarrow &\frac{\partial F_{Y\mid X=x, W^*=w^*}(y)}{\partial x}+\frac{\partial F_{Y\mid X=x, W^*=w^*}(y)}{\partial y}\frac{\partial v(x,w^*,\delta)}{\partial x} =0 \notag\\
    \Rightarrow &\left.\frac{\partial v(x,w^*,\delta)}{\partial x} = -\frac{\partial F_{Y\mid X=x, W^*=w^*}(y)}{\partial x}\right/\frac{\partial F_{Y\mid X=x, W^*=w^*}(y)}{\partial y},\notag
\end{align}
Similarly
\begin{align*}
    &\frac{\partial v(x,w^*,\delta)}{\partial w^*} = -\left.\frac{\partial F_{Y\mid X=x, W^*=w^*}(y)}{\partial w^*}\right/\frac{\partial F_{Y\mid X=x, W^*=w^*}(y)}{\partial y}\\
    &\frac{\partial h(w^*,\eta)}{\partial w^*} = -\left.\frac{\partial F_{X\mid W^*=w^*}(x)}{\partial w^*}\right/\frac{\partial F_{X\mid  W^*=w^*}(x)}{\partial x},
\end{align*}
where $\delta$ is the unique value such that $y=v(x,w^*,\delta)$ and $\eta$ is the unique value such that $ x=h(w^*,\eta) $. Plug into (\ref{pmpx}), one can get
\begin{align}\label{pmpxbar}
    \rho(\bar y, \bar x)&\equiv\left.\frac{\partial m(x, \epsilon)}{\partial x} \right|_{x=\bar x, \epsilon=m^{-1}(\bar x, \bar y)}\notag\\
    &= -\frac{\frac{\partial F_{Y\mid X=\bar x, W^*=w^*}(\bar y)}{\partial x}}{\frac{\partial F_{Y\mid X=\bar x, W^*=w^*}(\bar y)}{\partial y}}+\frac{\frac{\partial F_{Y\mid X=\bar x, W^*=w^*}(\bar y)}{\partial w^*}}{\frac{\partial F_{Y\mid X=\bar x, W^*=w^*}(\bar y)}{\partial y}}\times \frac{\frac{\partial F_{X\mid W^*=w^*}(\bar x)}{\partial x}}{\frac{\partial F_{X\mid W^*=w^*}(\bar x)}{\partial w^*}}
\end{align}
for all $\bar y, \bar x$ belongs to support, and values of $\bar\epsilon$ such that $\bar y = m(\bar x, \bar\epsilon)$. Note that as long as $\bar y$ and $\bar x$ are fixed, $\bar\epsilon$ is fixed, no matter which value is picked for $w^*$. This means there's actually overidentification for the structural derivative. To identify the \textbf{LAR} and \textbf{AR}, one needs independent variations of $\epsilon$ given $X$. Rewrite (\ref{F_Y|XW^*}) with slightly different notations yields:
\begin{align*}
    &F_{Y\mid X=x, W^*=w^*}(v(x,w^*,\delta))=F_{\delta}(\delta)=\delta\\
    \Rightarrow & v(x,w^*,\delta)= F_{Y\mid X=x, W^*=w^*}^{-1}(\delta) 
\end{align*}
so that one can write
\begin{align}\label{pmpxdelta}
     \left.\frac{\partial m(x,\epsilon)}{\partial x}\right|_{\epsilon =s(r(x, w^*),\delta)}=\frac{\partial F^{-1}_{Y\mid X=x, W^*=w^*}(\delta)}{\partial x}-\frac{\partial F^{-1}_{Y\mid X=x, W^*=w^*}(\delta)}{\partial w^*}\frac{\frac{\partial F_{X\mid W^*=w^*}(x)}{\partial x}}{\frac{\partial F_{X\mid W^*=w^*}(x)}{\partial w^*}}.
\end{align}
To show the identification of the \textbf{LAR} and \textbf{AR}, one needs to show that the support of $s(r(X, W^*), \delta)$ given $X=x$ is the same as the support of $\epsilon$ given $X=x$. This is stated formally in the lemma below. 
\begin{lemma}\label{commonsupport}
Define random variable $\mathbf{s}\equiv s(r(X, W^*), \delta)$. Denote the support of $\mathbf{s}$ and $\epsilon$ conditional on $X=x$ as $\mathbb{S}_{\mathbf{s}\mid X=x}$ and $\mathbb{S}_{\epsilon\mid X=x}$, respectively. If Assumption \ref{Fepsilon|eta} holds, then $\mathbb{S}_{\mathbf{s}\mid X=x} = \mathbb{S}_{\epsilon\mid X=x}$, for all $x$ belonging to its support.
\end{lemma}
\begin{proof}
See Appendix B.
\end{proof}
Lemma \ref{deltaindependent} and \ref{commonsupport} ensures that $F_{\mathbf{s}\mid X}$ is the same as $F_{\epsilon\mid X}$, so that integrating $\frac{\partial m(x,\epsilon)}{\partial x}$ over $\epsilon$ given $X=x$ will be equivalent to integrating $\frac{\partial m(x,\mathbf{s})}{\partial x}$ over $\mathbf{s}$ given $X=x$. Then I have the following identification results:
\begin{align}
    &\textbf{LAR}:\  E\left[\frac{\partial m(X,\epsilon)}{\partial x}\mid X=x\right]=\int \int_0^1 \left.\frac{\partial m(x,\epsilon)}{\partial x}\right|_{\epsilon =s(r(x, w^*),\delta)} f_{W^*|X=x}(w^*)d\delta dw^* \label{LAR}\\
    &\textbf{AR}:\  E\left[\frac{\partial m(X,\epsilon)}{\partial x}\right]=\int\int \int_0^1 \left.\frac{\partial m(x,\epsilon)}{\partial x}\right|_{\epsilon =s(r(x, w^*),\delta)} f_{X,W^*}(x,w^*)d\delta dw^* dx. \label{AR}
\end{align}
\subsection{Identification of the joint density of $Y, X, W^*$}
\indent First note that by Theorem 1 in \cite{schennach2004estimation}
\begin{align*}
    \phi_{W^*}(t)=\exp\left(\int_0^{t}\frac{E\left[iW_1 e^{i\xi W_2}\right]}{E\left[e^{i\xi W_2}\right]}d\xi \right).
\end{align*}
Then I can identify the density $f_{W^*}(w^*)$ by taking the Inverse Fourier Transform:
\begin{align*}
    f_{W^*}(w^*) =\frac{1}{2\pi}\int e^{-itw^*}\phi_{W^*}(t)dt.
\end{align*}
For the joint density $f_{Y,X,W^*}(y, x, w^*)$, by Assumption \ref{measerrors}, I have the following convolution:
\begin{align*}
    f_{Y,X,W_2}(x,w_2) = \int f_{Y,X,W^*}(x,\nu) f_{\Delta W_2}(w_2-\nu)d\nu
\end{align*}
Applying Fourier transformation on both sides (holding $x$ constant) yields
\begin{align*}
    \phi_{f_{Y,X,W_2}(x,\cdot )}(t)&=\phi_{f_{Y,X,W^*}(y,x,\cdot)}(t) \phi_{\Delta W_2}(t),
\end{align*}
which, by Assumption \ref{posdenom} implies
\begin{align}
    \phi_{f_{Y,X,W^*}(y,x,\cdot)}(t) &= \frac{\phi_{f_{Y,X,W_2}(y,x,\cdot)}(t)}{\phi_{\Delta W_2}(t)}=\frac{\phi_{f_{Y,X,W_2}(y,x,\cdot)}(t)\phi_{W^*}(t)}{\phi_{W_2}(t)}.\label{phifxstary}
\end{align}
Then applying the inverse Fourier transform, one can get
\begin{align}
    f_{Y,X,W^*}(x,w^*) &= \frac{1}{2\pi}\int e^{-itw^*}\frac{\phi_{f_{X,W_2}(y,x,\cdot)}(t)\phi_{W^*}(t)}{\phi_{W_2}(t)} dt, \label{fXW*},
\end{align}
and similarly, 
\begin{align}
     f_{Y,X, W^*}(y, x, w^*) &= \frac{1}{2\pi}\int e^{-itw^*}\frac{\phi_{f_{Y,X,W_2}(y, x, \cdot)}(t)\phi_{W^*}(t)}{\phi_{W_2}(t)} dt. \label{fYXW*}
\end{align}
Then $F_{X\mid W^*=w^*}(x)$ follows from
\begin{align}
    F_{X\mid W^*=w^*}(x) = \frac{\int_{-\infty}^x f_{X,W^*}(y,x,w^*)dx}{\int_{-\infty}^\infty f_{X,W^*}(y,x,w^*)dx}\label{FX|W*}
\end{align}
and $F_{Y\mid X=x, W^*=w^*}(y)$ follows from
\begin{align*}
    F_{Y\mid X=x, W^*=w^*}(y) = \frac{\int_{-\infty}^y f_{Y, X,W^*}(y, x,w^*)dy}{\int_{-\infty}^\infty f_{Y, X,W^*}(y,x,w^*)dy}.
\end{align*}
\indent From the equations above, one can write the derivatives of the conditional CDFs as functionals of the joint densities:
\begin{footnotesize}
\begin{align*}
    &\frac{\partial F_{Y\mid X=x, W^*=w^*}(y)}{\partial w^*} \\
    &= \frac{\left(\int_{-\infty}^y \frac{\partial f_{Y,X,W^*}(s,x,w^*)}{\partial w^*}ds\right)f_{X,W^*}(x,w^*)-\left(\int_{-\infty}^y f_{Y,X,W^*}(s,x,w^*)ds\right)\frac{\partial f_{X,W^*}(x,w^*)}{\partial w^*}}{f^2_{X,W^*}(x,w^*)}\\
    &\frac{\partial F_{Y\mid X=x, W^*=w^*}(y)}{\partial x}\\
    &=\frac{\left(\int_{-\infty}^y \frac{\partial f_{Y,X,W^*}(s,x,w^*)}{\partial x}ds\right)f_{X,W^*}(x,w^*)-\left(\int_{-\infty}^y f_{Y,X,W^*}(s,x,w^*)ds\right)\frac{\partial f_{X,W^*}(x,w^*)}{\partial x}}{f^2_{X,W^*}(x,w^*)}\\
    &\frac{\partial F_{Y\mid X=x, W^*=w^*}(y)}{\partial y}=\frac{f_{Y,X,W^*}(y,x,w^*)}{f_{X,W^*}(x,w^*)}\\
    &\frac{\partial F_{X\mid W^*=w^*}(x)}{\partial w^*}=\frac{\left(\int_{-\infty}^{x} \frac{\partial f_{X,W^*}(s,w^*)}{\partial w^*}ds\right) f_{W^*}(w^*)-\left(\int_{-\infty}^{x} f_{X,W^*}(s,w^*)ds\right)\frac{\partial  f_{W^*}(w^*)}{\partial w^*}}{f^2_{W^*}(w^*)}\\
    &\frac{\partial F_{X\mid W^*=w^*}(x)}{\partial x}=\frac{f_{X,W^*}(x,w^*)}{ f_{W^*}(w^*)}.
\end{align*}
\end{footnotesize}
This means that the right-hand side of the equation (\ref{pmpxbar}) can be written as a functional of the joint density $f_{Y,X, W^*}$, which can be identified from the data. To show identification of the left-hand side of the equation (\ref{pmpxdelta}), one needs to build the relationship between the derivatives of conditional quantiles and the joint densities, which is shown in the following lemma: 
\begin{lemma}\label{dcquantile}Let $\delta\in[0,1]$ be a constant, then
\begin{align*}
    &\frac{\partial F^{-1}_{Y|X=x, W^*=w^*}(\delta)}{\partial w^*} = \frac{\delta \frac{\partial f_{X,W^*}(x,w^*)}{\partial w^*}-\int_{-\infty}^{F^{-1}_{Y|X=x, W^*=w^*}(\delta)}\frac{\partial f_{Y,X,W^*}(y,x,w^*)}{\partial w^*}dy}{f_{Y,X,W^*}\left(F^{-1}_{Y|X=x, W^*=w^*}(\delta),x,w^*\right)}\\
    &\frac{\partial F^{-1}_{Y|X=x, W^*=w^*}(\delta)}{\partial x} = \frac{\delta \frac{\partial f_{X,W^*}(x,w^*)}{\partial x}-\int_{-\infty}^{F^{-1}_{Y|X=x, W^*=w^*}(\delta)}\frac{\partial f_{Y,X,W^*}(y,x,w^*)}{\partial x}dy}{f_{Y,X,W^*}\left(F^{-1}_{Y|X=x, W^*=w^*}(\delta),x,w^*\right)}.
\end{align*}
\end{lemma}
\begin{proof}
see Appendix B.
\end{proof}
The identification of the left-hand side of (\ref{LAR}) and (\ref{AR}) is a straightforward implication of Lemma \ref{dcquantile} and equation (\ref{pmpxdelta}).

\section{Estimation}
\subsection{The Estimator} \label{The estimator}
\indent From now on, I use the following function to denote the joint density of $Y, X,W^*$, or its derivative with respect to $Y, X,W^*$:
\begin{align*}
    g_{\lambda_1,\lambda_{2,1},\lambda_{2,2}}(y,x,w^*) = \frac{\partial^\lambda f_{Y,X,W^*}(y,x,w^*)}{\partial w^{*\lambda_1}\partial y^{\lambda_{2,1}}\partial x^{\lambda_{2,2}}},
\end{align*}
where $\lambda_1, \lambda_{2,1}, \lambda_{2,2}\in\{0,1\}$ and $\lambda\equiv\max\{\lambda_1, \lambda_{2,1},\lambda_{2,2}\}\leq 1$. The $0-th$ order derivative of a function is defined as the function itself. For convenience of notation, I also define $\lambda_2\equiv\max\{ \lambda_{2,1},\lambda_{2,2}\}$. In this section, I first propose an estimator for $g_{\lambda_1,\lambda_{2,1}, \lambda_{2,2}}(y,x,w^*)$, and then propose plug-in estimators for the structural derivative and a weighted average version of the \textbf{LAR} since they can be written as known functionals of $g_{\lambda_1,\lambda_{2,1}, \lambda_{2,2}}(y,x,w^*)$. \\
\indent To make the expressions more transparent and clear, I show the explicit forms of the functionals by taking  the structural derivative as an example. I write each component on the right-hand-side of equation \ref{pmpxbar} as 
\begin{align*}
    &\quad\frac{\partial F_{Y\mid X=x, W^*=w^*}(y)}{\partial w^*} \\
    &= \frac{\left(\int_{-\infty}^y  g_{1,0,0}(s,x,w^*)ds\right) }{\left(\int_{-\infty}^\infty g_{0,0,0}(y,x,w^*)dy\right)}-\frac{\left(\int_{-\infty}^y g_{0,0,0}(s,x,w^*)ds\right)\left(\int_{-\infty}^\infty g_{1,0,0}(y,x,w^*)dy\right)}{\left(\int_{-\infty}^\infty g_{0,0,0}(y,x,w^*)dy\right)^2}\\
    &\quad\text{and}\\
    &\quad\frac{\partial F_{Y\mid X=x, W^*=w^*}(y)}{\partial x}\\
    &=\frac{\left(\int_{-\infty}^y  g_{0,0,1}(s,x,w^*)ds\right) }{\left(\int_{-\infty}^\infty g_{0,0,0}(y,x,w^*)dy\right)}-\frac{\left(\int_{-\infty}^y g_{0,0,0}(s,x,w^*)ds\right)\left(\int_{-\infty}^\infty g_{0,0,1}(y,x,w^*)dy\right)}{\left(\int_{-\infty}^\infty g_{0,0,0}(y,x,w^*)dy\right)^2}\\
    &\quad\text{and}\\
    &\quad\frac{\partial F_{Y\mid X=x, W^*=w^*}(y)}{\partial y}=\frac{ g_{0,0,0}(y,x,w^*)}{\left(\int_{-\infty}^\infty g_{0,0,0}(y,x,w^*)dy\right)}\\
    &\quad\text{and}\\
    &\quad\frac{\partial F_{X\mid W^*=w^*}(x)}{\partial w^*}\\
    &=\frac{\left(\int_{-\infty}^{x}\int_{-\infty}^\infty g_{1,0,0}(y,s,w^*)dyds\right)}{\left(\int_{-\infty}^\infty\int_{-\infty}^\infty g_{0,0,0}(y,x,w^*)dydx\right)}\\
    &\quad-\frac{\left(\int_{-\infty}^{x}\int_{-\infty}^\infty g_{0,0,0}(y,s,w^*)dyds\right)\left(\int_{-\infty}^\infty\int_{-\infty}^\infty g_{1,0,0}(y,x,w^*)dydx\right)}{\left(\int_{-\infty}^\infty\int_{-\infty}^\infty g_{0,0,0}(y,x,w^*)dydx\right)^2}\\
    &\quad\text{and}\\
    &\quad\frac{\partial F_{X\mid W^*=w^*}(x)}{\partial x}=\frac{\left(\int_{-\infty}^\infty g_{0,0,0}(y,x,w^*)dy\right)}{\left(\int_{-\infty}^\infty\int_{-\infty}^\infty g_{0,0,0}(y,x,w^*)dydx\right)}.
\end{align*}
\indent For the \textbf{LAR}, I will introduce a weighted average version of it, and propose an estimator for this weighted average version of the \textbf{LAR}. The aim of introducing the weight is to ensure that the integration is taken on a set where the joint density $f_{Y,X,W^*}$ is bounded away from zero. This has the benefit of guaranteeing that the relevant functional is Fr\'echet differentiable when $f_{Y,X,W^*}$ appears in the denominator of it. The weighted average version of the \textbf{LAR} is defined as follows: 
\begin{align}
    &\textbf{WLAR}(x)\equiv E\left[\omega(X,W^*,\epsilon)\frac{\partial m(X,\epsilon)}{\partial x}\mid X=x\right]\notag\\
    &=\int \int_0^1 \omega(x,w^*,\epsilon)\left.\frac{\partial m(x,\epsilon)}{\partial x}\right|_{\epsilon =s(r(x, w^*),\delta)} f_{W^*|X=x}(w^*)d\delta dw^*,\label{weightedLAR}
\end{align}
where $\omega(x,w^*,\epsilon)$ is a known or estimable weighting function taking value 0 outside a compact set $\mathbb M$. As a weighting function, $\omega(x,w^*,\epsilon)$ also satisfies that given any values of $x$ such that $\omega(x,\cdot,\cdot)$ could obtain non-zero values, $\int_{w^*}\int_\epsilon\omega(x,w^*,\epsilon)f_{\epsilon, W^*|X=x}(\epsilon, w^*) d\epsilon dw^*=1$. The specific form of $\omega$ is up to the researcher's choice, however, to ensure easy calculation of $\omega(x,w^*,s(r(x,w^*),\delta))$ on the right-hand side of (\ref{weightedLAR}), I require that function $\omega(x,w^*,\epsilon)$ satisfy the following restriction: $\omega(x,w^*,\epsilon)=\tilde \omega(x,w^*)$ if $\epsilon \in \left[q_{x,w^*}(\tau_l), q_{x,w^*}(\tau_u)\right]$ and $\omega(x,w^*,\epsilon)=0$ otherwise, where $\tilde\omega(x,w^*)$ is a known function with compact support, and $q_{x,w^*}(\tau)$ denotes the conditional-$\tau$ quantile of $\epsilon$ given $X=x, W^*=w^*$. The specific form of function $\tilde w$ and the specific values of $\tau_l$ and $\tau_u$ are up to the researcher's choice. Under this restriction, $\omega(x,w^*,s(r(x,w^*),\delta))$ is equal to $\tilde w(x,w^*)$ if $\delta$ is between $\tau_l$ and $\tau_u$, and is equal to 0 otherwise. An example of function $\omega(x,w^*,\epsilon)$ is when $\tilde \omega(x,w^*)$ is a constant equal to $\left[(\tau_u-\tau_l)\int_{\underline{w}^*}^{\bar w^*} f_{W^*|X=x}(w^*)dw^*\right]^{-1}$. The \textbf{WLAR} in this case can be interpreted as the Local Average Response of a subgroup of individuals whose unobserved $\epsilon$ is between its $\tau_l$ and $\tau_u$ conditional quantile given $X, W^*$ and whose $W^*$ is between $\underline w^*$ and $\bar w^*$. Note that even though $X$ is correlated with $\epsilon$, neither the distribution of $\epsilon|X$ nor the subgroup of individuals changes as I take the derivative with respect to $x$. This means that my objective of interest is the same group of people when studying the effect of the counterfactual changes of the endogenous $X$. The \textbf{WLAR} can also be written as a functional of $g_{\lambda_1,\lambda_{2,1}, \lambda_{2,2}}(y,x,w^*)$. The explicit form can be derived from Lemma \ref{dcquantile} and equation (\ref{pmpxdelta}).\\
\indent So far I have written the parameters of interest as functionals of $g_{\lambda_1,\lambda_{2,1}, \lambda_{2,2}}(y,x,w^*)$ in explicit forms. Next I focus on the estimation of $g_{\lambda_1,\lambda_{2,1}, \lambda_{2,2}}(y,x,w^*)$. To deal with the well-known ill-posed inverse problem when inverting a convolution operator, I base my estimator on a smoothed version of $g_{\lambda_1,\lambda_{2,1}, \lambda_{2,2}}(y,x,w^*)$:
\begin{align*}
    &g_{\lambda_1,\lambda_{2,1}, \lambda_{2,2}}(y,x,w^*,h_1)\equiv\int \frac{1}{h_1}K\left(\frac{\tilde w^*-w^*}{h_1}\right)g_{\lambda_1,\lambda_{2,1}, \lambda_{2,2}}(y,x,\tilde w^*)d\tilde w^*,
\end{align*}
and use kernel functions satisfying the following assumption:
\begin{assumption}\label{phik}
The kernel functions $K: \mathbb{R}\rightarrow \mathbb{R}$, $G_Y: \mathbb{R}\rightarrow \mathbb{R}$ and $G_X: \mathbb{R}\rightarrow \mathbb{R}$ are measurable, symmetric. $\int K(x)dx=1$, $\int G_Y(y)dy=1$, $\int G_X(x)dx=1$.  $G_Y$, $G_X$ are differentiable. $G_Y$, $G_X$, and their derivatives are bounded, and denote the maximum of these bounds as $\bar G$. Their Fourier transforms  $\xi\rightarrow\phi_{K}(\xi)$, $\xi\rightarrow\phi_{G_Y}(\xi)$ and $\xi\rightarrow\phi_{G_X}(\xi)$ obey: (i) $\phi_{F}$ is compactly supported (without loss of generality, the support is $[-1,1]$) for $F\in \{K, G_X, G_Y\}$; and (ii) there exists $\bar \xi_F$ such that $\phi_{F}(\xi)=1$ for $|\xi|\leq\bar\xi_F$, where $F\in \{K,G_X,G_Y\}$.
\end{assumption}
By assuming (ii), I use flat-top kernels proposed by \cite{politis1999multivariate}. The benefit of using flat-top kernels is that the rate of decrease of the bias term will not be affected by the order of the kernel, and will only be affected by the smoothness of the function to be estimated. The restriction of compact support of the Fourier transform is without loss of generality. As discussed by \cite{schennach2004nonparametric}, given any kernel $K$, one can always create a modified kernel $\tilde K$ that satisfies the assumption by using a ``windowing'' function. For example,
\begin{align*}
    \phi_{\tilde K}(t)=W(t)\phi_{K}(t).
\end{align*}
where
\begin{align*}
    W(t)=\begin{cases}1 & \text { if }|t| \leq \bar{t} \\
\left(1+\exp \left((1-\bar{t})\left((1-|t|)^{-1}-(|t|-\bar{t})^{-1}\right)\right)\right)^{-1} & \text { if } 1 \geq|t|>\bar{t} \\
0 & \text { if }|t|>1
    \end{cases}
\end{align*}
for some $\bar t\in(0,1)$.
\begin{assumption}\label{E<infty}
$E[|\Delta W_1|]<\infty$.
\end{assumption}

\begin{comment}
\textcolor{red}{see \cite{song2015estimating} P758 for arguments}\end{comment} 
The following lemma shows that the smoothed version $g_{\lambda_1,\lambda_{2,1}, \lambda_{2,2}}(y,x,w^*,h_1)$ can be written as an expression involving the Fourier transform of the kernel function and the characteristic functions of different variables. This expression makes it more straightforward to introduce the estimator that will appear soon below.
\begin{lemma}\label{convolution thm}
For $(y,x,w^*)\in\mathbb{S}_{(Y,X,W^*)}$ and $h>0$, let
\begin{align*}
    &g_{\lambda_1,\lambda_{2,1}, \lambda_{2,2}}(y,x,w^*,h_1)\equiv\int \frac{1}{h_1}K\left(\frac{\tilde w^*-w^*}{h_1}\right)g_{\lambda_1,\lambda_{2,1}, \lambda_{2,2}}(y,x,\tilde w^*)d\tilde w^*.
\end{align*}
where $K$ satisfies Assumption \ref{phik}. Then under Assumptions \ref{posdenom}, \ref{cpctcontdiff}, and \ref{E<infty},
\begin{align*}
    &g_{\lambda_1,\lambda_{2,1}, \lambda_{2,2}}(y,x,w^*,h_1) = \frac{1}{2\pi}\int \left(-\mathbf{i}t\right)^{\lambda_1}e^{-\mathbf{i}tw^*} \phi_K(h_1t)\frac{\phi_{f^{(\lambda_{2,1},\lambda_{2,2})}_{Y,X,W_2}(y,x, \cdot)}(t)\phi_{W^*}(t)}{\phi_{W_2}(t)} dt
\end{align*}
where $f^{(\lambda_{2,1},\lambda_{2,2})}_{Y,X,W_2}(y,x,w_2)\equiv \frac{\partial^{\lambda_2}f_{Y,X,W_2}(y,x,w_2)}{\partial y^{\lambda_{2,1}}\partial x^{\lambda_{2,2}}}$.
\end{lemma}
\begin{proof}
See Appendix B.
\end{proof}
I now define my estimator for $g_{\lambda_1,\lambda_{2,1}, \lambda_{2,2}}(y,x,w^*,h_1)$ by replacing $\phi_{f^{(\lambda_{2,1},\lambda_{2,2})}_{Y,X,W_2}(y,x,\cdot)}(t)$, $\phi_{W^*}(t)$, and $\phi_{W_2}(t)$ in Lemma \ref{convolution thm} with their sample analogues. Formally,
\begin{definition}
Let $h_n\equiv(h_{1n}, h_{2n,1}, h_{2n,2})\rightarrow 0 $ as $n\rightarrow 0$. The estimator for $g_{\lambda_1,\lambda_{2,1}, \lambda_{2,2}}(y,x,w^*)$ is defined as 
\begin{align*}
    &\hat g_{\lambda_1,\lambda_{2,1}, \lambda_{2,2}}(y,x,w^*,h_n)\equiv \frac{1}{2\pi}\int\left(-\mathbf{i}t\right)^{\lambda_1} e^{-\mathbf{i}tw^*} \phi_K(h_{1n}t)\frac{\hat\phi_{f^{(\lambda_{2,1},\lambda_{2,2})}_{Y,X,W_2}(y,x,\cdot)}(t)\hat\phi_{W^*}(t)}{\hat\phi_{W_2}(t)} dt,
\end{align*}
where 
\begin{align*}
   &\hat\phi_{f^{(\lambda_{2,1},\lambda_{2,2})}_{Y,X,W_2}(y,x,\cdot)}(t) \equiv \hat E\left[e^{\mathbf{i}tW_{2}}\frac{1}{h_{2n,1}^{1+\lambda_{2,1}}h_{2n,2}^{1+\lambda_{2,2}}}G_Y^{(\lambda_{2,1})}\left(\frac{y-Y}{h_{2n,1}}\right)G_X^{(\lambda_{2,2})}\left(\frac{x-X}{h_{2n,2}}\right)\right],\\
    &\hat\phi_{W_2}(t)\equiv\hat E\left[e^{\mathbf{i}tW_2}\right],\\
    &\hat\phi_{W^*}(t) \equiv \exp\left(\int_0^{t}\frac{\hat E\left[iW_1 e^{i\xi W_2}\right]}{\hat E\left[e^{i\xi W_2}\right]}d\xi\right),
\end{align*}
where $\hat E$ denotes the sample average, $G_Y, G_X$ denotes the kernel functions for $Y,X$, respectively, and $G_X^{(\lambda_{2,2})}(x)\equiv\frac{\partial^{\lambda_{2,2}} G_X(x)}{\partial x^{\lambda_{2,2}}}$. $G_Y, G_X$ are not necessarily the same as $K$.
\end{definition}
\indent Having defined $\hat g_{\lambda_1,\lambda_{2,1}, \lambda_{2,2}}(y, x ,w^*,h_n)$, I can now define the estimator for $\rho(\bar y, \bar x)$ by replacing $g_{\lambda_1,\lambda_{2,1}, \lambda_{2,2}}(y,x,w^*)$ with its estimator on the right-hand-side of Equation (\ref{pmpxbar}). Note that by Fubini's Theorem 
\begin{footnotesize}
    \begin{align*}
        &\quad\int_{-\infty}^y \hat g_{\lambda_1,0,0, \lambda_{2,2}}(s,x,w^*,h_n)ds=\frac{1}{2\pi}\int\left(-\mathbf{i}t\right)^{\lambda_1} e^{-\mathbf{i}tw^*} \phi_K(h_{1n}t)\frac{\int_{-\infty}^y\hat\phi_{f^{(0,\lambda_{2,2})}_{Y, X,W_2}(s,x, \cdot)}(t)ds\hat\phi_{W^*}(t)}{\hat\phi_{W_2}(t)} dt\\
        &\quad\int_{-\infty}^\infty \hat g_{\lambda_1,0,0, \lambda_{2,2}}(y,x,w^*,h_n)dy=\frac{1}{2\pi}\int\left(-\mathbf{i}t\right)^{\lambda_1} e^{-\mathbf{i}tw^*} \phi_K(h_{1n}t)\frac{\hat\phi_{f^{(0,\lambda_{2,2})}_{X,W_2}(x, \cdot)}(t)\hat\phi_{W^*}(t)}{\hat\phi_{W_2}(t)} dt\\
        &\quad\int_{-\infty}^x\int_{-\infty}^\infty\hat g_{\lambda_1,0,0}(y,s,w^*,h_n)dyds=\frac{1}{2\pi}\int\left(-\mathbf{i}t\right)^{\lambda_1} e^{-\mathbf{i}tw^*} \phi_K(h_{1n}t)\frac{\int_{-\infty}^x\hat\phi_{f_{X,W_2}(s, \cdot)}(t)ds\hat\phi_{W^*}(t)}{\hat\phi_{W_2}(t)} dt\\
        &\quad\int_{-\infty}^\infty\int_{-\infty}^\infty\hat g_{\lambda_1,0,0}(y,x,w^*,h_n)dydx=\frac{1}{2\pi}\int\left(-\mathbf{i}t\right)^{\lambda_1} e^{-\mathbf{i}tw^*} \phi_K(h_{1n}t)\hat\phi_{W^*}(t) dt,
    \end{align*}
\end{footnotesize}
where I've let $\tilde G_Y(y)\equiv\int_{-\infty}^y G_Y\left(u\right)du$ and $\tilde G_X(x)\equiv\int_{-\infty}^x G_X\left(u\right)du$ denote the kernel CDFs, and I've defined:
\begin{align*}
    &\int_{-\infty}^y \hat\phi_{f^{(0,\lambda_{2,2})}_{Y, X,W_2}(s,x, \cdot)}(t)ds =\hat E\left[e^{\mathbf{i}tW_{2}}\frac{1}{h_{2n,2}^{1+\lambda_{2,2}}}\tilde G_Y\left(\frac{y-Y}{h_{2n,1}}\right)G_X^{(\lambda_{2,2})}\left(\frac{x-X}{h_{2n,2}}\right)\right]\\
    &\hat\phi_{f^{(\lambda_{2,2})}_{X,W_2}(x, \cdot)}(t)ds =\hat E\left[e^{\mathbf{i}tW_{2}}\frac{1}{h_{2n,2}^{1+\lambda_{2,2}}}G_X^{(\lambda_{2,2})}\left(\frac{x-X}{h_{2n,2}}\right)\right]\\
    &\int_{-\infty}^x\hat\phi_{f_{X,W_2}(s, \cdot)}(t)ds =\hat E\left[e^{\mathbf{i}tW_{2}}\tilde G_X\left(\frac{x-X}{h_{2n,2}}\right)\right].
\end{align*}
\subsection{The Asymptotic Properties of the Estimator}
In this subsection, I analyze the asymptotic properties for $\hat g_{\lambda_1,\lambda_{2,1}, \lambda_{2,2}}(y,x,w^*,h_n)$. I begin by decomposing the difference between the estimator and the true value of the parameter as the sum of a bias term, a variance term, and a remainder term. 
\begin{lemma}\label{decomposition}
Suppose $\left\{Y_i, X_i, W^*_i, \Delta W_{1,i}, W_{2,i}\right\}$ is an IID sequence satisfying Assumptions \ref{repeatedmeas}, \ref{measerrors}, \ref{posdenom}, \ref{cpctcontdiff}, and \ref{E<infty}, and that Assumption \ref{phik} holds. Then for $(y,x,w^*)\in\mathbb{S}_{(Y,X,W^*)}$, and $h>0$,
\begin{align*}
    &\quad\hat g_{\lambda_1,\lambda_{2,1}, \lambda_{2,2}}(y, x, w^*,h)-g_{\lambda_1,\lambda_{2,1}, \lambda_{2,2}}(y, x, w^*) \\
    &= B_{\lambda_1,\lambda_{2,1}, \lambda_{2,2}}(y, x, w^*, h)+L_{\lambda_1,\lambda_{2,1}, \lambda_{2,2}}(y, x, w^*, h)+R_{\lambda_1,\lambda_{2,1}, \lambda_{2,2}}(y, x, w^*, h)
\end{align*}
where $B_{\lambda_1,\lambda_{2,1}, \lambda_{2,2}}(y, x, w^*, h)$ is the bias term admitting the linear representation:
\begin{footnotesize}
\begin{align*}
    &\quad B_{\lambda_1,\lambda_{2,1}, \lambda_{2,2}}(y, x, w^*, h) \\
    &= E\left[\bar g_{\lambda_1,\lambda_{2,1}, \lambda_{2,2}}(y, x, w^*,h)\right] - g_{\lambda_1,\lambda_{2,1},\lambda_{2,2}}(y,x,w^*)\\
    &= g_{\lambda_1,\lambda_{2,1},\lambda_{2,2}}(y,x,w^*,h_1)-g_{\lambda_1,\lambda_{2,1},\lambda_{2,2}}(y,x,w^*)\\
    &\quad+\frac{1}{2\pi}\int\left(-\mathbf{i}t\right)^{\lambda_1} e^{-\mathbf{i}tw^*} \phi_K(h_{1n}t)\frac{E\left[\hat\phi_{f^{(\lambda_{2,1}, \lambda_{2,2})}_{Y,X,W_2}(y,x,\cdot)}(t)-\phi_{f^{(\lambda_{2,1}, \lambda_{2,2})}_{Y,X,W_2}(y,x,\cdot)}(t)\right]\phi_{W^*}(t)}{\phi_{W_2}(t)} dt;
\end{align*}
\end{footnotesize}
$L_{\lambda_1,\lambda_{2,1}, \lambda_{2,2}}(y, x, w^*, h)$ is the variance term admitting the linear representation:
\begin{footnotesize}
\begin{align*}
    &\quad L_{\lambda_1,\lambda_{2,1}, \lambda_{2,2}}(y, x, w^*, h) \\
    &= \bar g_{\lambda_1,\lambda_{2,1}, \lambda_{2,2}}(y, x, w^*,h)-E\left[\bar g_{\lambda_1,\lambda_{2,1}, \lambda_{2,2}}(y, x, w^*,h)\right] \\
   &=\hat E\Bigg[\int \Psi_{1, \lambda_1,\lambda_{2,1},\lambda_{2,2}}(\xi, y, x, w^*, h_1)\left(W_1e^{\mathbf{i}\xi W_2}-E[W_1e^{\mathbf{i}\xi W_2}]\right)d\xi\notag\\
    &\quad+\int \Psi_{2,\lambda_1,\lambda_{2,1},\lambda_{2,2}}(\xi, y, x, w^*, h_1)\left(e^{\mathbf{i}\xi W_2}-E[e^{\mathbf{i}\xi W_2}]\right)d\xi\notag\\
    &\quad+\int \Psi_{3,\lambda_1,\lambda_{2,1},\lambda_{2,2}}(\xi, y, x, w^*, h_1)\times\\
    &\quad\quad\quad\quad\left(e^{\mathbf{i}\xi W_2}\frac{1}{h_{2,1}^{1+\lambda_{2,1}}h_{2,2}^{1+\lambda_{2,2}}}G_Y^{(\lambda_{2,1})}\left( \frac{y-Y}{h_{2,1}}\right)G_X^{(\lambda_{2,2})}\left(\frac{x-X}{h_{2,2}}\right)\right.\\
    &\quad\quad\quad\quad\quad\quad\left.-E\left[e^{\mathbf{i}\xi W_2}\frac{1}{h_{2,1}^{1+\lambda_{2,1}}h_{2,2}^{1+\lambda_{2,2}}}G_Y^{(\lambda_{2,1})}\left( \frac{y-Y}{h_{2,1}}\right)G_X^{(\lambda_{2,2})}\left(\frac{x-X}{h_{2,2}}\right)\right]\right) d\xi\Bigg]\notag\\
    &=\hat E[l_{\lambda_1,\lambda_{2,1}, \lambda_{2,2}}(y, x, w^*,h;Y, X, W_1, W_2)],
\end{align*}
\end{footnotesize}
where I've let $\theta(\xi) \equiv E\left[W_1 e^{\mathrm{i} \xi W_{2}}\right]$ and defined
\begin{align*}
    &\Psi_{1, \lambda_1,\lambda_{2,1},\lambda_{2,2}}(\xi, y, x, w^*, h_1) = \frac{1}{2 \pi} \frac{\mathbf{i}}{\phi_{W_2}(\xi)}\int_{\xi}^{\pm\infty} \left(-\mathbf{i}t\right)^{\lambda_1} e^{-\mathbf itw^*} \phi_{K}(h_1t)\phi_{f^{(\lambda_{2,1},\lambda_{2,2})}_{Y,X,W^*}(y,x,\cdot)}(t)dt \\
    &\Psi_{2,\lambda_1,\lambda_{2,1},\lambda_{2,2}}(\xi, y, x, w^*, h_1) \\
    &= -\frac{1}{2 \pi} \frac{\mathbf{i}\theta(\xi)}{(\phi_{W_2}(\xi))^2}\int_{\xi}^{\pm\infty}  \left(-\mathbf{i}t\right)^{\lambda_1}e^{-\mathbf itw^*} \phi_{K}(h_1 t)\phi_{f^{(\lambda_{2,1},\lambda_{2,2})}_{Y,X,W^*}(y,x,\cdot)}(t)dt\\
    &\quad -\frac{1}{2 \pi} \left(-\mathbf{i}t\right)^{\lambda_1}e^{-\mathbf i\xi w^*} \phi_{K}(h _1\xi)\frac{\phi_{f^{(\lambda_{2,1},\lambda_{2,2})}_{Y,X,W^*}(y,x,\cdot)}(\xi)}{\phi_{W_2}(\xi)}\\
    &\Psi_{3,\lambda_1,\lambda_{2,1},\lambda_{2,2}}(\xi, y, x, w^*, h_1) = \frac{1}{2 \pi}   \left(-\mathbf{i}t\right)^{\lambda_1}e^{-\mathbf i\xi w^*} \phi_{K}(h_1\xi)\frac{\phi_{W^*}(\xi)}{\phi_{W_2}(\xi)}.
\end{align*}
where for a given function $\zeta\rightarrow f(\zeta)$, I've written $\int_{\xi}^{\pm \infty}f(\zeta)d\zeta\equiv\lim_{c\rightarrow+\infty}\int_{\xi}^{c\xi}f(\zeta)d\zeta$; and $R_{\lambda_1,\lambda_{2,1}, \lambda_{2,2}}(y, x, w^*, h)$ is the remainder term:
\begin{align*}
    R_{\lambda_1,\lambda_{2,1}, \lambda_{2,2}}(y, x, w^*, h) = \hat g_{\lambda_1,\lambda_{2,1}, \lambda_{2,2}}(y, x, w^*,h)-\bar g_{\lambda_1,\lambda_{2,1}, \lambda_{2,2}}(y, x, w^*,h).
\end{align*}
\end{lemma}
\begin{proof}
See Appendix B.
\end{proof}
\indent To derive the uniform rate of convergence, I impose some assumptions on the smoothness of $g_{\lambda_1,\lambda_{2,1}, \lambda_{2,2}}(y,x,w^*)$, stated in terms of the rate of decay of the tail of its Fourier transform:
\begin{assumption}\label{smoothness1}
(i) There exist constants $C_{\phi}>0$,  $\alpha_{\phi}\leq0$, $\beta_{\phi}\geq0$ and $\gamma_{\phi}\in\mathbb{R}$ such that $\beta_{\phi}\gamma_{\phi}\geq 0$ and for $j=1,2$
    \begin{align*}
    &\max_{\lambda_2\in\{0,1\}}\sup_{(y,x)\in\mathbb{S}_{(Y,X)}}\left|\phi_{f^{(\lambda_{2,1},\lambda_{2,2})}_{Y,X,W^*}(y,x, \cdot)}(t)\right| \leq C_{\phi}(1+|t|)^{\gamma_{\phi}} \exp \left(\alpha_{\phi}|t|^{\beta_{\phi}}\right)\\
    &\left|\phi_{W^*}(t)\right| \leq C_{\phi}(1+|t|)^{\gamma_{\phi}} \exp \left(\alpha_{\phi}|t|^{\beta_{\phi}}\right),
    \end{align*}
where $f^{(\lambda_{2,1},\lambda_{2,2})}_{Y,X,W_2}(y,x,w_2)\equiv \frac{\partial^{\lambda_2}f_{Y,X,W_2}(y,x,w_2)}{\partial y^{\lambda_{2,1}}\partial x^{\lambda_{2,2}}}$. Moreover, if $\beta_{\phi}=0$, then for given $\lambda \in\{0,1\}, \gamma_{\phi}<-\lambda_1-1$.\\
(ii) Denote the Fourier transform of $f^{(\lambda_{2,1}, \lambda_{2,2})}_{Y,X,W^*}(y,x,w^*)$ as $\phi_{f^{(\lambda_{2,1}, \lambda_{2,2})}}(t_1,t_2,\xi)$. There exists constants $C_{f_1}>0$, $C_{f_2}>0$, $\alpha_{f_1}\leq 0$, $\alpha_{f_2}\leq 0$, $\beta_{f_1}\geq 0$, $\beta_{f_2}\geq 0$ and $\gamma_{f_1}\in\mathbb{R}$, $\gamma_{f_2}\in\mathbb{R}$ such that $\beta_{f_1}\gamma_{f_1}\geq 0$, $\beta_{f_2}\gamma_{f_2}\geq 0$, and
\begin{small}
\begin{align*}
   \left|\phi_{f^{(\lambda_{2,1}, \lambda_{2,2})}}(t_1,t_2,\xi)\right| \leq C_{f_1}C_{f_2}(1+|t_1|)^{\gamma_{f_1}}(1+|t_2|)^{\gamma_{f_2}}\exp\left(\alpha_{f_1}|t_1|^{\beta_{f_1}}+\alpha_{f_2}|t_2|^{\beta_{f_2}}\right)\left|\phi_{W^*}(t)\right|.
\end{align*}    
\end{small}
Moreover, if $\beta_{f_1}=0$, $\gamma_{f_1}<-1$, and if $\beta_{f_2}=0$, $\gamma_{f_2}<-1$.
\end{assumption}
In Assumption \ref{smoothness1}(i) I impose the same bound for the tail behavior of $\phi_{f^{(\lambda_{2,1},\lambda_{2,2})}_{Y,X,W^*}(y,x,\cdot)}$ and $\phi_{W^*}(t)$. This is without loss of generality since they have the same effect on the convergence rate. Assumption \ref{smoothness1}(ii) is similar to Assumption A7 in \cite{li2002robust}, and can be viewed as a generalization of smoothness condition from the univariate case to the multivariate case.\\
\indent I next state the first main result of this paper, the uniform asymptotic rate of the bias term:
\begin{theorem}\label{biasrate}
Let the conditions of Lemma \ref{decomposition} hold with $\{Y_i,X_i,W_{i}^*, \Delta W_{1,i}, \Delta W_{2,i}\}$ IID, and suppose in addition that Assumption \ref{smoothness1} holds. Then for $h>0$,
\begin{align*}
    \sup_{(y,x,w^*)\in\mathbb{S}_{(Y,X,W^*)}}\left|B_{\lambda_1,\lambda_{2,1}, \lambda_{2,2}}(y,x,w^*,h)\right| = O\left(\left(h_1^{-1}\right)^{1+\gamma_{\phi}+\lambda_1}\exp\left(\alpha_{\phi}\bar\zeta_K^{\beta_{\phi}}\left(h_1^{-1}\right)^{\beta_{\phi}}\right)\right).
\end{align*}
\end{theorem}
\begin{proof}
See Appendix B.
\end{proof}
I impose the following assumption to ensure finite variance:
\begin{assumption}\label{finite variance}
For some $\delta>0$, $E[\left|W_1\right|^{2+\delta}]<\infty$, $\sup_{w_2\in\mathbb{S}_{W_2}}E\left[\left|W_1\right|^{2+\delta}\mid W_2=w_2\right]<\infty$ .
\end{assumption}
To derive the rate for $L_{\lambda_1,\lambda_{2,1}, \lambda_{2,2}}(y,x,w^*,h)$, I impose the following assumption on the tail behavior of the Fourier transforms involved. These are common in the deconvolution literature (e.g. \cite{fan1991optimal}, \cite{fan1993nonparametric}, and \cite{schennach2004nonparametric}).
\begin{assumption}\label{smoothness2}
\begin{enumerate}[(i)]
    \item There exist constants $C_{2}>0$,  $\alpha_{2}\leq 0$, $\beta_{2}\geq\beta_{\phi}\geq 0$ and $\gamma_{2}\in\mathbb{R}$ such that $\beta_{2}\gamma_{2}\geq 0$ and 
    \begin{align*}
        \left|\phi_{W_2}(t)\right| \geq C_{2}(1+|t|)^{\gamma_{2}} \exp \left(\alpha_{2}|t|^{\beta_{2}}\right).
    \end{align*}
    Moreover, if $\beta_2=0$, $\lambda_1-\gamma_2+\gamma_{\phi}>0$.
    \item There exist constants $C_*>0$ and $\gamma_*\geq 0$ such that
    \begin{align*}
        \left|\frac{\phi'_{W^*}(t)}{\phi_{W^*}(t)}\right|\leq C_*\left(1+|t|\right)^{\gamma_*}.
    \end{align*}
\end{enumerate}
\end{assumption}
I explicitly impose $\beta_2\geq\beta_{\phi}$ because
\begin{align*}
    C_{\phi}\left(1+|t|\right)^{\gamma_{\phi}}\exp\left(\alpha_{W^*}|t|^{\beta_{\phi}}\right)&\geq \left|\phi_{W^*}(t)\right|=\left|E[e^{\mathbf{i}tW^*}]\right|\\
    &\geq \left|E[e^{\mathbf{i}tW^*}]\right|\left|E[e^{\mathbf{i}t\Delta W_2}]\right|\geq \left|E[e^{\mathbf{i}tW_2}]\right|\\
    &=|\phi_{W_2}(t)|\geq C_{2}\left(1+|t|\right)^{\gamma_{2}}\exp\left(\alpha_{2}|t|^{\beta_{2}}\right).
\end{align*}
The following theorem states the asymptotic properties of the linear term \\ $L_{\lambda_1, \lambda_{2,1}, \lambda_{2,2}}(y,x,w^*,h)$ and facilitates the analysis of various quantities of interest later. 
\begin{theorem}\label{Omegaprop}
Suppose the conditions of Lemma \ref{decomposition} hold. (i) Then for each $(y, x, w^*)\in\mathbb{S}_{(y, x, w^*)}$, $E[L_{\lambda_1, \lambda_{2,1}, \lambda_{2,2}}(y,x,w^*,h)]=0$, and if Assumption \ref{finite variance} also holds, then \begin{footnotesize}
    $E[L_{\lambda_1, \lambda_{2,1}, \lambda_{2,2}}^2(y,x,w^*,h)]$$=n^{-1}\Omega_{\lambda_1, \lambda_{2,1}, \lambda_{2,2}}(y,x,w^*,h)$
\end{footnotesize}, where \begin{footnotesize}
    $\Omega_{\lambda_1, \lambda_{2,1}, \lambda_{2,2}}(y,x,w^*,h)\equiv E[\left(l_{\lambda_1, \lambda_{2,1}, \lambda_{2,2}}(y,x,w^*,h;Y, X, W_1, W_2)\right)^2]$
\end{footnotesize}.\\
\indent Further, if Assumption \ref{smoothness1} and \ref{smoothness2} also holds then
\begin{footnotesize}
\begin{align}
    &\quad\sqrt{\sup_{(y, x, w^*)\in\mathbb{S}_{(Y,X,W^*)}}\Omega_{\lambda_1, \lambda_{2,1}, \lambda_{2,2}}(y,x,w^*,h)}\notag\\&= O\bigg(\max\left\{(h_1^{-1})^{1+\gamma_*},(h_{2,1}^{-1})^{1+\lambda_{2,1}}(h_{2,2}^{-1})^{1+\lambda_{2,2}}\right\}\notag\\
    &\quad\quad\quad\quad\left.(h_1^{-1})^{1-\gamma_2+\gamma_{\phi}+\lambda_1}\exp\left(\left(\alpha_{\phi}\mathbf{1}\{\beta_{\phi}=\beta_2\}-\alpha_2\right)\left(h_1^{-1}\right)^{\beta_2}\right)\right). \label{sqrtsupOmega}
\end{align}
\end{footnotesize}
I also have
\begin{footnotesize}
\begin{align}
    &\quad\sup_{(y, x, w^*)\in\mathbb{S}_{(Y,X,W^*)}} \left|L_{\lambda_1, \lambda_{2,1}, \lambda_{2,2}}(y, x, w^*,h)\right|\notag\\
    &=O\bigg(n^{-\frac{1}{2}}\max\left\{(h_1^{-1})^{1+\gamma_*},(h_{2,1}^{-1})^{1+\lambda_{2,1}}(h_{2,2}^{-1})^{1+\lambda_{2,2}}\right\}\notag\\
    &\quad\quad\quad\quad\left.(h_1^{-1})^{1-\gamma_2+\gamma_{\phi}+\lambda_1}\exp\left(\left(\alpha_{\phi}\mathbf{1}\{\beta_{\phi}=\beta_2\}-\alpha_2\right)\left(h_1^{-1}\right)^{\beta_2}\right)\right). \label{supL}
\end{align}
\end{footnotesize}
\indent (ii) If Assumption \ref{finite variance} also holds, and if for each $(y,x,w^*)\in\mathbb{S}_{(y, x, w^*)}$, $\Omega_{\lambda_1, \lambda_{2,1}, \lambda_{2,2}}(y,x,w^*,h_n)>0$ for all $n$ sufficiently large, then for each $(y, x, w^*)\in\mathbb{S}_{(y,x,w^*)}$
\begin{align*}
    n^{1/2}\left(\Omega_{\lambda_1, \lambda_{2,1}, \lambda_{2,2}}(y,x,w^*,h_n)\right)^{-1/2}L_{\lambda_1, \lambda_{2,1}, \lambda_{2,2}}(y,x,w^*,h_n)\xrightarrow{d}N(0,1).
\end{align*}
\end{theorem}
\begin{proof}
See Appendix B.
\end{proof}

\begin{comment}
\begin{assumption}\textcolor{red}{for asymptotic normality}\label{bandwidth1}
If $\beta_2=0$ in Assumption \ref{smoothness2}, then $h_n^{-1}=O\left(n^{-\eta} n^{(3/2)/(4-\gamma_2+\gamma_{W^*}+\gamma_*)}\right)$ for some $\eta>0$; otherwise $h^{-1}_n=O\left(\left(\ln n\right)^{\beta_2^{-1}-\eta}\right)$ for some $\eta>0$.
\end{assumption}
\end{comment}

Finally, I bound the remainder term. To do that, I first impose some restrictions on the moments of $W_2$:
\begin{assumption}\label{finite W2 mmts}
$E[|W_2|]< \infty, E[|W_1W_2|]<\infty$. 
\end{assumption}
I then impose the following bounds for the bandwidths:
\begin{assumption}\label{bandwidth2}
If $\beta_2=0$ in Assumption \ref{smoothness2}, then $h_{1n}^{-1}=O\left(n^{(4+4\gamma_*-4\gamma_{2})^{-1}-\eta}\right)$ and $h_{2n,j}^{-1}=O\left(n^{(16+8\lambda_2)^{-1}-\eta}\right)$, for some $\eta>0$, $j=1,2$; otherwise $h^{-1}_{1n}=O\left(\left(\ln n\right)^{\beta_2^{-1}-\eta}\right)$ and $h_{2n,j}^{-1}=O\left(n^{(8+4\lambda_2)^{-1}-\eta}\right)$, for some $\eta>0$, $j=1,2$.
\end{assumption}

\begin{theorem}\label{Remainder}(i) Suppose the conditions of Lemma \ref{decomposition} and Assumption \ref{smoothness1}, \ref{finite variance}, \ref{smoothness2}, and \ref{finite W2 mmts} hold. Then
\begin{footnotesize}
\begin{align*}
    &\quad\sup_{(y,x,w^*)\in\mathbb{S}_{(Y,X,W^*)}}\left|R_{\lambda_1,\lambda_{2,1}, \lambda_{2,2}}(y,x,w^*,h_n)\right| \\
    &=o_p\left(h_{2,1}^{-2}(h_{2,2}^{-1})^{2+2\lambda_2}n^{-1+2\epsilon}(h_{1n}^{-1})^{1+\gamma_{*}-\gamma_2}\exp\left(-\alpha_2\left(h_{1n}^{-1}\right)^{\beta_2}\right)\right)\\
    &\quad\times O\bigg(\max\left\{(h_1^{-1})^{1+\gamma_*},(h_{2,1}^{-1})^{1+\lambda_{2,1}}(h_{2,2}^{-1})^{1+\lambda_{2,2}}\right\}\\
    &\quad\quad\quad\quad\left.(h_1^{-1})^{1-\gamma_2+\gamma_{\phi}+\lambda_1}\exp\left(\left(\alpha_{\phi}\mathbf{1}\{\beta_{\phi}=\beta_2\}-\alpha_2\right)\left(h_1^{-1}\right)^{\beta_2}\right)\right)
\end{align*}
\end{footnotesize}
for arbitrarily small $\epsilon>0$. (ii) If Assumption \ref{bandwidth2} also holds, then
\begin{footnotesize}
\begin{align*}
     &\quad\sup_{(y,x,w^*)\in\mathbb{S}_{(Y,X,W^*)}}\left|R_{\lambda_1,\lambda_{2,1}, \lambda_{2,2}}(y,x,w^*,h_n)\right| =\\
     &o_p\bigg(n^{-\frac{1}{2}}\max\left\{(h_1^{-1})^{1+\gamma_*},(h_{2,1}^{-1})^{1+\lambda_{2,1}}(h_{2,2}^{-1})^{1+\lambda_{2,2}}\right\}\\
     &\quad\quad\quad\quad\left.(h_1^{-1})^{1-\gamma_2+\gamma_{\phi}+\lambda_1}\exp\left(\left(\alpha_{\phi}\mathbf{1}\{\beta_{\phi}=\beta_2\}-\alpha_2\right)\left(h_1^{-1}\right)^{\beta_2}\right)\right).  
\end{align*}
\end{footnotesize}
\end{theorem}
\begin{proof}
See Appendix B.
\end{proof}

\indent Collecting the results from Theorem \ref{biasrate},\ref{Omegaprop} and \ref{Remainder} yields a straightforward corollary of the uniform rate of convergence of $\hat g_{\lambda_1,\lambda_{2,1}, \lambda_{2,2}}(y,x,w^*,h_n)$:
\begin{corollary}\label{supg}
If the conditions of Theorem \ref{Remainder} (ii) hold, then
\begin{align*}
    &\quad\sup_{(y,x,w^*)\in \mathbb{S}_{(Y,X,W^*)}}\left|\hat g_{\lambda_1,\lambda_{2,1}, \lambda_{2,2}}(y,x,w^*,h_n)-g_{\lambda_1,\lambda_{2,1}, \lambda_{2,2}}(y,x,w^*)\right|\\
    &=O_p\left(\epsilon_{n,\lambda_1}\right)+O_p\left(\tilde\epsilon_{n,\lambda_1,\lambda_{2,1},\lambda_{2,2}}\right).
\end{align*}
where 
\begin{footnotesize}
\begin{align*}
    &\epsilon_{n,\lambda_1}\equiv\left(h_1^{-1}\right)^{1+\gamma_{\phi}+\lambda_1}\exp\left(\alpha_{\phi}\bar\zeta_K^{\beta_{\phi}}\left(h_1^{-1}\right)^{\beta_{\phi}}\right)\\
    &\tilde\epsilon_{n,\lambda_1,\lambda_{2,1},\lambda_{2,2}}\\
    &\equiv n^{-1/2}\max\left\{(h_1^{-1})^{1+\gamma_*},(h_{2,1}^{-1})^{1+\lambda_{2,1}}(h_{2,2}^{-1})^{1+\lambda_{2,2}}\right\}(h_1^{-1})^{1-\gamma_2+\gamma_{\phi}+\lambda_1}\\
    &\quad\quad\quad\quad\quad\quad\quad\quad\quad\quad\quad\quad\quad\quad\quad\quad\times\exp\left(\left(\alpha_{\phi}\mathbf{1}\{\beta_{\phi}=\beta_2\}-\alpha_2\right)\left(h_1^{-1}\right)^{\beta_2}\right).
\end{align*}
\end{footnotesize}
\end{corollary}
The uniform (over the whole support) rate of convergence of the kernel estimators of the density or its derivatives has been considered in \cite{andrews1995nonparametric}, \cite{hansen2008uniform} and \cite{schennach2012local}\footnote{\cite{schennach2012local} studies the case both with and without measurement errors. Their Theorem 3.2 corresponds to the no-measurement error case, and Corollary 4.7 corresponds to the case with measurement errors.}. \cite{hansen2008uniform} obtained a faster rate than \cite{andrews1995nonparametric} and \cite{schennach2012local} (Theorem 3.2, for the no-measurement error case), but requires more assumptions on the kernel functions (their Assumption 1 and 3), including bounded support or an integrable tail,  which, unfortunately, are not satisfied by infinite order kernels. \cite{andrews1995nonparametric}'s conclusion is stated assuming a kernel with finite order. Infinite order kernels are not essential when there is no measurement error (like in \cite{andrews1995nonparametric} and \cite{hansen2008uniform}), but are especially advantageous when there are measurement errors. If infinite order kernels are used, only the smoothness of the functions, but not the order of the kernels will affect the rate of convergence. Corollary \ref{supg} is similar to Corollary 4.7 in \cite{schennach2012local}.\\
\indent With the uniform rate of convergence of $\hat g$, I can then derive the uniform rate of convergence of the plug-in estimators of the structural derivative $\rho(y,x)$ and the \textbf{WLAR}, denoted as $\widehat{\rho(y,x)}$ and $\widehat{\mathbf{WLAR}(x)}$, respectively. 
\begin{theorem}\label{unif convergent rate} Suppose that $\left\{Y_j, X_j, W^*_j, \Delta W_{1,j}, W_{2,j}\right\}$ is an IID sequence satisfying the conditions of Corollary \ref{supg} with $\lambda_1,\lambda_{2,1},\lambda_{2,2}\in\{0,1\}$ and $\max\{\lambda_1,\lambda_{2,1},\lambda_{2,2}\}\leq 1$. Suppose in addition, Assumption \ref{pospartialh} and \ref{Fepsilon|eta} hold. \\
(i) Define $\mathbb S_{\tau_n}\equiv \left\{(y,x,w^*)\in\mathbb{S}_{(Y,X,W^*)}: \left|\frac{\partial F_{X\mid W^*=w^*}(x)}{\partial w^*}\right|>\tau_n \text{ and } f_{Y,X,W^*}(y,x, w^*)>\tau_n\right\}$. Then
\begin{align*}
    &\sup_{(y,x,w^*)\in \mathbb  S_{\tau_n}}\left|\widehat{\rho(y,x)}-\rho(y,x)\right|\leq O_p\left(\tilde\epsilon_{n,0,0,1}\right)+\frac{O_p\left(\epsilon_{n,1}\right)+O_p\left(\tilde\epsilon_{n,1,0,0}\right)}{\tau_n^2}.
\end{align*}
and there exists $\{\tau_n\}$ such that $\tau_n>0, \tau_n\rightarrow 0$ as $n\rightarrow\infty$, and 
\begin{align*}
   &\sup_{(y,x,w^*)\in \mathbb  S_{\tau_n}}\left|\widehat{\rho(y,x)}-\rho(y,x)\right|= o_p(1).
\end{align*}
(ii) 
\begin{align*}
    \sup_{x\in\mathbb{S}_X}\left|\widehat{\mathbf{WLAR}(x)}-\mathbf{WLAR}(x)\right|\leq O_p\left(\epsilon_{n,1}\right)+O_p\left(\tilde\epsilon_{n,0,0,1}\right)+O_p\left(\tilde\epsilon_{n,1,0,0}\right).
\end{align*}
\end{theorem}
\begin{proof}
See Appendix B.
\end{proof}
\indent In Appendix C, I state an asymptotic normality result for $\hat g_{\lambda_1,\lambda_{2,1}, \lambda_{2,2}} (y,x,w^*,h_n)$ and $\widehat{\rho(y,x)}$. The proof of them requires a lower bound on $\Omega_{\lambda_1,\lambda_{2,1}, \lambda_{2,2}}(y,x,w^*.h_n)$ relative to $B_{\lambda_1,\lambda_{2,1}, \lambda_{2,2}}(y,x,w^*.h_n)$ and $R_{\lambda_1,\lambda_{2,1}, \lambda_{2,2}}(y,x,w^*.h_n)$. The assumptions to ensure the lower bound are stated at a high level, and more primitive sufficient conditions need to be derived.

\section{Monte Carlo Simulations} \label{Monte Carlo}
In this section, I conduct some Monte Carlo simulation exercises to study the finite sample performance of my proposed estimators. I consider two simulation designs. Design 1 is when both of the equations are linear, and Design 2 is when the equations are nonlinear.\\
\textbf{Design 1}
\begin{align*}
    &Y = 0.25X+0.25\epsilon\\
    &X = W^*+\eta
\end{align*}
\textbf{Design 2}
\begin{align*}
    &Y = \ln\left(\exp(X+\epsilon)+1\right)\\
    &X = \frac{3^3}{4^4}\frac{W^{*4}}{(-\eta)^3}
\end{align*}
In both cases, $\epsilon$ and $\eta$ are correlated through a common component $\theta$: $\epsilon = \theta + \epsilon_1$, and $\eta = \theta + \eta_1$, where $\theta$, $\epsilon_1$ and $\eta_1$ are mutually independent, and are independent with $W^*$.
In both designs, I have two error-laden measurements $W_1 = W^*+\Delta W_1, W_2 = W^*+\Delta W_2$, where $\Delta W_1$ and $\Delta W_2$ are independent and they are independent with all other variables. The distributions of variables in both designs are listed in Table \ref{dist of var} below:
% Table generated by Excel2LaTeX from sheet 'Sheet4'
\begin{table}[htbp]
  \centering
  \caption{Distributions of Variables}
  \label{dist of var}
    \begin{tabular}{lll}
    \hline
    \hline
    Variables & Design 1 (linear) & Design 2 (nonlinear)\\
    \hline
    $W^*$ & $N(0, \sqrt{0.5})$ & $N(6, 1)$ \\
    $\theta$ & $N(0, 0.5)$ & $N(-3, \sqrt{0.5})$ \\
    $\epsilon_1$ & $N(0,\sqrt{0.75})$ & $N(3, \sqrt{0.5})$ \\
    $\eta_1$ & $N(0,0.5)$ & $N(-3, \sqrt{0.5})$ \\
    $\epsilon$ & $N(0,1)$ & $N(0, 1)$ \\
    $\eta$ & $N(0,\sqrt{0.5})$ & $N(-6, 1)$ \\
    $\Delta W_1$ & $N(0, \sqrt{0.5})$ & $N(0, \sqrt{0.5})$ \\
    $\Delta W_2$ & $\chi^2(2)-2$ & $\chi^2(2)-2$ \\
    \hline
    \hline
    \end{tabular}%
\end{table}%

In the estimation, I use the following flat-top kernel\footnote{I may use different flat-top kernels for $G_X$ and $G_Y$. For simplicity I use the same flat-top kernel.} proposed by \cite{politis1999multivariate}:
\begin{align*}
    K(x)= \frac{\sin^2(2\pi u)-\sin^2(\pi u)}{\pi^2u^2}
\end{align*}
I use the sample size of 500 and replicate the estimation 500 times. \\
\indent First I show the performance of my proposed method versus two other methods for estimating the structural derivative: (1) 2SLS using the error-laden measurement $W_2$ as IV; (2) a plug-in estimator replacing $W^*$ with $W_2$ in my identification equation (\ref{pmpxbar}). Method (1) will be valid under the linear design (Design 1), since $W_2$, although error-laden, still satisfies the exclusion restriction and the rank condition for linear IV estimation. However, because of its misspecification of the model, it won't be valid under the nonlinear design (Design 2). Method (2) won't be valid under either the linear or nonlinear design, as the estimator will converge to a population value that is not equal to the right-hand side of equation (\ref{pmpxbar}) in general. I estimate  $\rho(y,x)$ evaluated at $y=0, x=0, w^*=0.7$ for Design 1, and $y=0.6, x=0.6, w^*=7$ for Design 2, yielding a true value of $0.25$ and $0.4512$, respectively. There are three bandwidths used in the estimation: $h_1$, $h_{21}$, $h_{22}$. The latter two correspond to estimating the unknown distribution of observed variables $X$ and $Y$. I thus choose the values of $h_{21}$ and $h_{22}$ by cross-validation (i.e. minimizing the estimated MISE of $\hat f_{Y,X}(y,x)$). For the bandwidth $h_1$, I scan a set of values ranging from 0.5 to 3 for my proposed estimator and a set of values ranging from 0.5 to 6 for the method that uses the error-laden measurement. Table \ref{my rho}, \ref{using W2 rho}, and \ref{2sls} show the MSE, VAR, and the absolute value of BIAS of my proposed estimator, Method (2), and Method (1), respectively. \\
\indent From Table \ref{using W2 rho}, one can see that the bias from the error-contaminated estimator does not shrink toward zero as bandwidth decreases. Comparing results from Table \ref{my rho} and Table \ref{using W2 rho}, one can see that my estimator also gives smaller variances than the error-contaminated method. As a result, my estimator performs better than the error-contaminated method in terms of MSE under both the linear and the nonlinear designs. Comparing Table \ref{my rho} with Table \ref{2sls}, one can see that under the linear design, both my estimator and 2SLS work well except that my estimator has a slightly larger variance. Under the nonlinear design, the bias and MSE of 2SLS are much larger than my estimator, which aligns with the theory.\\
\indent Table \ref{my rho consistency} below shows the Monte Carlo simulation results of my proposed estimator as a function of the sample size. The value of $h_{21}$ and $h_{22}$ are selected by cross-validation for the respective sample sizes and the optimal values of $h_1$ are selected by minimizing the MSE in the corresponding set of values the same as when $N=500$. The MSE, VAR, and Bias decrease as the sample size increases, which is in accordance with the theory.
\begin{landscape}
% Table generated by Excel2LaTeX from sheet 'Sheet5'
\begin{table}[p]
    \footnotesize
  \centering
  \caption{Monte Carlo simulation results of my proposed estimator for $\rho(y,x)$}
    \label{my rho}
    \begin{tabular}{cccccc|cccccc}
    \hline
    \hline
    \multicolumn{6}{c|}{Design 1 (linear)}        & \multicolumn{6}{c}{Design 2 (nonlinear)} \\
    $h_1$ & $h_{21}$ & $h_{22}$ & MSE   & VAR   & abs(BIAS) & $h_1$ & $h_{21}$ & $h_{22}$ & MSE   & VAR   & abs(BIAS) \\
    \hline
    \multicolumn{1}{l}{0.5} & \multicolumn{1}{l}{1.05} & \multicolumn{1}{l}{2.92} & 0.01072 & 0.01072 & 0.00009 & \multicolumn{1}{l}{0.5} & \multicolumn{1}{l}{2.09} & \multicolumn{1}{l}{0.88} & 0.31292 & 0.30489 & 0.08961 \\
    \multicolumn{1}{l}{0.75} & \multicolumn{1}{l}{1.05} & \multicolumn{1}{l}{2.92} & 0.00766 & 0.00766 & 0.00050 & \multicolumn{1}{l}{0.75} & \multicolumn{1}{l}{2.09} & \multicolumn{1}{l}{0.88} & 0.17724 & 0.16339 & 0.11769 \\
    \multicolumn{1}{l}{1} & \multicolumn{1}{l}{1.05} & \multicolumn{1}{l}{2.92} & 0.00728 & 0.00728 & 0.00138 & \multicolumn{1}{l}{1} & \multicolumn{1}{l}{2.09} & \multicolumn{1}{l}{0.88} & 0.15425 & 0.13973 & 0.12050 \\
    \multicolumn{1}{l}{1.25} & \multicolumn{1}{l}{1.05} & \multicolumn{1}{l}{2.92} & 0.00738 & 0.00738 & 0.00253 & \multicolumn{1}{l}{1.25} & \multicolumn{1}{l}{2.09} & \multicolumn{1}{l}{0.88} & 0.14824 & 0.13360 & 0.12098 \\
    \multicolumn{1}{l}{1.5} & \multicolumn{1}{l}{1.05} & \multicolumn{1}{l}{2.92} & 0.00755 & 0.00754 & 0.00328 & \multicolumn{1}{l}{1.5} & \multicolumn{1}{l}{2.09} & \multicolumn{1}{l}{0.88} & 0.14611 & 0.13129 & 0.12175 \\
    \multicolumn{1}{l}{1.75} & \multicolumn{1}{l}{1.05} & \multicolumn{1}{l}{2.92} & 0.00772 & 0.00771 & 0.00376 & \multicolumn{1}{l}{1.75} & \multicolumn{1}{l}{2.09} & \multicolumn{1}{l}{0.88} & 0.14540 & 0.13034 & 0.12273 \\
    \multicolumn{1}{l}{2} & \multicolumn{1}{l}{1.05} & \multicolumn{1}{l}{2.92} & 0.00789 & 0.00787 & 0.00408 & \multicolumn{1}{l}{2} & \multicolumn{1}{l}{2.09} & \multicolumn{1}{l}{0.88} & 0.14540 & 0.13008 & 0.12380 \\
    \multicolumn{1}{l}{2.25} & \multicolumn{1}{l}{1.05} & \multicolumn{1}{l}{2.92} & 0.00806 & 0.00804 & 0.00431 & \multicolumn{1}{l}{2.25} & \multicolumn{1}{l}{2.09} & \multicolumn{1}{l}{0.88} & 0.14581 & 0.13023 & 0.12482 \\
    \multicolumn{1}{l}{2.5} & \multicolumn{1}{l}{1.05} & \multicolumn{1}{l}{2.92} & 0.00822 & 0.00820 & 0.00447 & \multicolumn{1}{l}{2.5} & \multicolumn{1}{l}{2.09} & \multicolumn{1}{l}{0.88} & 0.14640 & 0.13060 & 0.12569 \\
    \multicolumn{1}{l}{2.75} & \multicolumn{1}{l}{1.05} & \multicolumn{1}{l}{2.92} & 0.00836 & 0.00834 & 0.00461 & \multicolumn{1}{l}{2.75} & \multicolumn{1}{l}{2.09} & \multicolumn{1}{l}{0.88} & 0.14707 & 0.13109 & 0.12640 \\
    \multicolumn{1}{l}{3} & \multicolumn{1}{l}{1.05} & \multicolumn{1}{l}{2.92} & 0.00851 & 0.00849 & 0.00474 & \multicolumn{1}{l}{3} & \multicolumn{1}{l}{2.09} & \multicolumn{1}{l}{0.88} & 0.14777 & 0.13163 & 0.12703 \\
    \multicolumn{6}{c|}{Optimal}                  & \multicolumn{6}{c}{Optimal} \\
    $h_1$ & $h_{21}$ & $h_{22}$ & MSE   & VAR   & BIAS  & $h_1$ & $h_{21}$ & $h_{22}$ & MSE   & VAR   & BIAS \\
    1     & 1.05 & 2.92 & 0.00728 & 0.00728 & 0.00138 & 1.75  & 2.09 & 0.88 & 0.14540 & 0.13034 & 0.12273 \\
    \hline
    \hline
    \end{tabular}%
\end{table}%
\end{landscape}

\begin{landscape}
% Table generated by Excel2LaTeX from sheet 'Sheet5'
\begin{table}[p]
    \footnotesize
    \centering
    \caption{Monte Carlo simulation results of plugging in error-laden $W_2$ to (\ref{pmpxbar}) to estimate $\rho(y,x)$ (Method (2) above)}
    \label{using W2 rho}
    \begin{tabular}{cccccc|cccccc}
    \hline
    \hline
    \multicolumn{6}{c|}{Design 1 (linear)}        & \multicolumn{6}{c}{Design 2 (nonlinear)} \\
    $h_1$ & $h_{21}$ & $h_{22}$ & MSE   & VAR   & abs(BIAS) & $h_1$ & $h_{21}$ & $h_{22}$ & MSE   & VAR   & abs(BIAS) \\
    \hline
    \multicolumn{1}{l}{0.5} & \multicolumn{1}{l}{1.05} & \multicolumn{1}{l}{2.92} & 831.49239 & 827.49615 & 1.99906 & \multicolumn{1}{l}{0.5} & \multicolumn{1}{l}{2.09} & \multicolumn{1}{l}{0.88} & 342390.30868 & 342380.03811 & 3.20477 \\
    \multicolumn{1}{l}{1} & \multicolumn{1}{l}{1.05} & \multicolumn{1}{l}{2.92} & 15.36480 & 15.18324 & 0.42611 & \multicolumn{1}{l}{1} & \multicolumn{1}{l}{2.09} & \multicolumn{1}{l}{0.88} & 3268.04002 & 3264.78862 & 1.80316 \\
    \multicolumn{1}{l}{1.5} & \multicolumn{1}{l}{1.05} & \multicolumn{1}{l}{2.92} & 1609.24991 & 1605.98439 & 1.80708 & \multicolumn{1}{l}{1.5} & \multicolumn{1}{l}{2.09} & \multicolumn{1}{l}{0.88} & 795.12130 & 795.09105 & 0.17391 \\
    \multicolumn{1}{l}{2} & \multicolumn{1}{l}{1.05} & \multicolumn{1}{l}{2.92} & 981.89810 & 980.63434 & 1.12417 & \multicolumn{1}{l}{2} & \multicolumn{1}{l}{2.09} & \multicolumn{1}{l}{0.88} & 1113.49474 & 1113.43014 & 0.25417 \\
    \multicolumn{1}{l}{2.5} & \multicolumn{1}{l}{1.05} & \multicolumn{1}{l}{2.92} & 2739.03247 & 2729.13653 & 3.14578 & \multicolumn{1}{l}{2.5} & \multicolumn{1}{l}{2.09} & \multicolumn{1}{l}{0.88} & 7439.31131 & 7399.31171 & 6.32452 \\
    \multicolumn{1}{l}{3} & \multicolumn{1}{l}{1.05} & \multicolumn{1}{l}{2.92} & 48.28533 & 48.27034 & 0.12243 & \multicolumn{1}{l}{3} & \multicolumn{1}{l}{2.09} & \multicolumn{1}{l}{0.88} & 2589.00861 & 2583.05235 & 2.44054 \\
    \multicolumn{1}{l}{3.5} & \multicolumn{1}{l}{1.05} & \multicolumn{1}{l}{2.92} & 227.69069 & 227.32253 & 0.60677 & \multicolumn{1}{l}{3.5} & \multicolumn{1}{l}{2.09} & \multicolumn{1}{l}{0.88} & 7367.37049 & 7363.00243 & 2.08999 \\
    \multicolumn{1}{l}{4} & \multicolumn{1}{l}{1.05} & \multicolumn{1}{l}{2.92} & 326.60891 & 325.63719 & 0.98576 & \multicolumn{1}{l}{4} & \multicolumn{1}{l}{2.09} & \multicolumn{1}{l}{0.88} & 4810.95117 & 4805.88266 & 2.25133 \\
    \multicolumn{1}{l}{4.5} & \multicolumn{1}{l}{1.05} & \multicolumn{1}{l}{2.92} & 29.37044 & 29.30256 & 0.26055 & \multicolumn{1}{l}{4.5} & \multicolumn{1}{l}{2.09} & \multicolumn{1}{l}{0.88} & 12005.60726 & 11951.43663 & 7.36007 \\
    \multicolumn{1}{l}{5} & \multicolumn{1}{l}{1.05} & \multicolumn{1}{l}{2.92} & 7904.23502 & 7887.80380 & 4.05354 & \multicolumn{1}{l}{5} & \multicolumn{1}{l}{2.09} & \multicolumn{1}{l}{0.88} & 863.16882 & 862.91399 & 0.50480 \\
    \multicolumn{1}{l}{5.5} & \multicolumn{1}{l}{1.05} & \multicolumn{1}{l}{2.92} & 16.49962 & 16.48066 & 0.13769 & \multicolumn{1}{l}{5.5} & \multicolumn{1}{l}{2.09} & \multicolumn{1}{l}{0.88} & 6884.05497 & 6883.55213 & 0.70911 \\
    \multicolumn{1}{l}{6} & \multicolumn{1}{l}{1.05} & \multicolumn{1}{l}{2.92} & 54.17185 & 54.16807 & 0.06147 & \multicolumn{1}{l}{6} & \multicolumn{1}{l}{2.09} & \multicolumn{1}{l}{0.88} & 844.96901 & 843.85308 & 1.05637 \\
    \multicolumn{6}{c|}{Optimal}                  & \multicolumn{6}{c}{Optimal} \\
    $h_1$ & $h_{21}$ & $h_{22}$ & MSE   & VAR   & BIAS  & $h_1$ & $h_{21}$ & $h_{22}$ & MSE   & VAR   & BIAS \\
    1     & 1.05 & 2.92 & 15.36480 & 15.18324 & 0.42611 & 1.5   & 2.09 & 0.88 & 795.12130 & 795.09105 & 0.17391 \\
    \hline
    \hline
    \end{tabular}%
\end{table}%
\end{landscape}

\begin{landscape}
% Table generated by Excel2LaTeX from sheet 'Sheet5'
\begin{table}[p]
\footnotesize
  \centering
  \caption{Monte Carlo simulation results of using 2SLS and error-laden $W_2$ to estimate $\rho(y,x)$ (Method (1) above)}
  \label{2sls}
    \begin{tabular}{ccc|ccc}
    \hline
    \hline
    \multicolumn{3}{c|}{Design 1 (linear)} & \multicolumn{3}{c}{Design 2 (nonlinear)} \\
    MSE   & VAR   & abs(BIAS) & MSE   & VAR   & abs(BIAS) \\
    \hline
    0.00267 & 0.00267 & 0.00190 & 1.11718 & 0.00147 & 1.05627 \\
    \hline
    \hline
    \end{tabular}%
\end{table}%

% Table generated by Excel2LaTeX from sheet 'Sheet6'
\begin{table}[p]
\footnotesize
  \centering
  \caption{Monte Carlo simulation results of my $\hat\rho(y,x)$ as a function of the sample size}
  \label{my rho consistency}
    \begin{tabular}{ccccccc|ccccccc}
    \hline
    \hline
    \multicolumn{7}{c|}{Design 1 (linear)}                & \multicolumn{7}{c}{Design 2 (nonlinear)} \\
    $N$   & $h_1$ & $h_{21}$ & $h_{22}$ & MSE   & VAR   & abs(BIAS) & $N$   & $h_1$ & $h_{21}$ & $h_{22}$ & MSE   & VAR   & abs(BIAS) \\
    \hline
    500   & 1     & 1.05  & 2.92  & 0.00728 & 0.00728 & 0.00138 & 500   & 1.75  & 2.09  & 0.88  & 0.14540 & 0.13034 & 0.12273 \\
    1000  & 1     & 0.97  & 2.61  & 0.00354 & 0.00353 & 0.00359 & 1000  & 1.5   & 1.58  & 0.87  & 0.08341 & 0.06822 & 0.12323 \\
    5000  & 1.25  & 0.82  & 2.05  & 0.00102 & 0.00102 & 0.00026 & 5000  & 1.75  & 0.93  & 0.65  & 0.03053 & 0.03048 & 0.00708 \\
    \hline
    \hline
    \end{tabular}%
\end{table}%
\end{landscape}

Next, I study the performance of my estimator of the weighted LAR. For Design 1 (linear), I evaluate the WLAR at $x=0$ and set the weighting function to be $\omega(0,w^*,\epsilon)=\omega_1$ if $\epsilon\in\left[q_{0,w^*}(0.25),q_{0,w^*}(0.35)\right]$ and $w^*\in\left[0.70, 0.90\right]$. The constant $\omega_1$ is selected such that $\int_{0.7}^{0.9}\int_{q_{0,w^*}(0.25)}^{q_{0,w^*}(0.35)}\omega_1f_{\epsilon,W^*|X=0}(\epsilon,w^*)d\epsilon d w^*=1$. For Design 2 (nonlinear), I evaluate the WLAR at $x=0.6$ and set the weighting function to be $\omega(0.6,w^*,\epsilon)=\omega_2$ if $\epsilon\in\left[q_{0.6,w^*}(0.25),q_{0.6,w^*}(0.35)\right]$ and $w^*\in\left[6, 6.23\right]$. The constant $\omega_2$ is selected to ensure that $\int_{6}^{6.23}\int_{q_{0.6,w^*}(0.25)}^{q_{0.6,w^*}(0.35)}\omega_2f_{\epsilon,W^*|X=0.6}(\epsilon,w^*)d\epsilon d w^*=1$. These choices of the weighting functions yield a true value of the WLAR of 0.25 and 0.5032 for Design 1 and 2, respectively. In the estimation, I use the optimal values of $h_1$ coming from Table \ref{my rho}, and the same values of $h_{21}$ and $h_{22}$ as in Table \ref{my rho}. Table \ref{my WLAR} below shows the results of my estimators of the WLAR under both designs. The performance of my estimator of the WLAR is comparable to the performance of my estimator for the point-wise structural derivative $\rho(y,x)$ (shown in Table \ref{my rho}). This shows evidence that my estimator could not only perform well at single points but also could maintain good performance in a global sense. The supports of the weighting functions I used here are not large. This is mainly due to computation constraints. One direction of future work is to look at the performance of my estimator for WLAR when the support of the weighting function is larger. This will give us a better understanding of the global performance of my estimator.
% Table generated by Excel2LaTeX from sheet 'Sheet7'
\begin{table}[htbp]
  \centering
  \caption{Monte Carlo simulation results of my proposed estimator for the WLAR}
    \label{my WLAR}%
    \begin{tabular}{lcccrrrr}
    \hline
    \hline
          & $h_1$ & $h_{21}$ & $h_{22}$ & \multicolumn{1}{c}{true value} & \multicolumn{1}{c}{MSE} & \multicolumn{1}{c}{VAR} & \multicolumn{1}{c}{abs(BIAS)} \\
    \hline
    Design 1 (linear) & 1     & 1.05  & 2.92  & 0.25  & 0.00779 & 0.00777 & 0.00405 \\
    Design 2 (nonlinear) & 1.75  & 1.75  & 2.09  & 0.5032 & 0.15082 & 0.13829 & 0.11193 \\
    \hline
    \hline
    \end{tabular}%
\end{table}%

\indent The Monte Carlo simulation exercise shown in this section is very limited. To have a comprehensive examination of the performance of my estimators, there are several future directions that I can work on. First, in my current exercise, the distribution of the measurement error $\Delta W_2$ is set to be $\chi^2$, which is ordinarily smooth. I can further explore the performance of my estimators under different distributions of the measurement error $\Delta W_2$ and other variables. Second, when estimating the WLAR, I set the support for the weighting function to be small because of the computation constraint. With more time and more computation power, I can look at the performance of my estimator when the support of the weighting function in the WLAR is larger, i.e. the WLAR is taking the average marginal effect over a larger population. Third, the way I select $h_1$ is by minimizing the actual MSE, i.e., the mean squared difference between my estimator and the true value of the parameter. This is infeasible in practice when one deals with real-world data. It will be helpful if I could propose a feasible way to select the bandwidth. 
\section{Conclusion}
I study the nonparametric identification and estimation of the nonseparable triangular equations model when the instrument variable $W^*$ is mismeasured. I don't assume linearity or separability of the functions governing the relationship between observables and unobservables. To deal with the challenges caused by the co-existence of the measurement error and nonseparability, I first employ the deconvolution technique developed in the measurement error literature (\cite{schennach2004estimation},\cite{schennach2004nonparametric}) to identify the joint distribution of $Y, X, W^*$ using two error-laden measurements of $W^*$. I then recover the structural derivative of the function of interest and the ``Local Average Response'' (LAR) via the ``unobserved instrument'' approach in \cite{matzkin2016independence}. Based on the constructive identification results, I propose plug-in nonparametric estimators for these parameters and derive their uniform rates of convergence. I also conducted limited Monte Carlo exercises to show the finite sample performance of my estimators. \\
\indent I recognize that there are some important future directions for this paper. First, in the main text, I only demonstrated the uniform rate of convergence of the estimators. The appendix contains proofs of their asymptotic normality, which rely on high-level assumptions. To further enhance my results, it would be beneficial to derive more primitive sufficient conditions for these assumptions. Second, as discussed in Section \ref{Monte Carlo}, to conduct a comprehensive examination of the performance of my estimators, additional Monte Carlo studies are necessary. It will also be extremely useful if I could propose a feasible way to select the bandwidth $h_1$. Last, it will be an interesting exercise to apply my proposed method to real-world data.

\newpage
\section{Appendix}
\subsection{Appendix A: Useful Lemmas}
Before stating the lemmas, I define some notations.\\
\indent Let $Y$ be a random variable, $Z$ be a random vector of dimension $2$. Denote $Z_1\equiv X, Z_2\equiv W$. Let $\mathbb{S}_{(Y,Z)}$ be the support of $(Y,Z)$, and define  \begin{align*}
     \mathbb S_\tau\equiv\left\{(y,z)\in\mathbb{S}_{(Y,Z)}: \left|\frac{\partial F_{X\mid W=w}(x)}{\partial w}\right|\geq\tau, \text{ and } f_{Y,X,W}(y,x, w)>\tau\right\}
\end{align*}
 and $\mathcal{S}_\tau\equiv\left\{(\delta, z)\in [0,1]\times\mathbb{S}_Z: \left|\frac{\partial F_{X\mid W=w}(x)}{\partial w}\right|\geq\tau \text{ and } f_{Y,X,W}(v(x,w,\delta), x, w)>\tau\right\}$. Let $\bar{\mathbb M}$ be a compact subset of the support of $(Y,X,W)$. Define the mapping $\mathbf{invm}$ as $\mathbf{invm}(Y,X,W)\equiv(m^{-1}(Y,X),X,W)$, where $m^{-1}$ is as defined in (\ref{introdmdx}). Define the mapping $\mathbf{invs}$ as $\mathbf{invs}(Y,X,W)\equiv(s^{-1}(r(X,W),m^{-1}(X,Y)), X, W)$, where $s^{-1}$ is the inverse of the $s$ function defined in Lemma \ref{deltaindependent} with respect to its second argument. Let $\mathbb{M}\equiv \mathbf{invm}(\bar{\mathbb M})$, and $\mathcal M=\mathbf{invs}(\bar{\mathbb M})$. By Assumption \ref{aspm&h} and \ref{cpctcontdiff}, $\mathbb{M}$ and $\mathcal M$ are also compact, and there exist $\tau>0$ such that $\bar{\mathbb M}\subset \mathbb S_\tau$ and $\mathcal{M}\subset\mathcal{S}_\tau$.\\
 \indent For any function $g: \mathbb{R}^3\rightarrow\mathbb{R}$, define $\tilde{g}(z)\equiv\int g(y, z) d y$, $\tilde g_{Z_2}(z_2)\equiv\int \tilde g(z)dx$, $G(y,z)\equiv\int_{-\infty}^y\int_{-\infty}^z g(s,t)dtds$, and if $\tilde g(z)\neq 0$, define $G_{Y|Z=z}(y)=\left(\int_{-\infty}^yg(s,z)ds\right)/\tilde g(z)$; if $\tilde g_{Z_2}(z_2)\neq 0$, define $G_{Z_1\mid Z_2=z_2}(z_1)=\left(\int _{-\infty}^{z_1}\tilde g(s,z_2)ds\right)/\tilde g_{Z_2}(z_2)$. Let $\mathrm{F}$ denote a set of twice continuously differentiable functions $g: \mathbb{R}^3\rightarrow \mathbb{R}$ such that $g$ vanishes outside of $\mathbb S_{(Y,Z)}$. Let $f$ denote the joint density of $(Y,Z)$. Assume that $f$ belongs to $\mathrm F$. For any value $(y,z)\in \mathbb{S}_{Y,Z}$, and any value $\delta\in[0,1]$, define the functionals $\Lambda(\cdot)$, $\alpha(\cdot)$, $\Phi_{(j)}(\cdot)$,  $\Psi_{(j)}(\cdot)$ and $\tilde\Lambda(\cdot)$, $\tilde\Psi(\cdot)$ on $\mathrm{F}$ by $\Lambda(g)\equiv G_{Y|Z=z}(y)$, $\alpha(g)\equiv G_{Y|Z=z}^{-1}(\delta)$, $ \Phi_{(j)}(g)\equiv \frac{\partial G_{Y|Z=z}^{-1}(\delta)}{\partial z_j}$, $ \Psi_{(j)}(g)\equiv- \left.\frac{\partial G_{Y|Z=z}(y)}{\partial z_j}\right/\frac{\partial G_{Y|Z=z}(y)}{\partial y}$ and $\tilde\Lambda(g)\equiv G_{X\mid W=w}(x)$, $\tilde\Psi(g)\equiv- \left.\frac{\partial G_{X|W=w}(x)}{\partial w}\right/\frac{\partial G_{X|W=w}(x)}{\partial x}$.
 For simplicity, I leave the argument $(z, y, \delta)$ implicit.
\begin{lemma}\label{file1}
	\begin{align*}
	    \Phi_{(j)}(f) &= \frac{\delta \frac{\partial \tilde f(z)}{\partial z_j}-\int_{-\infty}^{\alpha(f)}\frac{\partial f(y,z)}{\partial z_j}dy}{f(\alpha(f),z)}\\
	    \Psi_{(j)}(f) &= \frac{\Lambda(f) \frac{\partial \tilde f(z)}{\partial z_j}-\int_{-\infty}^{y}\frac{\partial f(s,z)}{\partial z_j}ds}{f(y,z)}\\
	    \tilde\Psi(f) &= \frac{\tilde\Lambda(f) \frac{\partial \tilde f_{W}(w)}{\partial w}-\int_{-\infty}^{x}\frac{\partial \tilde f(s,w)}{\partial w}ds}{\tilde f(x,w)}.
	\end{align*}
\end{lemma}

\begin{proof}
By definition, 
\begin{align*}
    \delta &= F_{Y|Z=z}(\alpha(f))\\
    &=\frac{\int_{-\infty}^{\alpha(f)}f(y,z)dy}{\int_{-\infty}^\infty f(y,z)dy}\\
    \implies\delta \tilde f(z) &=\int_{-\infty}^{\alpha(f)}f(y,z)dy
\end{align*}
Taking derivatives on both sides with respect to $z_j$ yields
\begin{align*}
    &f(\alpha(f), z)\Phi_{(j)}(f)+\int_{-\infty}^{\alpha(f)}\frac{\partial f(y,z)}{\partial z_j}dy = \delta \frac{\partial \tilde f(z)}{\partial z_j}\\
    \implies &\Phi_{(j)}(f) = \frac{\delta \frac{\partial \tilde f(z)}{\partial z_j}-\int_{-\infty}^{\alpha(f)}\frac{\partial f(y,z)}{\partial z_j}dy}{f(\alpha(f),z)}.
\end{align*}
Then note that
\begin{align*}
    \frac{\partial G_{Y\mid Z=z}(x)}{\partial z_j}&=\frac{\left(\int_{-\infty}^y \frac{\partial f(s,z)}{\partial z_j}ds\right)\tilde f(z)-\left(\int_{-\infty}^y f(s,z)ds\right)\frac{\partial \tilde f(z)}{\partial z_j}}{\tilde f^2(z)}\\
    \frac{\partial G_{Y\mid Z=z}(x)}{\partial y}&=\frac{f(y,z)}{\tilde f(z)}.
\end{align*}
Then 
\begin{align*}
    \Psi_{(j)}(f) &\equiv -\left. \frac{\partial G_{Y\mid Z=z}(x)}{\partial z_j}\right/ \frac{\partial G_{Y\mid Z=z}(x)}{\partial y}\\
    &=\frac{F_{Y\mid Z=z}(y)\frac{\partial\tilde f(z)}{\partial z_j}-\int_{-\infty}^y \frac{\partial f(s,z)}{\partial z_j}ds}{f(y,z)}.
\end{align*}
The conclusion for $\tilde\Psi(f)$ follows by the same argument. 
\end{proof}

\begin{lemma}\label{file6} For any $h$ in $\mathrm{F}$ such that $\sup_{(y,z)\in\mathbb{S}_{(Y,Z)}}|h|$ is small enough, I have that
\begin{align*}
    \Lambda(f+h)-\Lambda(f) = D\Lambda(f;h)+R\Lambda(f;h)\\
    \tilde\Lambda(f+h)-\tilde\Lambda(f) = D\tilde\Lambda(f;h)+R\tilde\Lambda(f;h), 
\end{align*}
where
\begin{align*}
    D\Lambda(f;h) &= \frac{\int_{-\infty}^yh (s,z)ds-\tilde h(z) F_{Y|Z=z}(y)}{\tilde f(z)}\\
    R\Lambda(f;h) &= \left[\frac{\int_{-\infty}^yh(s,z) ds-\tilde h(z) F_{Y|Z=z}(y)}{\tilde f(z)}\right]\left[\frac{\tilde h(z)}{\tilde f(z)+\tilde h(z)}\right]\\
    D\tilde\Lambda(f;h) &= \frac{\int_{-\infty}^{x}\tilde h (s,w)ds-\tilde h_W(w) F_{X_1|W=w}(x)}{\tilde f_W(w)}\\
    R\tilde\Lambda(f;h) &= \left[\frac{\int_{-\infty}^{x}\tilde h(s,w) ds-\tilde h_W(w) F_{X_1|W=w}(x)}{\tilde f_W(w)}\right]\left[\frac{\tilde h_W(w)}{\tilde f_W(w)+\tilde h_W(w)}\right]
\end{align*}
and for some $a_1<\infty$,
\begin{align*}
    &\sup_{(y,z)\in\mathbb{S}_\tau}\left|D\Lambda(f;h)\right|\leq a_1 \sup_{(y,z)\in\mathbb{S}_{(Y,Z)}}|h|\\
    &\sup_{(y,z)\in\mathbb{S}_\tau}\left|R\Lambda(f;h)\right|\leq a_1 \sup_{(y,z)\in\mathbb{S}_{(Y,Z)}}|h|^2;
\end{align*}
for some $\tilde a_1<\infty$,
\begin{align*}
    &\sup_{(y,z)\in\mathbb{S}_\tau}\left|D\tilde\Lambda(f;h)\right|\leq \tilde a_1 \sup_{(y,z)\in\mathbb{S}_{(Y,Z)}}|h|\\
    &\sup_{(y,z)\in\mathbb{S}_\tau}\left|R\tilde\Lambda(f;h)\right|\leq \tilde a_1 \sup_{(y,z)\in\mathbb{S}_{(Y,Z)}}|h|^2.
\end{align*}
\end{lemma}

\begin{proof}
\begin{footnotesize}
    \begin{align*}
        \Lambda(f+h)-\Lambda(f) &= \frac{\int_{-\infty}^y (f(s,z)+h(s,z))ds}{\tilde f(z)+\tilde h(z)}-\frac{\int_{-\infty}^y f(s,z)ds}{\tilde f(z)}\\
        &=\frac{\int_{-\infty}^y h ds}{\tilde f+\tilde h}-\frac{\tilde h\int_{-\infty}^y f ds}{\tilde f(\tilde f+\tilde h)}\\
        &=\frac{\int_{-\infty}^yh ds-\tilde h F_{Y|Z=z}(y)}{\tilde f}+\left[\frac{\int_{-\infty}^yh ds-\tilde h F_{Y|Z=z}(y)}{\tilde f}\right]\left[\frac{\tilde h}{\tilde f+\tilde h}\right];\\
        \tilde\Lambda(f+h)-\tilde\Lambda(f) &=\frac{\int_{-\infty}^{x}\tilde h ds-\tilde h_W F_{X_1|W=w}(x)}{\tilde f_W}+\left[\frac{\int_{-\infty}^{x}\tilde h ds-\tilde h_W F_{X_1|W=w}(x)}{\tilde f_W}\right]\left[\frac{\tilde h_W}{\tilde f_W+\tilde h_W}\right].
    \end{align*}
\end{footnotesize}
Define
\begin{align*}
    D\Lambda(f;h) &= \frac{\int_{-\infty}^yh ds-\tilde h F_{Y|Z=z}(y)}{\tilde f}\\
    R\Lambda(f;h) &= \left[\frac{\int_{-\infty}^yh ds-\tilde h F_{Y|Z=z}(y)}{\tilde f}\right]\left[\frac{\tilde h}{\tilde f+\tilde h}\right]\\
    D\tilde\Lambda(f;h) &= \frac{\int_{-\infty}^{x}\tilde h ds-\tilde h_W F_{X_1|W=w}(x)}{\tilde f_W}\\
    R\tilde\Lambda(f;h) &= \left[\frac{\int_{-\infty}^{x}\tilde h ds-\tilde h_W F_{X_1|W=w}(x)}{\tilde f_W}\right]\left[\frac{\tilde h_W}{\tilde f_W+\tilde h_W}\right].
\end{align*}
By the boundedness of $\mathbb{S}_{(Y,Z)}$, there exist finite constants $C_X, C_Y$ such that for all $h\in\mathrm{F}$,
\begin{align*}
    &\sup_{z\in\mathbb{S}_Z}\left|\tilde h(z)\right| \leq\sup_{z\in\mathbb{S}_Z} \int \left|h(y,z)\right|dy\leq C_Y\sup_{(y,z)\in\mathbb{S}_{(Y,Z)}}|h|\\
    &\sup_{w\in\mathbb{S}_W}\left|\tilde h_W(w)\right| =\sup_{w\in\mathbb{S}_W} \int \left|\tilde h(x,w)\right|dx\leq C_{X}C_Y\sup_{(y,z)\in\mathbb{S}_{(Y,Z)}}|h|
\end{align*}
By the definition of set $\mathbb{S}_\tau$, I have that $\inf_{(y,z)\in\mathbb{S}_\tau}\tilde f(z)> \tau$ and that $\inf_{(y,z)\in\mathbb{S}_\tau}\tilde f_W(w)=\inf_{{(y,z)\in\mathbb{S}_\tau}}\int \tilde f(s, w)ds>\tau C_{X}$
Let $\epsilon_0\equiv \min \{\frac{\tau}{2C_Y},1\}$, For any $h$ in $\mathrm{F}$ that $\sup_{(y,z)\in\mathbb{S}_{(Y,Z)}}|h|\leq \epsilon_0$ , I have that $\sup_{z\in\mathbb{S}_Z}\left|\tilde h(z)\right|\leq \tau/2$, $\sup_{w\in\mathbb{S}_W}\left|\tilde h_W(w)\right|\leq C_{X}\tau/2$, so that $\inf_{(y,z)\in\mathbb{S}_\tau}  (\tilde f(z)+\tilde h(z))\geq \tau/2$, and $\inf_{(y,z)\in\mathbb{S}_\tau}  (\tilde f_W(w)+\tilde h_W(w))\geq C_{X}\tau/2$. Then there exists a finite constant $a_0$ such that
\begin{align*}
    &\sup_{(y,z)\in\mathbb{S}_\tau}\left|D\Lambda(f;h)\right|\leq \frac{a_0}{\tau} \sup_{(y,z)\in\mathbb{S}_{(Y,Z)}}|h|\\
    &\sup_{(y,z)\in\mathbb{S}_\tau}\left|R\Lambda(f;h)\right|\leq \frac{a_0}{\tau} \sup_{(y,z)\in\mathbb{S}_{(Y,Z)}}|h|\times \frac{2\sup_{(y,z)\in\mathbb{S}_{(Y,Z)}}|h|}{\tau}=\frac{2a_0}{\tau^2}\sup_{(y,z)\in\mathbb{S}_{(Y,Z)}}|h|^2.
\end{align*}
Hence let $a_1\equiv \max\{a_0/\tau, 2a_0/\tau^2\}$, I have that
\begin{align*}
    &\sup_{(y,z)\in\mathbb{S}_\tau}\left|D\Lambda(f;h)\right|\leq a_1 \sup_{(y,z)\in\mathbb{S}_{(Y,Z)}}|h|\\
    &\sup_{(y,z)\in\mathbb{S}_\tau}\left|R\Lambda(f;h)\right|\leq a_1 \sup_{(y,z)\in\mathbb{S}_{(Y,Z)}}|h|^2.
\end{align*}
Similarly, I can show that there exists a constant $\tilde a_1$ such that
\begin{align*}
    &\sup_{(y,z)\in\mathbb{S}_\tau}\left|D\tilde\Lambda(f;h)\right|\leq \tilde a_1 \sup_{(y,z)\in\mathbb{S}_{(Y,Z)}}|h|\\
    &\sup_{(y,z)\in\mathbb{S}_\tau}\left|R\tilde\Lambda(f;h)\right|\leq \tilde a_1 \sup_{(y,z)\in\mathbb{S}_{(Y,Z)}}|h|^2.
\end{align*}
\end{proof}

\begin{lemma}\label{file5} For any $h$ in $\mathrm{F}$ such that $\sup_{(y,z)\in\mathbb{S}_{(Y,Z)}} |h|$ is small enough, I have that
\begin{align*}
    \alpha(f+h)-\alpha(f) = D\alpha(f;h)+R\alpha(f;h), 
\end{align*}
where
\begin{align*}
    &D\alpha(f;h)\equiv \frac{\tilde{h}(z)\int^{\alpha(f) }  f(y, z) d y - \tilde{f}(z)\int^{\alpha(f)}   h(y, z) d y}{\tilde{f}(z)f(\alpha(f), z)}\\
    &R\alpha(f;h)\equiv -\left[\frac{\frac{\partial f(r_f',z)}{\partial y}\left(r_{f}-\alpha(f)\right)+h(r_h,z)}{f(r_f, z)+h\left(r_{h}, z\right)}\right]D\alpha(f;h)
\end{align*}
for some $r_f$ and $r_f'$ between $\alpha(f+h)$ and $\alpha(f)$ defined in the proof. Moreover, for some $0<a<\infty$,
\begin{align*}
    &\sup_{(\delta,z)\in \mathcal{S}_\tau}\left|D\alpha(f;h)\right|\leq a \sup_{(y,z)\in\mathbb{S}_{(Y,Z)}} |h|\\
    &\sup_{(\delta,z)\in \mathcal{S}_\tau}\left|R\alpha(f;h)\right|\leq a \sup_{(y,z)\in\mathbb{S}_{(Y,Z)}} |h|^2.
\end{align*}
\end{lemma}

\begin{proof}
First I show that for some $0<a_2<\infty$ and $a_1$ mentioned in Lemma \ref{file6},
\begin{align*}
    \sup_{(\delta,z)\in \mathcal{S}_\tau}\left|\alpha(f+h)-\alpha(f)\right|\leq \frac{2a_1a_2}{\tau} \sup_{(y,z)\in\mathbb{S}_{(Y,Z)}} |h|.
\end{align*}
By the Mean Value Theorem, there exist $r_1$ between $\alpha(f+h)$ and $\alpha(f)$, such that
\begin{align*}
    &\quad\left(F+H\right)_{Y|Z=z}\left(\alpha(f+h)\right)-\left(F+H\right)_{Y|Z=z}\left(\alpha(f)\right)\\
    &=(f+h)_{Y|Z=z}(r_1)\left(\alpha(f+h)-\alpha(f)\right)
\end{align*}
Hence since 
\begin{align*}
    &\quad\left(F+H\right)_{Y|Z=z}\left(\alpha(f+h)\right)\\
    &=\left(F+H\right)_{Y|Z=z}\left(\left(F+H\right)^{-1}_{Y|Z=z}(\delta)\right)=\delta=F_{Y|Z=z}\left(F^{-1}_{Y|Z=z}(\delta)\right)
\end{align*}
it follows that
\begin{align*}
    \alpha(f+h)-\alpha(f) &= \frac{F_{Y|Z=z}\left(F_{Y|Z=z}^{-1}(\delta)\right)-\left(F+H\right)_{Y|Z=z}\left(F^{-1}_{Y|Z=z}(\delta)\right)}{(f+h)_{Y|Z=z}(r_1)}.
\end{align*}
By the compactness of $\mathbb{S}_{Y,Z}$, for any $h$ in $\mathrm F$, there exist some finite constant $a_2$ such that $\sup_{z\in\mathbb{S}_{Z}} (\tilde f(z)+\tilde h(z))\leq a_2$. By the definition of $\mathbb{S}_\tau$, $\inf_{(y,z)\in\mathbb{S}_\tau} f(y,z)>\tau$. By similar argument as in Lemma \ref{file6}, for any $h$ in $\mathrm{F}$ such that $\sup_{(y,z)\in\mathbb{S}_{(Y,Z)}} |h|\leq \min\left\{\tau/2,1\right\}$, I have $\inf_{(y,z)\in \mathbb{S}_\tau} (f(y,z)+h(y, z))>\tau/2$. Then $\inf_{(y,z)\in \mathbb{S}_\tau} (f+h)_{Y|Z=z}(r_1)=\frac{\inf _{(y,z)\in \mathbb{S}_\tau} f(r_1,z)+h(r_1, z)}{\sup_{(y,z)\in \mathbb{S}_\tau}\tilde f(z)+\tilde h(z)}>\frac{\tau}{2a_2}$. \\
\indent By Lemma \ref{file6}, $$\sup_{(\delta,z)\in \mathcal{S}_\tau }\left|F_{Y|Z=z}\left(F_{Y|Z=z}^{-1}(\delta)\right)-\left(F+H\right)_{Y|Z=z}\left(F^{-1}_{Y|Z=z}(\delta)\right)\right|\leq a_1 \sup_{(y,z)\in\mathbb{S}_{(Y,Z)}} |h|$$ for some $0<a_1<\infty$. Then I have that 
\begin{align}\label{alphalesh}
    \sup_{(\delta,z)\in \mathcal{S}_\tau }\left|\alpha(f+h)-\alpha(f)\right|\leq \frac{2a_1a_2}{\tau} \sup_{(y,z)\in\mathbb{S}_{(Y,Z)}} |h|.
\end{align}
Next, I will obtain a first-order expansion for $\alpha\left(f+h\right)$. By the fact that $\left(F+H\right)_{Y|Z=z}\left(\alpha(f+h)\right) = \delta = F_{Y|Z=z}\left(\alpha(f)\right)$, I have
\begin{align*}
    &\frac{\int_{-\infty}^{\alpha(f)} f(y, z) d y}{\tilde{f}(z)}=\frac{\int^{\alpha(f+h)}(f(y, z)+h(y, z)) d y}{\widetilde{f}(z)+\widetilde{h}(z)}\\
    \implies &(\tilde{f}(z)+\tilde{h}(z))\int^{\alpha(f)}  f(y, z) d y=\tilde{f}(z)\int^{\alpha(f+h)} \left( f(y, z)+ h(y, z)\right) d y\\
    \implies & \tilde{h}(z)\int^{\alpha(f) }  f(y, z) d y - \tilde{f}(z)\int^{\alpha(f)}   h(y, z) d y\\
    &\quad= \tilde{f}(z)\int_{\alpha(f)}^{\alpha(f+h)}  h(y, z) d y+\tilde{f}(z)\int_{\alpha(f)}^{\alpha(f+h)}  f(y, z) d y
\end{align*}
By the Mean Value Theorem, there exist $r_f$ and $r_h$, between $\alpha(f)$ and $\alpha(f+h)$, such that
\begin{align*}
    &\int_{\alpha(f)}^{\alpha(f+h)} f(y, z) d y = f(r_f,z)\left(\alpha(f+h)-\alpha(f)\right) \text{\ and}\\
    &\int_{\alpha(f)}^{\alpha(f+h)} h(y, z) d y = h(r_h,z)\left(\alpha(f+h)-\alpha(f)\right).
\end{align*}
Denote $Az \equiv \tilde{h}(z)\int^{\alpha(f) }  f(y, z) d y - \tilde{f}(z)\int^{\alpha(f)}   h(y, z) d y$. Then
\begin{align*}
    &Az = \tilde f(z)\left[f(r_f,z)+h(r_h,z)\right]\left(\alpha(f+h)-\alpha(f)\right)\\
    \implies & \alpha(f+h)-\alpha(f) = \frac{Az}{\tilde f(z)\left[f(r_f,z)+h(r_h,z)\right]}.
\end{align*}
By the Mean Value Theorem, there exist $r_f'$ between $\alpha(f)$ and $r_f$ such that $f(r_f,z)-f(\alpha(f),z)=[\partial f(r_f',z)/\partial y](r_f-\alpha(f))$. Hence
\begin{footnotesize}
\begin{align*}
    \alpha(f+h)-\alpha(f) &= \frac{Az}{\tilde{f}(z)\left(f(\alpha(f), z)+\frac{\partial f\left(r_{f}^{\prime}, z\right)}{\partial y}\left(r_{f}-\alpha(f)\right)+h\left(r_{h}, z\right)\right)}\\
    &=\frac{Az}{\tilde{f}(z)f(\alpha(f), z)}-\left[\frac{\frac{\partial f(r_f',z)}{\partial y}\left(r_{f}-\alpha(f)\right)+h(r_h,z)}{f(\alpha(f), z)+\frac{\partial f\left(r_{f}^{\prime}, z\right)}{\partial y}\left(r_{f}-\alpha(f)\right)+h\left(r_{h}, z\right)}\right]\frac{Az}{\tilde{f}(z)f(\alpha(f), z)}
\end{align*}
\end{footnotesize}
Denote 
\begin{align*}
    &D\alpha(f;h)\equiv \frac{Az}{\tilde{f}(z)f(\alpha(f), z)} \text{\ and}\\
    &R\alpha(f;h)\equiv -\left[\frac{\frac{\partial f(r_f',z)}{\partial y}\left(r_{f}-\alpha(f)\right)+h(r_h,z)}{f(r_f, z)+h\left(r_{h}, z\right)}\right]\frac{Az}{\tilde{f}(z)f(\alpha(f), z)}.
\end{align*}
Then
\begin{align*}
    \alpha(f+h)-\alpha(f) = D\alpha(f;h)+R\alpha(f;h).
\end{align*}
By the definition of $r_f$ and by (\ref{alphalesh}), $\sup_{(\delta,z)\in \mathcal{S}_\tau}|r_f-\alpha(f)|\leq \sup_{(\delta,z)\in \mathcal{S}_\tau}\left|\alpha(f+h)-\alpha(f)\right|\leq \frac{2a_1a_2}{\tau} \sup_{(y,z)\in\mathbb{S}_{(Y,Z)}} |h|$. It follows by the continuity of $\partial f/\partial y$ and the compactness of $\mathbb{S}_{(Y,Z)}$ that
\begin{align*}
    \sup_{(\delta,z)\in \mathcal{S}_\tau}\left|\frac{\partial f(r_f',z)}{\partial y}\left(r_{f}-\alpha(f)\right)+h(r_h,z)\right|\leq d_1\sup_{(y,z)\in\mathbb{S}_{(Y,Z)}} |h|
\end{align*}
for some finite constant $d_1$. Then for all $h\in\mathrm{F}$ such that $\sup_{(y,z)\in\mathbb{S}_{(Y,Z)}} |h|\leq\min\left\{\tau/2,1\right\}$, I have $\inf_{(\delta,z)\in \mathcal{S}_\tau}f(r_f, z)+h\left(r_{h}, z\right)>\tau/2$. Then there exists a finite constant $a$ such that
\begin{align*}
    &\sup_{(\delta,z)\in \mathcal{S}_\tau}\left|D\alpha(f;h)\right|\leq a \sup_{(y,z)\in\mathbb{S}_{(Y,Z)}} |h|\text{\ and\ }\\
    &\sup_{(\delta,z)\in \mathcal{S}_\tau}\left|R\alpha(f;h)\right|\leq a\sup_{(y,z)\in\mathbb{S}_{(Y,Z)}} |h|^2.
\end{align*}
\end{proof}

\begin{lemma}\label{file3}
For any value $z\in \mathbb{S}_{Z}$, and any value $\delta\in(0,1)$, define functional $\Phi_{1,(j)}(\cdot)$ as 
$\Phi_{1,(j)}(g) \equiv \int_{-\infty}^{\alpha(g)}\frac{\partial g(y,z)}{\partial z_j}dy$. For simplicity, I leave the argument $(z, \delta)$ implicit. For any $h$ in $\mathrm{F}$ that $\sup_{(y,z)\in\mathbb{S}_{(Y,Z)}} |h|$ is sufficiently small, I have that
\begin{align*}
    \Phi_{1,(j)}(f+h)-\Phi_{1,(j)}(f) = D\Phi_{1,(j)}(f;h)+R\Phi_{1,(j)}(f;h), 
\end{align*}
where
\begin{align*}
    D\Phi_{1,(j)}(f;h)&\equiv \frac{\partial f(\alpha(f),z)}{\partial z_j}D\alpha(f;h)+\int_{-\infty}^{\alpha(f)}\frac{\partial h(y,z)}{\partial z_j}dy\\
    R\Phi_{1,(j)}(f;h)&\equiv \frac{\partial f(\alpha(f),z)}{\partial z_j}R\alpha(f;h)+\frac{\partial^2 f(\bar r_f',z)}{\partial y\partial z_j}\left(\bar r_f-\alpha(f)\right)\left(\alpha(f+h)-\alpha(f)\right)\\
    &\quad\quad+\frac{\partial h(\bar r_h,z)}{\partial z_j}\left(\alpha(f+h)-\alpha(f)\right),
\end{align*}
for some $\bar r_f$ and $\bar r_h$, and $\bar r_f'$ between $\alpha(f+h)$ and $\alpha(f)$ defined in the proof. Moreover, for some $0<b_1<\infty$,
\begin{align*}
    \sup_{(\delta,z)\in \mathcal{S}_\tau}\left|D\Phi_{1,(j)}(f;h)\right|&\leq b_1\sup_{(y,z)\in\mathbb{S}_{(Y,Z)}} |h|+b_1\sup_{(y,z)\in\mathbb{S}_{Y,Z}}\left|\frac{\partial h(y,z)}{\partial z_j}\right|\\
    \sup_{(\delta,z)\in \mathcal{S}_\tau}\left|R\Phi_{1,(j)}(f;h)\right|&\leq b_1\sup_{(y,z)\in\mathbb{S}_{(Y,Z)}} |h|^2+b_1\sup_{(y,z)\in\mathbb{S}_{(Y,Z)}} |h|\sup_{(y,z)\in\mathbb{S}_{Y,Z}}\left|\frac{\partial h(y,z)}{\partial z_j}\right|.
\end{align*}
\end{lemma}
\begin{proof}
\begin{footnotesize}
\begin{align*}
    \Phi_{1,(j)}(f+h)-\Phi_{1,(j)}(f)&=\int_{-\infty}^{\alpha(f+h)}\left(\frac{\partial f(y,z)}{\partial z_j}+\frac{\partial h(y,z)}{\partial z_j}\right)dy-\int_{-\infty}^{\alpha(f)}\frac{\partial f(y,z)}{\partial z_j}dy\\
    &=\int_{-\infty}^{\alpha(f)}\frac{\partial h(y,z)}{\partial z_j}dy+\int_{\alpha(f)}^{\alpha(f+h)}\frac{\partial h(y,z)}{\partial z_j}dy+\int_{\alpha(f)}^{\alpha(f+h)}\frac{\partial f(y,z)}{\partial z_j}dy
\end{align*}    
\end{footnotesize}
By the Mean Value Theorem, there exist $\bar r_f$ and $\bar r_h$ between $\alpha(f)$ and $\alpha(f+h)$ such that 
\begin{align*}
    &\int_{\alpha(f)}^{\alpha(f+h)}\frac{\partial f(y,z)}{\partial z_j}dy=\frac{\partial f(\bar r_f,z)}{\partial z_j}\left(\alpha(f+h)-\alpha(f)\right)\\
    &\int_{\alpha(f)}^{\alpha(f+h)}\frac{\partial h(y,z)}{\partial z_j}dy=\frac{\partial h(\bar r_h,z)}{\partial z_j}\left(\alpha(f+h)-\alpha(f)\right)
\end{align*}
Apply the Mean Value Theorem again. I have that there exist $\bar r_f'$ between $\alpha(f)$ and $\alpha(f+h)$ such that
\begin{footnotesize}
\begin{align*}
    \int_{\alpha(f)}^{\alpha(f+h)}\frac{\partial f(y,z)}{\partial z_j}dy=\frac{\partial f(\alpha(f),z)}{\partial z_j}\left(\alpha(f+h)-\alpha(f)\right)+\frac{\partial^2 f(\bar r_f',z)}{\partial y\partial z_j}\left(\bar r_f-\alpha(f)\right)\left(\alpha(f+h)-\alpha(f)\right).
\end{align*}    
\end{footnotesize}
Let
\begin{align*}
    D\Phi_{1,(j)}(f;h)&\equiv \frac{\partial f(\alpha(f),z)}{\partial z_j}D\alpha(f;h)+\int_{-\infty}^{\alpha(f)}\frac{\partial h(y,z)}{\partial z_j}dy\\
    R\Phi_{1,(j)}(f;h)&\equiv \frac{\partial f(\alpha(f),z)}{\partial z_j}R\alpha(f;h)+\frac{\partial^2 f(\bar r_f',z)}{\partial y\partial z_j}\left(\bar r_f-\alpha(f)\right)\left(\alpha(f+h)-\alpha(f)\right)\\
    &\quad\quad+\frac{\partial h(\bar r_h,z)}{\partial z_j}\left(\alpha(f+h)-\alpha(f)\right).
\end{align*}
By the boundedness of $\mathbb{S}_{Y,Z}$, continuity of $\partial f/\partial z_j$, $\partial^2 f/\partial y\partial z_j$, $\partial h/\partial z_j$ and Lemma \ref{file5} there exist finite constants $d_2, d_3, d_4, d_5$ such that 
\begin{align*}
    \sup_{(\delta,z)\in \mathcal{S}_\tau}\left|D\Phi_{1,(j)}(f;h)\right|&\leq d_2\sup_{(y,z)\in\mathbb{S}_{(Y,Z)}} |h|+d_3\sup_{(y,z)\in\mathbb{S}_{Y,Z}}\left|\frac{\partial h(y,z)}{\partial z_j}\right|\\
    \sup_{(\delta,z)\in \mathcal{S}_\tau}\left|R\Phi_{1,(j)}(f;h)\right|&\leq d_4\sup_{(y,z)\in\mathbb{S}_{(Y,Z)}} |h|^2+d_5\sup_{(y,z)\in\mathbb{S}_{(Y,Z)}} |h|\sup_{(y,z)\in\mathbb{S}_{Y,Z}}\left|\frac{\partial h(y,z)}{\partial z_j}\right|
\end{align*}
Let $b_1\equiv\max\{d_2, d_3, d_4, d_5\}$, then
\begin{align*}
    \sup_{(\delta,z)\in \mathcal{S}_\tau}\left|D\Phi_{1,(j)}(f;h)\right|&\leq b_1\sup_{(y,z)\in\mathbb{S}_{(Y,Z)}} |h|+b_1\sup_{(y,z)\in\mathbb{S}_{Y,Z}}\left|\frac{\partial h(y,z)}{\partial z_j}\right|\\
    \sup_{(\delta,z)\in \mathcal{S}_\tau}\left|R\Phi_{1,(j)}(f;h)\right|&\leq b_1\sup_{(y,z)\in\mathbb{S}_{(Y,Z)}} |h|^2+b_1\sup_{(y,z)\in\mathbb{S}_{(Y,Z)}} |h|\sup_{(y,z)\in\mathbb{S}_{Y,Z}}\left|\frac{\partial h(y,z)}{\partial z_j}\right|.
\end{align*}
\end{proof}

\begin{lemma}\label{file4}
Define functional $\Phi_2()$ on $\mathrm{F}$ by $\Phi_2(g)=g(\alpha(g),z)$. For any $h$ in $\mathrm{F}$ that $\sup_{(y,z)\in\mathbb{S}_{(Y,Z)}} |h|$ is small enough, I have that
\begin{align*}
    \Phi_2(f+h)-\Phi_2(f) = D\Phi_2(f;h)+R\Phi_2(f;h), 
\end{align*}
where
\begin{align*}
    D\Phi_2(f;h)&\equiv \frac{\partial f(\alpha(f),z)}{\partial y}D\alpha(f;h)+h(\alpha(f),z)\\
    R\Phi_2(f;h)&\equiv \frac{\partial f(\alpha(f),z)}{\partial y}R\alpha(f;h)+\frac{\partial^2 f(\tilde r_f',z)}{\partial y^2}\left(\tilde r_f-\alpha(f)\right)\left(\alpha(f+h)-\alpha(f)\right)+\\
    &\quad\quad\frac{\partial h(\tilde r_h,z)}{\partial y}\left(\alpha(f+h)-\alpha(f)\right).
\end{align*}
for some $\tilde r_f$ and $\tilde r_h$, and $\tilde r_f'$ between $\alpha(f+h)$ and $\alpha(f)$ defined in the proof. Moreover, for some $0<b_2<\infty$,
\begin{align*}
    &\sup_{(\delta,z)\in \mathcal{S}_\tau}\left|D\Phi_2(f;h)\right|\leq b_2 \sup_{(y,z)\in\mathbb{S}_{(Y,Z)}} |h|.\\
    &\sup_{(\delta,z)\in \mathcal{S}_\tau}\left|R\Phi_2(f;h)\right|\leq b_2\sup_{(y,z)\in\mathbb{S}_{(Y,Z)}} |h|^2+b_2\sup_{(y,z)\in\mathbb{S}_{(Y,Z)}} |h|\sup_{(y,z)\in\mathbb{S}_{Y,X}}\left|\frac{\partial h(y,z)}{\partial y}\right|.
\end{align*}
\end{lemma}

\begin{proof}
\begin{align*}
    &\quad\Phi_2(f+h)-\Phi_2(f)\\
    &=(f+h)\left(\alpha(f+h),z\right)-f\left(\alpha(f),z\right)\\
    &=f\left(\alpha(f+h),z\right)-f\left(\alpha(f),z\right)+h\left(\alpha(f+h),z\right)-h\left(\alpha(f),z\right)+h\left(\alpha(f),z\right)
\end{align*}
By the Mean Value Theorem, there exist $\tilde r_f$ and $\tilde r_h$ between $\alpha(f+h)$ and $\alpha(f)$, such that
\begin{align*}
    f\left(\alpha(f+h),z\right)-f\left(\alpha(f),z\right) &= \frac{\partial f(\tilde r_f,z)}{\partial y}\left(\alpha(f+h)-\alpha(f)\right)\\
    h\left(\alpha(f+h),z\right)-h\left(\alpha(f),z\right) &= \frac{\partial h(\tilde r_h,z)}{\partial y}\left(\alpha(f+h)-\alpha(f)\right).
\end{align*}
Apply the Mean Value Theorem again. There exist $\tilde r_f'$ between $\tilde r_f$ and $\alpha(f)$ such that
\begin{align*}
    &\quad f\left(\alpha(f+h),z\right)-f\left(\alpha(f),z\right) \\
    &= \frac{\partial f(\alpha(f),z)}{\partial y}\left(\alpha(f+h)-\alpha(f)\right)+\frac{\partial^2 f(\tilde r_f',z)}{\partial y^2}\left(\tilde r_f-\alpha(f)\right)\left(\alpha(f+h)-\alpha(f)\right)
\end{align*}
Then by Lemma \ref{file5} I have that
\begin{align*}
    &\quad\Phi_2(f+h)-\Phi_2(f)\\
    &= \frac{\partial f(\alpha(f),z)}{\partial y}\left(D\alpha(f;h)+R\alpha(f;h)\right)+\frac{\partial^2 f(\tilde r_f',z)}{\partial y^2}\left(\tilde r_f-\alpha(f)\right)\left(\alpha(f+h)-\alpha(f)\right)+\\
    &\quad\quad\quad\quad\frac{\partial h(\tilde r_h,z)}{\partial y}\left(\alpha(f+h)-\alpha(f)\right)+h(\alpha(f),z),
\end{align*}
where for some $0<a<\infty$ and $\sup_{(y,z)\in\mathbb{S}_{(Y,Z)}} |h|\leq\epsilon_0\equiv\min\{\tau/2,1\}$,
\begin{align*}
    &\sup_{(\delta,z)\in \mathcal{S}_\tau}\left|D\alpha(f;h)\right|\leq a \sup_{(y,z)\in\mathbb{S}_{(Y,Z)}} |h|\\
    &\sup_{(\delta,z)\in \mathcal{S}_\tau}\left|R\alpha(f;h)\right|\leq a \sup_{(y,z)\in\mathbb{S}_{(Y,Z)}} |h|^2.
\end{align*}
Define
\begin{align*}
    D\Phi_2(f;h)&\equiv \frac{\partial f(\alpha(f),z)}{\partial y}D\alpha(f;h)+h(\alpha(f),z)\\
    R\Phi_2(f;h)&\equiv \frac{\partial f(\alpha(f),z)}{\partial y}R\alpha(f;h)+\frac{\partial^2 f(\tilde r_f',z)}{\partial y^2}\left(\tilde r_f-\alpha(f)\right)\left(\alpha(f+h)-\alpha(f)\right)+\\
    &\quad\quad\frac{\partial h(\tilde r_h,z)}{\partial y}\left(\alpha(f+h)-\alpha(f)\right).
\end{align*}
Then
\begin{align*}
    \Phi_2(f+h)-\Phi_2(f)=D\Phi_2(f;h)+R\Phi_2(f;h).
\end{align*}
By the compactness of $\mathbb{S}_{(Y,X)}$ and continuity of $\partial f/\partial y$, $\partial^2 f/\partial y^2$, and $\partial h/\partial y$, I have that there exist some finite constant $d_6$ such that 
\begin{align*}
    &\sup_{(\delta,z)\in \mathcal{S}_\tau}\left|D\Phi_2(f;h)\right|\leq d_6 \sup_{(y,z)\in\mathbb{S}_{(Y,Z)}} |h|.
\end{align*}
holds for all $h\in\mathrm F$. On the other hand, by Lemma \ref{file5}, there exist some finite constant $a_1, a_2$ such that $\sup_{(\delta,z)\in \mathcal{S}_\tau}|\tilde r_f-\alpha(f)|\leq \sup_{(\delta,z)\in \mathcal{S}_\tau}\left|\alpha(f+h)-\alpha(f)\right|\leq \frac{2a_1a_2}{\tau} \sup_{(y,z)\in\mathbb{S}_{(Y,Z)}} |h|$. Then there exist some finite constants $d_7, d_8$ such that
\begin{align*}
    \sup_{(\delta,z)\in \mathcal{S}_\tau}\left|R\Phi_2(f;h)\right|&\leq d_7\sup_{(y,z)\in\mathbb{S}_{(Y,Z)}} |h|^2+d_8\sup_{(y,z)\in\mathbb{S}_{(Y,Z)}} |h|\sup_{(y,z)\in\mathbb{S}_{Y,Z}}\left|\frac{\partial h(y,z)}{\partial y}\right|
\end{align*}
Let $b_2\equiv\max\{d_6, d_7, d_8\}$. Then
\begin{align*}
    &\sup_{(\delta,z)\in \mathcal{S}_\tau}\left|D\Phi_2(f;h)\right|\leq b_2 \sup_{(y,z)\in\mathbb{S}_{(Y,Z)}} |h|.\\
    &\sup_{(\delta,z)\in \mathcal{S}_\tau}\left|R\Phi_2(f;h)\right|\leq b_2\sup_{(y,z)\in\mathbb{S}_{(Y,Z)}} |h|^2+b_2\sup_{(y,z)\in\mathbb{S}_{(Y,Z)}} |h|\sup_{(y,z)\in\mathbb{S}_{Y,Z}}\left|\frac{\partial h(y,z)}{\partial y}\right|.
\end{align*}
\end{proof}

\begin{lemma}\label{file2} For any value $z\in \mathbb{S}_{Z}$, and any value $\delta\in(0,1)$, define functionals $\Psi_{1,(j)}(\cdot)$ and $\tilde\Psi_1(\cdot)$ as $\Psi_{1,(j)}(g) \equiv \int_{-\infty}^{y}\frac{\partial  g(s,z)}{\partial z_j}ds$ and 
$\tilde\Psi_1(g) \equiv \int_{-\infty}^{x}\frac{\partial \tilde g(s,w)}{\partial w}ds$. For simplicity, I leave the argument $z$ implicit. For any $h$ in $\mathrm{F}$ that $\sup_{(y,z)\in\mathbb{S}_{(Y,Z)}} |h|$ is sufficiently small, I have that
\begin{align*}
    &\Phi_{(j)}(f+h)-\Phi_{(j)}(f) = D\Phi_{(j)}(f;h)+R\Phi_{(j)}(f;h), \\
    &\Psi_{(j)}(f+h)-\Psi_{(j)}(f) = D\Psi_{(j)}(f;h)+R\Psi_{(j)}(f;h), \\
    &\tilde\Psi(f+h)-\tilde\Psi(f) = D\tilde\Psi(f;h)+R\tilde\Psi(f;h), 
\end{align*}
where
\begin{scriptsize}
\begin{align*}
    &D\Phi_{(j)}(f;h)\equiv\frac{\left(\delta\frac{\partial \tilde h(z)}{\partial z_j}-D\Phi_{1,(j)}(f;h)\right)\Phi_2(f)-\left(\delta\frac{\partial \tilde f(z)}{\partial z_j}-\Phi_{1,(j)}(f)\right)D\Phi_2(f;h)}{\Phi^2_2(f)}\\
    &R\Phi_{(j)}(f;h)\equiv-\frac{\Phi_2(f)R\Phi_{1,(j)}(f;h)+\left(\delta\frac{\partial \tilde f(z)}{\partial z_j}-\Phi_{1,(j)}(f)\right)R\Phi_2(f;h)}{\Phi_2^2(f)}\\
    &\quad-\frac{\Delta\Phi_2(f;h)\left[\left(\delta\frac{\partial \tilde h(z)}{\partial z_j}-\Delta\Phi_{1,(j)}(f;h)\right)\Phi_2(f)-\left(\delta \frac{\partial \tilde f(z)}{\partial z_j}-\Phi_{1,(j)}(f)\right)\Delta\Phi_2(f;h)\right]}{\Phi_2(f+h)\Phi^2_2(f)}\\
    &D\Psi_{(j)}(f;h)\equiv\frac{\left(D\Lambda(f;h)\frac{\partial \tilde f(z)}{\partial z_j}+\Lambda(f)\frac{\partial \tilde h(z)}{\partial z_j}-\int_{-\infty}^{y}\frac{\partial h(s,z)}{\partial z_j}ds\right)\times f(y,z)-\left(\Lambda(f)\frac{\partial \tilde f(z)}{\partial z_j}-\Psi_{1,(j)}(f)\right)\times h(y,z)}{f^2(y,z)}\\
    &R\Psi_{(j)}(f;h)\equiv\frac{\left(R\Lambda(f;h)\frac{\partial \tilde f(z)}{\partial z_j}+\frac{\partial \tilde h(z)}{\partial z_j}\left(D\Lambda(f;h)+R\Lambda(f;h)\right)\right)}{f(y,z)}\\
    &\quad\quad\quad\quad-h(y,z)\times\left[\frac{\Lambda(f)\frac{\partial \tilde h(z)}{\partial z_j}f(y,z)+\Delta\Lambda(f;h)\left(\frac{\partial \tilde f(z)}{\partial z_j}+\frac{\partial \tilde h(z)}{\partial z_j}\right)f(y,z)}{\left(f(y,z)+ h(y,z)\right) f^2(y,z)}\right.\\
    &\quad\quad\quad\quad\quad\quad\quad\quad\quad\quad\left.-\frac{\Lambda(f)\frac{\partial \tilde f(z)}{\partial z_j}h(y,z)+\Delta\Psi_{1,(j)}(f;h)f(y,z)-\Psi_{1,(j)}(f;h)h(y,z)}{\left(f(y,z)+ h(y,z)\right) f^2(y,z)}\right]\\
    &D\tilde\Psi(f;h)\equiv\frac{\left(D\tilde\Lambda(f;h)\frac{\partial \tilde f_W(w)}{\partial w}+\tilde\Lambda(f)\frac{\partial \tilde h_W(w)}{\partial w}-\int_{-\infty}^{x}\frac{\partial \tilde h(s,w)}{\partial w}ds\right)\times \tilde f(x,w)-\left(\tilde \Lambda(f)\frac{\partial \tilde f_W(w)}{\partial w}-\tilde\Psi_1(f)\right)\times \tilde h(x,w)}{\tilde f^2(x,w)}\\
    &R\tilde\Psi(f;h)\equiv\frac{\left(R\tilde\Lambda(f;h)\frac{\partial \tilde f_W(w)}{\partial w}+\frac{\partial \tilde h_W(w)}{\partial w}\left(D\tilde\Lambda(f;h)+R\tilde\Lambda(f;h)\right)\right)}{\tilde f(x,w)}\\
    &\quad\quad\quad\quad-\tilde h(x,w)\times\left[\frac{\tilde\Lambda(f)\frac{\partial \tilde h_W(w)}{\partial w}\tilde f(x,w)+\Delta\tilde\Lambda(f;h)\left(\frac{\partial \tilde f_W(w)}{\partial w}+\frac{\partial \tilde h_W(w)}{\partial w}\right)\tilde f(x,w)}{\left(\tilde f(x,w)+ \tilde h(x,w)\right) \tilde f^2(x,w)}\right.\\
    &\quad\quad\quad\quad\quad\quad\quad\quad\quad\quad \left.-\frac{\tilde\Lambda(f)\frac{\partial \tilde f_W(w)}{\partial w}\tilde h(x,w)+\Delta\tilde \Psi_{1,(j)}(f;h)\tilde f(x,w)-\tilde\Psi_1(f;h)\tilde h(x,w)}{\left(\tilde f(x,w)+ \tilde h(x,w)\right) \tilde f^2(x,w)}\right],
\end{align*}
\end{scriptsize}
with $\Delta\Phi_{l}(f;h)\equiv D\Phi_{l}(f;h)+R\Phi_{l}(f;h)$ for $l=1,(j)$ and $l=2$, $\Delta\Psi_{1,(j)}(f;h)\equiv D\Psi_{1,(j)}(f;h)+R\Psi_{1,(j)}(f;h)$ and $\Delta\tilde\Psi_1(f;h)\equiv D\tilde\Psi_1(f;h)+R\tilde\Psi_1(f;h)$. Moreover, for some $c_1, c_2, c_3<\infty$,
\begin{footnotesize}
\begin{align*}
    &\sup_{(\delta,z)\in \mathcal{S}_\tau}\left|D\Phi_{(j)}(f;h)\right|\leq c_1\sup_{(y,z)\in\mathbb{S}_{(Y,Z)}} |h|+c_1\sup_{(y,z)\in\mathbb{S}_{Y,Z}}\left|\frac{\partial h(y,z)}{\partial z_j}\right|\\
    &\sup_{(\delta,z)\in \mathcal{S}_\tau}\left|R\Phi_{(j)}(f;h)\right|\leq c_1\sup_{(y,z)\in\mathbb{S}_{(Y,Z)}} |h|^2+c_1\sup_{(y,z)\in\mathbb{S}_{(Y,Z)}} |h|\sup_{(y,z)\in\mathbb{S}_{Y,Z}}\left|\frac{\partial h(y,z)}{\partial y}\right|\\
    &\quad\quad\quad\quad\quad\quad\quad\quad\quad\quad\quad\quad+c_1\sup_{(y,z)\in\mathbb{S}_{(Y,Z)}} |h|\sup_{(y,z)\in\mathbb{S}_{Y,Z}}\left|\frac{\partial h(y,z)}{\partial z_j}\right|\\
    &\sup_{(y,z)\in\mathbb{S}_\tau }\left|D\Psi_{(j)}(f;h)\right|\leq c_2\sup_{(y,z)\in\mathbb{S}_{(Y,Z)}} |h|+c_2\sup_{(y,z)\in\mathbb{S}_{(Y,Z)} }\left|\frac{\partial h(y,z)}{\partial z_{j}}\right|\\
    &\sup_{(y,z)\in\mathbb{S}_\tau }\left|R\Psi_{(j)}(f;h)\right|\leq c_2\sup_{(y,z)\in\mathbb{S}_{(Y,Z)}} |h|^2+c_2\sup_{(y,z)\in\mathbb{S}_{(Y,Z)}} |h|\sup_{(y,z)\in\mathbb{S}_{(Y,Z)} }\left|\frac{\partial h(y,z)}{\partial z_{j}}\right|\\
    &\sup_{(y,z)\in\mathbb{S}_\tau }\left|D\tilde\Psi(f;h)\right|\leq c_3\sup_{(y,z)\in\mathbb{S}_{(Y,Z)}} |h|+c_3\sup_{(y,z)\in\mathbb{S}_{(Y,Z)} }\left|\frac{\partial h(y,z)}{\partial z_{j+1}}\right|\\
    &\sup_{(y,z)\in\mathbb{S}_\tau }\left|R\tilde\Psi(f;h)\right|\leq c_3\sup_{(y,z)\in\mathbb{S}_{(Y,Z)}} |h|^2+c_3\sup_{(y,z)\in\mathbb{S}_{(Y,Z)}} |h|\sup_{(y,z)\in\mathbb{S}_{(Y,Z)} }\left|\frac{\partial h(y,z)}{\partial z_{j+1}}\right|.
\end{align*}
\end{footnotesize}

\end{lemma}

\begin{proof}
By Lemma \ref{file1},
\begin{align*}
    &\Phi_{(j)}(f) = \frac{\delta \frac{\partial \tilde f(z)}{\partial z_j}-\Phi_{1,(j)}(f)}{\Phi_2(f)} \\
    &\Psi_{(j)}(f) = \frac{\Lambda(f) \frac{\partial \tilde f(z)}{\partial z_j}-\Psi_{1,(j)}(f)}{f(y,z)}\text{\  and}\\
    &\tilde\Psi(f) = \frac{\tilde\Lambda(f) \frac{\partial \tilde f_W(w)}{\partial w}-\tilde\Psi_1(f)}{\tilde f(x,w)}
\end{align*}
For all $h\in\mathrm{F}$,
\begin{footnotesize}
\begin{align*}
    \Phi_{(j)}(f+h) - \Phi_{(j)}(f) &= \frac{\delta \frac{\partial \tilde f(z)}{\partial z_j}+\delta\frac{\partial \tilde h(z)}{\partial z_j}-\Phi_{1,(j)}(f+h)}{\Phi_2(f+h)} -\frac{\delta \frac{\partial \tilde f(z)}{\partial z_j}-\Phi_{1,(j)}(f)}{\Phi_2(f)}\\
    &\equiv\frac{N_1'}{D'_1}-\frac{N_1}{D_1}\\
    \Psi_{(j)}(f+h) - \Psi_{(j)}(f) &= \frac{\Lambda(f+h)\left(\frac{\partial \tilde f(z)}{\partial z_j}+\frac{\partial \tilde h(z)}{\partial z_j}\right)-\Psi_{1,(j)}(f+h)}{f(y,z)+ h(y,z)} -\frac{ \Lambda(f)\frac{\partial \tilde f(z)}{\partial z_j}-\Psi_{1,(j)}(f)}{f(y,z)}\\
    &\equiv\frac{N_2'}{D'_2}-\frac{N_2}{D_2}\\
    \tilde\Psi(f+h) - \tilde\Psi(f) &= \frac{\tilde\Lambda(f+h)\left(\tilde f_W(w)+\tilde h_W(w)\right)-\tilde\Psi_1(f+h)}{\tilde f(x,w)+\tilde h(x,w)} -\frac{\tilde \Lambda(f)\tilde f_W(w)-\tilde\Psi_1(f)}{\tilde f(x,w)}\\
    &\equiv\frac{N_3'}{D'_3}-\frac{N_3}{D_3},
\end{align*}    
\end{footnotesize}
where I denote
\begin{align*}
    N_1' &= \delta \frac{\partial \tilde f(z)}{\partial z_j}+\delta\frac{\partial \tilde h(z)}{\partial z_j}-\Phi_{1,(j)}(f+h)\\
    N_1 &= \delta \frac{\partial \tilde f(z)}{\partial z_j}-\Phi_{1,(j)}(f)\\
    N_2' &= \Lambda(f+h)\left(\frac{\partial \tilde f(z)}{\partial z_j}+\frac{\partial \tilde h(z)}{\partial z_j}\right)-\Psi_{1,(j)}(f+h)\\
    N_2 &= \Lambda(f)\frac{\partial \tilde f(z)}{\partial z_j}-\Psi_{1,(j)}(f)\\
    D'_1&=\Phi_2(f+h)\\
    D_1 &= \Phi_2(f)\\
    D'_2&= f(y,z)+ h(y,z)\\
    D_2&=f(y,z).
\end{align*}
Note that 
\begin{align*}
    N_1'-N_1 &= \delta\frac{\partial \tilde h(z)}{\partial z_j} - D\Phi_{1,(j)}(f;h)-R\Phi_{1,(j)}(f;h)\equiv DN_1+RN_1\\
    N_2'-N_2 &= \left( \Lambda(f+h)- \Lambda(f)\right)\left(\frac{\partial \tilde f(z)}{\partial z_j}+\frac{\partial \tilde h(z)}{\partial z_j}\right)+\Lambda(f)\frac{\partial \tilde h(z)}{\partial z_j}-\int_{-\infty}^{y}\frac{\partial h(s,z)}{\partial z_j}ds\\
    &= \left(D\Lambda(f;h)+R\Lambda(f;h)\right)\left(\frac{\partial \tilde f(z)}{\partial z_j}+\frac{\partial \tilde h(z)}{\partial z_j}\right)+\Lambda(f)\frac{\partial \tilde h(z)}{\partial z_j} -\int_{-\infty}^{y}\frac{\partial h(s,z)}{\partial z_j}ds\\
    &\equiv DN_2+RN_2\\
    D'_1-D_1 &= D\Phi_2(f;h)+R\Phi_2(f;h)\\
    D'_2-D_2 &=h(y,z),
\end{align*}
where 
\begin{align*}
    DN_1&\equiv \delta\frac{\partial \tilde h(z)}{\partial z_j}-D\Phi_{1,(j)}(f;h)\\
    RN_1&\equiv -R\Phi_{1,(j)}(f;h)\\
    DN_2&\equiv D\Lambda(f;h)\frac{\partial \tilde f(z)}{\partial z_j}+\Lambda(f)\frac{\partial \tilde h(z)}{\partial z_j}-\int_{-\infty}^{y}\frac{\partial h(s,z)}{\partial z_j}ds\\
    RN_2&\equiv R\Lambda(f;h)\frac{\partial \tilde f(z)}{\partial z_j}+\frac{\partial \tilde h(z)}{\partial z_j}\left(D\Lambda(f;h)+R\Lambda(f;h)\right).
\end{align*}
I will make use of the equation:
\begin{align*}
    \frac{N_j'}{D'_j}-\frac{N_j}{D_j}&=\frac{N_j'D_j-N_jD'_j}{D_j^2}-\frac{(D'_j-D_j)(N_j'D_j-N_jD'_j)}{D'_jD_j^2}\\
    &=\frac{(N_j'-RN_j)D_j-N_j\left(D'_j-RD_j\right)}{D_j^2}\\
    &\quad+\frac{D_j\times RN_j-N_j\times RD_j}{D_j^2}-\frac{(D'_j-D_j)(N_j'D_j-N_jD'_j)}{D'_jD_j^2}.
\end{align*}
Then
\begin{scriptsize}
\begin{align*}
    &\quad\Phi_{(j)}(f+h) - \Phi_{(j)}(f) \\
    &= \frac{\left(\delta\frac{\partial \tilde h(z)}{\partial z_j}-D\Phi_{1,(j)}(f;h)\right)\Phi_2(f)-\left(\delta\frac{\partial \tilde f(z)}{\partial z_j}-\Phi_{1,(j)}(f)\right)D\Phi_2(f;h)}{\Phi^2_2(f)}\\
    &\quad-\frac{\Phi_2(f)R\Phi_{1,(j)}(f;h)+\left(\delta\frac{\partial \tilde f(z)}{\partial z_j}-\Phi_{1,(j)}(f)\right)R\Phi_2(f;h)}{\Phi_2^2(f)}\\
    &\quad-\frac{\left(D\Phi_2(f;h)+R\Phi_2(f;h)\right)\left(\left(\delta \frac{\partial \tilde f(z)}{\partial z_j}+\delta\frac{\partial \tilde h(z)}{\partial z_j}-\Phi_{1,(j)}(f+h)\right)\Phi_2(f)-\left(\delta \frac{\partial \tilde f(z)}{\partial z_j}-\Phi_{1,(j)}(f)\right)\Phi_2(f+h)\right)}{\Phi_2(f+h)\Phi^2_2(f)}.
\end{align*}
\end{scriptsize}

\begin{scriptsize}
\begin{align*}
    &\quad \Psi_{(j)}(f+h)-\Psi_{(j)}(f)\\
    &=\frac{DN_2\times D_2-N_2\times h(y,z)}{D_2^2}+\frac{D_2\times RN_2}{D_2^2}-\frac{h(y,z)(N_2'D_2-N_2D'_2)}{D'_2D_2^2}\\
    &=\frac{\left(D\Lambda(f;h)\frac{\partial \tilde f(z)}{\partial z_j}+\Lambda(f)\frac{\partial \tilde h(z)}{\partial z_j}-\int_{-\infty}^{y}\frac{\partial h(s,z)}{\partial z_j}ds\right)\times f(y,z)-\left(\Lambda(f)\frac{\partial \tilde f(z)}{\partial z_j}-\Psi_{1,(j)}(f)\right)\times h(y,z)}{f^2(y,z)}\\
    &\quad+\frac{\left(R\Lambda(f;h)\frac{\partial \tilde f(z)}{\partial z_j}+\frac{\partial \tilde h(z)}{\partial z_j}\left(D\Lambda(f;h)+R\Lambda(f;h)\right)\right)}{f(y,z)}\\
    &\quad-\frac{h(y,z)\left(\left(\Lambda(f+h)\left(\frac{\partial \tilde f(z)}{\partial z_j}+\frac{\partial \tilde h(z)}{\partial z_j}\right)-\Psi_{1,(j)}(f+h)\right)f(y,z)-\left(\Lambda(f)\frac{\partial \tilde f(z)}{\partial z_j}-\Psi_{1,(j)}(f)\right)\left(f(y,z)+ h(y,z)\right)\right)}{\left(f(y,z)+ h(y,z)\right) f^2(y,z)}.
\end{align*}
\end{scriptsize}

Denote $D\Phi_{(j)}(f;h)$, $D\Psi_{(j)}(f;h)$, $R\Phi_{(j)}(f;h)$, and $R\Psi_{(j)}(f;h)$ by
\begin{scriptsize}
\begin{align*}
    &D\Phi_{(j)}(f;h)\equiv\frac{\left(\delta\frac{\partial \tilde h(z)}{\partial z_j}-D\Phi_{1,(j)}(f;h)\right)\Phi_2(f)-\left(\delta\frac{\partial \tilde f(z)}{\partial z_j}-\Phi_{1,(j)}(f)\right)D\Phi_2(f;h)}{\Phi^2_2(f)}\\
    &R\Phi_{(j)}(f;h)\equiv-\frac{\Phi_2(f)R\Phi_{1,(j)}(f;h)+\left(\delta\frac{\partial \tilde f(z)}{\partial z_j}-\Phi_{1,(j)}(f)\right)R\Phi_2(f;h)}{\Phi_2^2(f)}\\
    &\quad-\frac{\left(D\Phi_2(f;h)+R\Phi_2(f;h)\right)\left(\left(\delta \frac{\partial \tilde f(z)}{\partial z_j}+\delta\frac{\partial \tilde h(z)}{\partial z_j}-\Phi_{1,(j)}(f+h)\right)\Phi_2(f)-\left(\delta \frac{\partial \tilde f(z)}{\partial z_j}-\Phi_{1,(j)}(f)\right)\Phi_2(f+h)\right)}{\Phi_2(f+h)\Phi^2_2(f)}\\
    &\quad\quad\quad\quad=-\frac{\Phi_2(f)R\Phi_{1,(j)}(f;h)+\left(\delta\frac{\partial \tilde f(z)}{\partial z_j}-\Phi_{1,(j)}(f)\right)R\Phi_2(f;h)}{\Phi_2^2(f)}\\
    &\quad-\frac{\Delta\Phi_2(f;h)\left[\left(\delta\frac{\partial \tilde h(z)}{\partial z_j}-\Delta\Phi_{1,(j)}(f;h)\right)\Phi_2(f)-\left(\delta \frac{\partial \tilde f(z)}{\partial z_j}-\Phi_{1,(j)}(f)\right)\Delta\Phi_2(f;h)\right]}{\Phi_2(f+h)\Phi^2_2(f)}\\
    &D\Psi_{(j)}(f;h)\equiv\frac{\left(D\Lambda(f;h)\frac{\partial \tilde f(z)}{\partial z_j}+\Lambda(f)\frac{\partial \tilde h(z)}{\partial z_j}-\int_{-\infty}^{y}\frac{\partial h(s,z)}{\partial z_j}ds\right)\times f(y,z)-\left(\Lambda(f)\frac{\partial \tilde f(z)}{\partial z_j}-\Psi_{1,(j)}(f)\right)\times h(y,z)}{f^2(y,z)}\\
    &R\Psi_{(j)}(f;h)\equiv\frac{\left(R\Lambda(f;h)\frac{\partial \tilde f(z)}{\partial z_j}+\frac{\partial \tilde h(z)}{\partial z_j}\left(D\Lambda(f;h)+R\Lambda(f;h)\right)\right)}{f(y,z)}\\
    &\quad-\frac{h(y,z)\left(\left(\Lambda(f+h)\left(\frac{\partial \tilde f(z)}{\partial z_j}+\frac{\partial \tilde h(z)}{\partial z_j}\right)-\Psi_{1,(j)}(f+h)\right)f(y,z)-\left(\Lambda(f)\frac{\partial \tilde f(z)}{\partial z_j}-\Psi_{1,(j)}(f)\right)\left(f(y,z)+ h(y,z)\right)\right)}{\left(f(y,z)+ h(y,z)\right) f^2(y,z)}\\
    &\quad\quad\quad\quad=\frac{\left(R\Lambda(f;h)\frac{\partial \tilde f(z)}{\partial z_j}+\frac{\partial \tilde h(z)}{\partial z_j}\left(D\Lambda(f;h)+R\Lambda(f;h)\right)\right)}{f(y,z)}\\
    &\quad\quad\quad\quad\quad\quad-h(y,z)\times\left[\frac{\Lambda(f)\frac{\partial \tilde h(z)}{\partial z_j}f(y,z)+\Delta\Lambda(f;h)\left(\frac{\partial \tilde f(z)}{\partial z_j}+\frac{\partial \tilde h(z)}{\partial z_j}\right)f(y,z)}{\left(f(y,z)+ h(y,z)\right) f^2(y,z)}\right.\\
    &\quad\quad\quad\quad\quad\quad\quad\quad\quad\quad \left.-\frac{\Lambda(f)\frac{\partial \tilde f(z)}{\partial z_j}h(y,z)+\Delta\Psi_{1,(j)}(f;h)f(y,z)-\Psi_{1,(j)}(f;h)h(y,z)}{\left(f(y,z)+ h(y,z)\right) f^2(y,z)}\right],
\end{align*}
\end{scriptsize}
where $\Delta\Phi_{l}(f;h)\equiv D\Phi_{l}(f;h)+R\Phi_{l}(f;h)$ for $l=1,(j)$ and $l=2$, and $\Delta\Psi_{1,(j)}(f;h)\equiv D\Psi_{1,(j)}(f;h)+R\Psi_{1,(j)}(f;h)$. Then 
\begin{align*}
   \Phi_{(j)}(f+h)-\Phi_{(j)}(f) = D\Phi_{(j)}(f;h)+ R\Phi_{(j)}(f;h)\\
   \Psi_{(j)}(f+h)-\Psi_{(j)}(f) = D\Psi_{(j)}(f;h)+ R\Psi_{(j)}(f;h).
\end{align*}
By the same logic, I can write
\begin{align*}
   \tilde\Psi(f+h)-\tilde\Psi(f) = D\tilde\Psi(f;h)+ R\tilde\Psi(f;h),
\end{align*}
where $D\tilde\Psi(f;h)$ and $R\tilde\Psi(f;h)$ are as defined in the Lemma.\\ 
\indent By compactness of $\mathbb{S}_{(Y,Z)}$ and continuity of $f$, $\partial f/\partial z$, I have that there exist finite constant $d_9$ such that
\begin{align}
    &\sup_{(\delta,z)\in [0,1]\times\mathbb{S}_\tau}\left|\Phi_{1,(j)}(f)\right|=\sup_{(\delta,z)\in [0,1]\times\mathbb{S}_\tau}\left|\int_{-\infty}^{\alpha(f)}\frac{\partial f(y,z)}{\partial z_j}dy\right|\leq d_9 \sup_{(y,z)\in\mathbb{S}_{(Y,Z)} }\left|\frac{\partial f(y,z)}{\partial z_j}\right|<\infty \notag \\
    &\sup_{(\delta,z)\in [0,1]\times\mathbb{S}_\tau }\left|\Phi_2(f)\right|\leq\sup_{(y,z)\in\mathbb{S}_{(Y,Z)} }\left|f(y,z)\right|<\infty
\end{align}
Also, by the definition of $\mathbb{S}_\tau$, 
\begin{align}
    &\inf_{(\delta,z)\in \mathcal{S}_\tau}\Phi^2_2(f)\geq\inf_{(y,z)\in\mathbb{S}_\tau }f^2(y,z)>\tau^2. \label{finite}
\end{align}
For $\sup_{(y,z)\in\mathbb{S}_{(Y,Z)}} |h|\leq\min\left\{\tau/2,1\right\}$, I have
\begin{align}
    \inf_{(y,z)\in\mathbb{S}_\tau }\Phi_2(f+h)&=\inf_{(y,z)\in\mathbb{S}_\tau} f(\alpha(f+h),z)+h(\alpha(f+h),z)\notag\\
   &>\tau-\frac{\tau}{2}=\frac{\tau}{2}. \label{denomtau2}
\end{align}
By Lemma \ref{file3}, \ref{file4} there exist finite constants $d_{10}, d_{11}, d_{12}, d_{13}, d_{14}, d_{15}$ such that
\begin{footnotesize}
\begin{align}
    &\sup_{(\delta,z)\in \mathcal{S}_\tau}\left|\delta\frac{\partial \tilde h(z)}{\partial z_j}-D\Phi_{1,(j)}(f;h)\right|\leq d_{10}\sup_{(y,z)\in\mathbb{S}_{(Y,Z)}} |h|+d_{10}\sup_{(y,z)\in\mathbb{S}_{Y,Z}}\left|\frac{\partial h(y,z)}{\partial z_j}\right|\notag\\
    &\sup_{(\delta,z)\in \mathcal{S}_\tau}\left|\left(\delta\frac{\partial \tilde f(z)}{\partial z_j}-\Phi_{1,(j)}(f)\right)D\Phi_2(f;h)\right|\leq d_{11}\sup_{(y,z)\in\mathbb{S}_{(Y,Z)}} |h|\notag\\
    &\sup_{(\delta,z)\in \mathcal{S}_\tau}\left|\left(\delta \frac{\partial \tilde f(z)}{\partial z_j}-\Phi_{1,(j)}(f)\right)R\Phi_2(f;h)\right|\leq d_{12}\sup_{(y,z)\in\mathbb{S}_{(Y,Z)}} |h|^{2}\notag\\
    &\quad\quad\quad\quad\quad\quad\quad\quad\quad\quad\quad\quad\quad\quad\quad\quad\quad\quad\quad+d_{12}\sup_{(y,z)\in\mathbb{S}_{(Y,Z)}} |h| \sup _{(y, z) \in \mathbb{S}_{Y, Z}}\left|\frac{\partial h(y, z)}{\partial y}\right|\notag\\
    &\sup_{(\delta,z)\in \mathcal{S}_\tau}\left|\Delta\Phi_2(f;h)\right|\leq d_{13}\sup_{(y,z)\in\mathbb{S}_{(Y,Z)}} |h|\notag\\
    &\sup_{(\delta,z)\in \mathcal{S}_\tau}\left|\delta\frac{\partial \tilde h(z)}{\partial z_j}-\Delta\Phi_{1,(j)}(f;h)\right|\leq d_{14}\sup_{(y,z)\in\mathbb{S}_{(Y,Z)}} |h|+d_{14} \sup _{(y, z) \in \mathbb{S}_{Y, Z}}\left|\frac{\partial h(y, z)}{\partial z_j}\right|\notag\\
    &\sup_{(\delta,z)\in \mathcal{S}_\tau}\left|\left(\delta \frac{\partial \tilde f(z)}{\partial z_j}-\Phi_{1}(f)\right) \Delta \Phi_{2}(f ; h)\right|\leq d_{15}\sup_{(y,z)\in\mathbb{S}_{(Y,Z)}} |h|.\label{numcomponents}
\end{align}
\end{footnotesize}
Combining results from (\ref{finite}), (\ref{denomtau2}),  and (\ref{numcomponents}), I have that for $\sup_{(y,z)\in\mathbb{S}_{(Y,Z)}} |h|\leq\min\left\{\tau/2,1\right\}$ and some finite constants $d_{16}$, $d_{17}$
\begin{footnotesize}
\begin{align*}
    &\sup_{(\delta,z)\in \mathcal{S}_\tau}\left|D\Phi_{(j)}(f;h)\right|\leq d_{16}\sup_{(y,z)\in\mathbb{S}_{(Y,Z)}} |h|+d_{16}\sup_{(y,z)\in\mathbb{S}_{Y,Z}}\left|\frac{\partial h(y,z)}{\partial z_j}\right|\\
    &\sup_{(\delta,z)\in \mathcal{S}_\tau}\left|R\Phi_{(j)}(f;h)\right|\leq d_{17}\sup_{(y,z)\in\mathbb{S}_{(Y,Z)}} |h|^2+d_{17}\sup_{(y,z)\in\mathbb{S}_{(Y,Z)}} |h|\sup_{(y,z)\in\mathbb{S}_{Y,Z}}\left|\frac{\partial h(y,z)}{\partial y}\right|\\
    &\quad\quad\quad\quad\quad\quad\quad\quad\quad\quad\quad\quad+d_{17}\sup_{(y,z)\in\mathbb{S}_{(Y,Z)}} |h|\sup_{(y,z)\in\mathbb{S}_{Y,Z}}\left|\frac{\partial h(y,z)}{\partial z_j}\right|.
\end{align*}
\end{footnotesize}

Let $c_1\equiv\max\{d_{16},d_{17}\}$ gives the desired result. By the same logic, it can be shown that
\begin{align*}
     &\sup_{(y,z)\in\mathbb{S}_\tau }\left|D\Psi_{(j)}(f;h)\right|\leq d_{18}\sup_{(y,z)\in\mathbb{S}_{(Y,Z)}} |h|+d_{18}\sup_{(y,z)\in\mathbb{S}_{(Y,Z)} }\left|\frac{\partial h(y,z)}{\partial z_{j}}\right|\\
     &\sup_{(y,z)\in\mathbb{S}_\tau }\left|R\Psi_{(j)}(f;h)\right|\leq d_{19}\sup_{(y,z)\in\mathbb{S}_{(Y,Z)}} |h|^2+d_{19}\sup_{(y,z)\in\mathbb{S}_{(Y,Z)}} |h|\sup_{(y,z)\in\mathbb{S}_{(Y,Z)} }\left|\frac{\partial h(y,z)}{\partial z_{j}}\right|\\
     &\sup_{(y,z)\in\mathbb{S}_\tau }\left|D\tilde\Psi(f;h)\right|\leq d_{20}\sup_{(y,z)\in\mathbb{S}_{(Y,Z)}} |h|+d_{20}\sup_{(y,z)\in\mathbb{S}_{(Y,Z)} }\left|\frac{\partial h(y,z)}{\partial z_{j+1}}\right|\\
     &\sup_{(y,z)\in\mathbb{S}_\tau }\left|R\tilde\Psi(f;h)\right|\leq d_{21}\sup_{(y,z)\in\mathbb{S}_{(Y,Z)}} |h|^2+d_{21}\sup_{(y,z)\in\mathbb{S}_{(Y,Z)}} |h|\sup_{(y,z)\in\mathbb{S}_{(Y,Z)} }\left|\frac{\partial h(y,z)}{\partial z_{j+1}}\right|.
\end{align*}
Let $c_2\equiv\max\{d_{18},d_{19}\}$, $c_3\equiv\max\{d_{20},d_{21}\}$ gives the desired result. 
\end{proof}

\begin{lemma}\label{file7} For any value $(y,z)\in \mathbb{S}_{(Y,Z)}$, and any value $\delta\in[0,1]$, define functional $\Xi$ and $\tilde\Xi$ on $\mathrm{F}$ by $\Xi(g)\equiv\Phi_{(1)}(g)+\frac{\Phi_{(2)}(g)}{\tilde\Psi(g)}$ and $\tilde\Xi(g)\equiv\Psi_{(1)}(g)+\frac{\Psi_{(2)}(g)}{\tilde\Psi(g)}$. Then I have that
\begin{align*}
    \Xi(f+h)-\Xi(f)=D\Xi(f;h)+R\Xi(f;h)\\
    \tilde\Xi(f+h)-\tilde\Xi(f)=D\tilde\Xi(f;h)+R\tilde\Xi(f;h),
\end{align*}
and that there exist constants $e_1, e_2$ such that 
\begin{footnotesize}
\begin{align*}
    \sup_{(\delta, x,w)\in \mathcal{S}_\tau}\left|D\Xi(f;h)\right|&\leq \frac{e_1}{\tau^2}\|h\|+e_1\sup_{(y,x,w)\in\mathbb{S}_{Y,X,W}}\left|\frac{\partial h(y,x,w)}{\partial x}\right|+\frac{e_1}{\tau^2}\sup_{(y,x,w)\in\mathbb{S}_{Y,X,W}}\left|\frac{\partial h(y,x,w)}{\partial w}\right|\\
    \sup_{(\delta, x,w)\in \mathcal{S}_\tau}\left|R\Xi(f;h)\right|&\leq\frac{e_1}{\tau^3}\|h\|^2+\frac{e_1}{\tau^3}\|h\|\sup_{(y,x,w)\in\mathbb{S}_{Y,X,W}}\left|\frac{\partial h(y,x,w)}{\partial w}\right|\\
    &\quad\quad+\frac{e_1}{\tau}\|h\|\sup_{(y,x,w)\in\mathbb{S}_{Y,X,W}}\left|\frac{\partial h(y,x,w)}{\partial y}\right|+e_1\|h\|\sup_{(y,x,w)\in\mathbb{S}_{Y,X,W}}\left|\frac{\partial h(y,x,w)}{\partial x}\right|\\
    &\quad\quad+\frac{e_1}{\tau^3}\sup_{(y,x,w)\in\mathbb{S}_{Y,X,W}}\left|\frac{\partial h(y,x,w)}{\partial w}\right|^2\\
    \sup_{(y, x,w)\in \mathbb{S}_{\tau}}\left|D\tilde\Xi(f;h)\right|&\leq \frac{e_2}{\tau^2}\|h\|+e_2\sup_{(y,x,w)\in\mathbb{S}_{Y,X,W}}\left|\frac{\partial h(y,x,w)}{\partial x}\right|+\frac{e_2}{\tau^2}\sup_{(y,x,w)\in\mathbb{S}_{Y,X,W}}\left|\frac{\partial h(y,x,w)}{\partial w}\right|\\
    \sup_{(y, x,w)\in \mathbb{S}_{\tau}}\left|R\tilde \Xi(f;h)\right|&\leq\frac{e_2}{\tau^3}\|h\|^2+\frac{e_2}{\tau^3}\|h\|\sup_{(y,x,w)\in\mathbb{S}_{Y,X,W}}\left|\frac{\partial h(y,x,w)}{\partial w}\right|\\
    &\quad\quad+e_2\|h\|\sup_{(y,x,w)\in\mathbb{S}_{Y,X,W}}\left|\frac{\partial h(y,x,w)}{\partial x}\right|+\frac{e_2}{\tau^3}\sup_{(y,x,w)\in\mathbb{S}_{Y,X,W}}\left|\frac{\partial h(y,x,w)}{\partial w}\right|^2.
\end{align*}
\end{footnotesize}

\end{lemma}

\begin{proof}
\begin{footnotesize}
\begin{align*}
    &\quad\Xi(f+h)-\Xi(f)\\
    &=\left(\Phi_{(1)}(f+h)+\frac{\Phi_{(2)}(f+h)}{\tilde\Psi_{(1)}(f+h)}\right)-\left(\Phi_{(1)}(f)+\frac{\Phi_{(2)}(f)}{\tilde\Psi_{(1)}(f)}\right)\\
    &=\Delta\Phi_{(1)}(f;h)+\frac{D\Phi_{(2)}(f;h)\tilde\Psi_{(1)}(f)-D\tilde\Psi_{(1)}(f;h)\Phi_{(2)}(f)}{\tilde\Psi_{(1)}^2(f)}\\
    &\quad\quad-\frac{D\Phi_{(2)}(f;h)\tilde\Psi_{(1)}(f)\Delta\tilde\Psi_{(1)}(f;h)-D\tilde\Psi_{(1)}(f;h)\Phi_{(2)}(f)\Delta\tilde\Psi_{(1)}(f;h)}{\tilde\Psi_{(1)}^2(f)\tilde\Psi_{(1)}(f+h)}\\
    &\quad\quad+\frac{R\Phi_{(2)}(f;h)\tilde\Psi_{(1)}(f)-R\tilde\Psi_{(1)}(f;h)\Phi_{(2)}(f)}{\tilde\Psi_{(1)}(f+h)\tilde\Psi_{(1)}(f)}
\end{align*}
\begin{align*}
    &\quad\tilde\Xi(f+h)-\tilde\Xi(f)\\
    &=\left(\Psi_{(1)}(f+h)+\frac{\Psi_{(2)}(f+h)}{\tilde\Psi_{(1)}(f+h)}\right)-\left(\Psi_{(1)}(f)+\frac{\Psi_{(2)}(f)}{\tilde\Psi_{(1)}(f)}\right)\\
    &=\Delta\Psi_{(1)}(f;h)+\frac{D\Psi_{(2)}(f;h)\tilde\Psi_{(1)}(f)-D\tilde\Psi_{(1)}(f;h)\Psi_{(2)}(f)}{\tilde\Psi_{(1)}^2(f)}\\
    &\quad\quad-\frac{D\Psi_{(2)}(f;h)\tilde\Psi_{(1)}(f)\Delta\tilde\Psi_{(1)}(f;h)-D\tilde\Psi_{(1)}(f;h)\Psi_{(2)}(f)\Delta\tilde\Psi_{(1)}(f;h)}{\tilde\Psi_{(1)}^2(f)\tilde\Psi_{(1)}(f+h)}\\
    &\quad\quad+\frac{R\Psi_{(2)}(f;h)\tilde\Psi_{(1)}(f)-R\tilde\Psi_{(1)}(f;h)\Psi_{(2)}(f)}{\tilde\Psi_{(1)}(f+h)\tilde\Psi_{(1)}(f)}
\end{align*}
Define
\begin{align*}
    D\Xi(f;h)&\equiv D\Phi_{(1)}(f;h)+\frac{D\Phi_{(2)}(f;h)\tilde\Psi_{(1)}(f)-D\tilde\Psi_{(1)}(f;h)\Phi_{(2)}(f)}{\tilde\Psi_{(1)}^2(f)}\\
    R\Xi(f;h)&\equiv R\Phi_{(1)}(f;h)-\frac{D\Phi_{(2)}(f;h)\tilde\Psi_{(1)}(f)\Delta\tilde\Psi_{(1)}(f;h)-D\tilde\Psi_{(1)}(f;h)\Phi_{(2)}(f)\Delta\tilde\Psi_{(1)}(f;h)}{\tilde\Psi_{(1)}^2(f)\tilde\Psi_{(1)}(f+h)}\\
    &\quad\quad+\frac{R\Phi_{(2)}(f;h)\tilde\Psi_{(1)}(f)-R\tilde\Psi_{(1)}(f;h)\Phi_{(2)}(f)}{\tilde\Psi_{(1)}(f+h)\tilde\Psi_{(1)}(f)}\\
    D\tilde\Xi(f;h)&\equiv D\Psi_{(1)}(f;h)+\frac{D\Psi_{(2)}(f;h)\tilde\Psi_{(1)}(f)-D\tilde\Psi_{(1)}(f;h)\Psi_{(2)}(f)}{\tilde\Psi_{(1)}^2(f)}\\
    R\tilde\Xi(f;h)&\equiv R\Psi_{(1)}(f;h)-\frac{D\Psi_{(2)}(f;h)\tilde\Psi_{(1)}(f)\Delta\tilde\Psi_{(1)}(f;h)-D\tilde\Psi_{(1)}(f;h)\Psi_{(2)}(f)\Delta\tilde\Psi_{(1)}(f;h)}{\tilde\Psi_{(1)}^2(f)\tilde\Psi_{(1)}(f+h)}\\
    &\quad\quad+\frac{R\Psi_{(2)}(f;h)\tilde\Psi_{(1)}(f)-R\tilde\Psi_{(1)}(f;h)\Psi_{(2)}(f)}{\tilde\Psi_{(1)}(f+h)\tilde\Psi_{(1)}(f)}.
\end{align*}
Then 
\begin{align*}
    \Xi(f+h)-\Xi(f) = D\Xi(f;h)+R\Xi(f;h)\\
    \tilde\Xi(f+h)-\tilde\Xi(f) = D\tilde\Xi(f;h)+R\tilde\Xi(f;h).
\end{align*}
Note that there exist some constant $d_{14}$ such that 
\begin{align*}
    \sup_{(\delta, x,w)\in \mathcal{S}_\tau}\left|\Phi_{(2)}(f)\right|&=\sup_{(\delta, x,w)\in \mathcal{S}_\tau}\left|\frac{\delta \frac{\partial \tilde f(x,w)}{\partial w}-\Phi_{1, (2)}(f)}{\Phi_2(f)}\right|\leq\frac{\sup_{(\delta, x,w)\in \mathcal{S}_\tau}\left|\delta \frac{\partial \tilde f(x,w)}{\partial w}-\Phi_{1, (2)}(f)\right|}{\inf_{(\delta, x,w)\in \mathcal{S}_\tau}\left|f(\alpha(f),x,w)\right|}\leq d_{14}\\
    \sup_{(y, x,w)\in \mathbb{S}_{\tau}}\left|\Psi_{(2)}(f)\right|&=\sup_{(y, x,w)\in \mathbb{S}_{\tau}}\left|\frac{\Lambda(f) \frac{\partial \tilde f(x,w)}{\partial w}-\Psi_{1, (2)}(f)}{f(y,x,w)}\right|\\
    &\leq\frac{\sup_{(y, x,w)\in \mathbb{S}_{\tau}}\left|\Lambda(f) \frac{\partial \tilde f(x,w)}{\partial w}-\Psi_{1, (2)}(f)\right|}{\inf_{(y, x,w)\in \mathbb{S}_{\tau}}\left|f(y,x,w)\right|}\leq d_{14}.
\end{align*}
Application of Lemma \ref{file2} yields the desired result.
\end{footnotesize}
\end{proof}

\begin{lemma}\label{file8}
In Lemma \ref{file1}, let the dimension of vector $L=2$. For any value $(y,z)\in \mathbb{S}_{(Y,Z)}$, and any value $\delta\in[0,1]$, define functional $\tilde\Gamma$ on $\mathrm{F}$ by
\begin{align*}
    \tilde\Gamma(g)\equiv\int\int_0^1\omega(x,w^*,s(r(x,w^*),\delta))\Xi(g)\frac{\int g(y,x.w) dy}{\int\int g(y,x,w) dy dw}d\delta dw,
\end{align*} 
where $\omega()$ is a known\footnote{The same logic carries over if $\omega()$ is estimated. Take the example mentioned in the main text for instance. When $\tilde \omega(x,w^*)$ is a constant equal to $\left[(\tau_u-\tau_l)\int_{\underline{w}^*}^{\bar w^*} f_{W^*|X=x}(w^*)dw^*\right]^{-1}$, and $f_{W^*|X=x}(w^*)$ is estimated by a plug-in estimator, the same argument in this proof will follow by replacing $\int\int f(y,x,w)dy dw$ with $\int_{\underline{w}^*}^{\bar w^*}\int f(y,x,w)dy dw$ in the definition of $\Gamma$.} compactly supported weighting function as defined in (\ref{weightedLAR}), and $s()$ and $r()$ are as defined in (\ref{epsilondecomposition}) and (\ref{vfunction}).
Then I have that
\begin{align*}
    \tilde\Gamma(f+h)-\tilde\Gamma(f) = D\tilde\Gamma(f;h)+R\tilde\Gamma(f;h),
\end{align*}
and there exists a constant $e_4$ such that
\begin{footnotesize}
\begin{align*}
    \sup_{x\in \mathbb S_{X}}\left|D\tilde\Gamma(f;h)\right|&\leq e_4\|h\|+e_4\sup_{(y,x,w)\in\mathbb{S}_{Y,X,W}}\left|\frac{\partial h(y,x,w)}{\partial x}\right|+e_4\sup_{(y,x,w)\in\mathbb{S}_{Y,X,W}}\left|\frac{\partial h(y,x,w)}{\partial w}\right|\\
    \sup_{x\in \mathbb S_{X}}\left|R\tilde\Gamma(f;h)\right|&\leq e_4\|h\|^2+e_4\|h\|\sup_{(y,x,w)\in\mathbb{S}_{Y,X,W}}\left|\frac{\partial h(y,x,w)}{\partial w}\right|+e_4\|h\|\sup_{(y,x,w)\in\mathbb{S}_{Y,X,W}}\left|\frac{\partial h(y,x,w)}{\partial y}\right|.
\end{align*}
\end{footnotesize}

\end{lemma}
\begin{proof}
Define functional $\Gamma$ on $\mathrm{F}$ by $\Gamma(g)\equiv\Xi(g)\frac{\int g(y,x.w) dy}{\int\int g(y,x,w) dy dw}$, then
\begin{align*}
    \Gamma(f+h)-\Gamma(f) &= \frac{\Xi(f+h)\int f(y,x.w)+h(y,x.w) dy}{\int\int f(y,x,w)+h(y,x.w) dy dw}-\frac{\Xi(f)\int f(y,x.w) dy}{\int\int f(y,x,w) dy dw}\\
    &\equiv\frac{N'}{D'}-\frac{N}{D}
\end{align*}
where I've defined 
\begin{align*}
    N'&\equiv \Xi(f+h)\int f(y,x.w)+h(y,x.w) dy\\
    N &\equiv \Xi(f)\int f(y,x.w) dy\\
    D' &\equiv \int\int f(y,x,w)+h(y,x.w) dy dw\\
    D &\equiv \int\int f(y,x,w) dy dw
\end{align*}
Note that
\begin{align*}
    N'-N &= D\Xi(f;h)\int f(y,x,w)dy +\Xi(f)\int h(y,x,w)dy\\
    &\quad\quad+D\Xi(f;h)\int h(y,x,w)dy+R\Xi(f;h)\int (f(y,x,w)+h(y,x,w))dy\\
    &\equiv DN+RN\\
    D'-D &= \int\int h(y,x,w) dy dw 
\end{align*}
where $DN\equiv D\Xi(f;h)\int f(y,x,w)dy +\Xi(f)\int h(y,x,w)dy$ and $RN\equiv D\Xi(f;h)\int h(y,x,w)dy+R\Xi(f;h)\int (f(y,x,w)+h(y,x,w))dy$, and $D\Xi(f;h)$ and $R\Xi(f;h)$ are as defined in Lemma \ref{file7}. Then by the same logic as Lemma \ref{file2}, 
\begin{scriptsize}
\begin{align*}
    &\quad\Gamma(f+h)-\Gamma(f) \\
    &=\frac{\left(D\Xi(f;h)\int f(y,x,w)dy +\Xi(f)\int h(y,x,w)dy\right)\int\int f(y,x,w) dy dw-\Xi(f)\int f(y,x.w) dy\left(\int\int h(y,x.w) dy dw\right)}{\left(\int\int f(y,x,w) dy dw\right)^2}\\
    &\quad+\frac{\int\int f(y,x,w) dy dw\times \left(D\Xi(f;h)\int h(y,x,w)dy+R\Xi(f;h)\int (f(y,x,w)+h(y,x,w))dy\right)}{\left(\int\int f(y,x,w) dy dw\right)^2}\\
    &\quad-\int\int h(y,x,w) dy dw\times\\
    &\quad\frac{\left(\left(\Xi(f+h)\int f(y,x.w)+h(y,x.w) dy\right)\int\int f(y,x,w) dy dw-\Xi(f)\int f(y,x.w) dy\left(\int\int f(y,x,w)+h(y,x.w) dy dw\right)\right)}{\left(\int\int f(y,x,w)+h(y,x.w) dy dw\right)\left(\int\int f(y,x,w) dy dw\right)^2}\\
    &=D\Gamma(f;h)+R\Gamma(f;h).
\end{align*}
\end{scriptsize}
where I've defined
\begin{scriptsize}
\begin{align*}
    &D\Gamma(f;h)\\
    &\quad\equiv \frac{\left(D\Xi(f;h)\int f(y,x,w)dy +\Xi(f)\int h(y,x,w)dy\right)\int\int f(y,x,w) dy dw-\Xi(f)\int f(y,x.w) dy\left(\int\int h(y,x.w) dy dw\right)}{\left(\int\int f(y,x,w) dy dw\right)^2}\\
   &R\Gamma(f;h)\\
    &\quad\equiv +\frac{\int\int f(y,x,w) dy dw\times \left(D\Xi(f;h)\int h(y,x,w)dy+R\Xi(f;h)\int (f(y,x,w)+h(y,x,w))dy\right)}{\left(\int\int f(y,x,w) dy dw\right)^2}\\
    &\quad-\int\int h(y,x,w) dy dw\times\\
    &\quad \frac{\left(\left(\Xi(f+h)\int f(y,x.w)+h(y,x.w) dy\right)\int\int f(y,x,w) dy dw-\Xi(f)\int f(y,x.w) dy\left(\int\int f(y,x,w)+h(y,x.w) dy dw\right)\right)}{\left(\int\int f(y,x,w)+h(y,x.w) dy dw\right)\left(\int\int f(y,x,w) dy dw\right)^2}.
\end{align*}
\end{scriptsize}
By the boundedness of support, the definition of $\mathcal{S}_\tau$ and Lemma \ref{file7}, there exists a constant $e_3$ such that
\begin{footnotesize}
\begin{align*}
    \sup_{(\delta, x,w)\in\mathcal{S}_\tau}\left|D\Gamma(f;h)\right|&\leq e_3\|h\|+e_3\sup_{(y,x,w)\in\mathbb{S}_{Y,X,W}}\left|\frac{\partial h(y,x,w)}{\partial x}\right|+e_3\sup_{(y,x,w)\in\mathbb{S}_{Y,X,W}}\left|\frac{\partial h(y,x,w)}{\partial w}\right|\\
    \sup_{(\delta, x,w)\in\mathcal{S}_\tau}\left|R\Gamma(f;h)\right|&\leq e_3\|h\|^2+e_3\|h\|\sup_{(y,x,w)\in\mathbb{S}_{Y,X,W}}\left|\frac{\partial h(y,x,w)}{\partial w}\right|+e_3\|h\|\sup_{(y,x,w)\in\mathbb{S}_{Y,X,W}}\left|\frac{\partial h(y,x,w)}{\partial y}\right|
\end{align*}
\end{footnotesize}
\indent Denote the support of $\omega(x,w^*,s(r(x,w^*),\delta))$ as $\mathcal{M}$, then by the argument stated at the beginning of Appendix A, there exists $\tau>0$ such that $\mathcal{M}\subset\mathcal{S}_\tau$ so that I have 
\begin{footnotesize}
\begin{align*}
    \sup_{(\delta, x,w)\in\mathcal{M}}\left|D\Gamma(f;h)\right|&\leq e_3\|h\|+e_3\sup_{(y,x,w)\in\mathbb{S}_{Y,X,W}}\left|\frac{\partial h(y,x,w)}{\partial x}\right|+e_3\sup_{(y,x,w)\in\mathbb{S}_{Y,X,W}}\left|\frac{\partial h(y,x,w)}{\partial w}\right|\\
    \sup_{(\delta, x,w)\in\mathcal{M}}\left|R\Gamma(f;h)\right|&\leq e_3\|h\|^2+e_3\|h\|\sup_{(y,x,w)\in\mathbb{S}_{Y,X,W}}\left|\frac{\partial h(y,x,w)}{\partial w}\right|+e_3\|h\|\sup_{(y,x,w)\in\mathbb{S}_{Y,X,W}}\left|\frac{\partial h(y,x,w)}{\partial y}\right|.
\end{align*}
\end{footnotesize}
By definition, I have
\begin{footnotesize}
\begin{align*}
    &\quad\int\int_0^1\omega(x,w^*,s(r(x,w^*),\delta))\Gamma(f+h)d\delta dw-\int\int_0^1\omega(x,w^*,s(r(x,w^*),\delta))\Gamma(f)d\delta dw\\
    &=\int\int_0^1\omega(x,w^*,s(r(x,w^*),\delta))D\Gamma(f;h)d\delta dw+\int\int_0^1\omega(x,w^*,s(r(x,w^*),\delta))R\Gamma(f;h)d\delta dw\\
    &\equiv D\tilde\Gamma(f;h)+R\tilde\Gamma(f;h)
\end{align*}    
\end{footnotesize}
where I've defined $D\tilde\Gamma(f;h)\equiv\int\int_0^1\omega(x,w^*,s(r(x,w^*),\delta))D\Gamma(f;h)d\delta dw$ and $R\tilde\Gamma(f;h)\equiv\int\int_0^1\omega(x,w^*,s(r(x,w^*),\delta))R\Gamma(f;h)d\delta dw$. By the compactness of $\mathcal{M}$, there exist a constant $e_4$ such that
\begin{footnotesize}
\begin{align*}
    \sup_{x\in \mathbb S_{X}}\left|D\tilde\Gamma(f;h)\right|&\leq  e_4\|h\|+e_4\sup_{(y,x,w)\in\mathbb{S}_{Y,X,W}}\left|\frac{\partial h(y,x,w)}{\partial x}\right|+e_4\sup_{(y,x,w)\in\mathbb{S}_{Y,X,W}}\left|\frac{\partial h(y,x,w)}{\partial w}\right|\\
    \sup_{x\in \mathbb S_{X}}\left|R\tilde\Gamma(f;h)\right|&\leq e_4\|h\|^2+e_4\|h\|\sup_{(y,x,w)\in\mathbb{S}_{Y,X,W}}\left|\frac{\partial h(y,x,w)}{\partial w}\right|+e_4\|h\|\sup_{(y,x,w)\in\mathbb{S}_{Y,X,W}}\left|\frac{\partial h(y,x,w)}{\partial y}\right|.
\end{align*}
\end{footnotesize}
\end{proof}

\subsection{Appendix B: Proofs of Theorems and Lemmas in the Main Text}
\begin{proof}[\proofname \ of Lemma \ref{deltaindependent}]
\indent Since $W^*\perp (\epsilon, \eta)$, then $\epsilon\perp W^* \mid \eta$. Define $\delta = F_{\epsilon\mid \eta}(\epsilon)$. Then I can write $\epsilon=F_{\epsilon\mid \eta}^{-1}(\delta)=:s(\eta,\delta)$. By Assumption \ref{Fepsilon|eta}, $s$ is strictly increasing in $\delta$. To see that $\delta \perp (X,W^*)$ and is uniform (0,1), note that
\begin{align*}
    F_{\delta\mid W^*=w^*, \eta=\bar\eta}(t)&= Pr\left(F_{\epsilon\mid\eta}(\epsilon)\leq t\mid W^*=w^*, \eta=\bar\eta\right)\\
    &=Pr\left(\epsilon\leq F^{-1}_{\epsilon\mid\eta=\bar\eta}(t)\mid W^*=w^*, \eta=\bar\eta\right)\\
    &=Pr\left(\epsilon\leq F^{-1}_{\epsilon\mid\eta=\bar\eta}(t)\mid  \eta=\bar\eta\right)\\
    &=F_{\epsilon\mid\eta=\bar\eta}(F^{-1}_{\epsilon\mid\eta=\bar\eta}(t))=t,
\end{align*}
which says that the conditional distribution of $\delta$ given $W^*, \eta$ is $U(0,1)$ regardless of the values of $W^*$ and $\eta$. In other words, $\delta\sim U(0,1)$ and is independent of $(W^*, \eta)$. Since $X$ is a function of $(W^*, \eta)$, $\delta$ is independent of $X$ as well.
\end{proof}

\begin{proof}[\proofname \ of Lemma \ref{commonsupport}]
(i)First I show that $ \mathbb{S}_{\epsilon\mid X=x}\subseteq\mathbb{S}_{\mathbf{s}\mid X=x}$. For any $\bar\epsilon\in  \mathbb{S}_{\epsilon\mid X=x}$, let $\bar\delta\equiv F_{\epsilon\mid\eta=r(x,\bar w^*)}(\bar\epsilon)$ for some $\bar w^*\in \mathbb{S}_{W^*\mid X=x}$. Then $\bar\delta\in[0,1]$ and thus $\bar \epsilon \equiv s(r(x, \bar w^*), \bar\delta)\in  \mathbb{S}_{\mathbf{s}\mid X=x}$. \\
\indent (ii) Then I show that 
$\mathbb{S}_{\mathbf{s}\mid X=x} \subseteq \mathbb{S}_{\epsilon\mid X=x}$. By definition of $\mathbb{S}_{\mathbf{s}\mid X=x}$, for each $\tilde s \in \mathbb{S}_{\mathbf{s}\mid X=x}$, there is some $(\tilde w^*, \tilde\delta)\in \mathbb{S}_{W^*\mid X=x}\times [0,1]$ such that $\tilde s=s(r(x,\tilde w^*), \tilde\delta)$ . Denote the upper and lower bound (potentially infinity) of $\mathbb{S}_{\epsilon\mid X=x}$ as $\bar e_{x}$ and $\underline{e}_{x}$, respectively. From Lemma \ref{deltaindependent}, I know that $\epsilon\perp W^*\mid \eta$. Since $X$ is a function of $W^*, \eta$, I have that $\epsilon\perp X\mid \eta$. Then if $\bar e_x\in \mathbb{S}_{\epsilon\mid X=x}$, then $Pr\left(\epsilon\leq\bar e_x\mid \eta=r(x, \tilde w^*)\right)=Pr\left(\epsilon\leq\bar e_x\mid \eta=r(x, \tilde w^*), X=x\right)=1$, which implies $F^{-1}_{\epsilon\mid \eta=r(x, \tilde w^*)}(1)\leq \bar e_x$. If $\bar e_x\notin \mathbb{S}_{\epsilon\mid X=x}$, then $Pr\left(\epsilon<\bar e_x\mid \eta=r(x, \tilde w^*)\right)=Pr\left(\epsilon<\bar e_x\mid \eta=r(x, \tilde w^*), X=x\right)=1$, which implies $F^{-1}_{\epsilon\mid \eta=r(x, \tilde w^*)}(1)<\bar e_x$. Similarly,  $F^{-1}_{\epsilon\mid \eta=r(x, \tilde w^*)}(0)\geq \bar e_x$, and the inequality is strict if $\underline{e}_x\notin\mathbb{S}_{\epsilon\mid X=x}$ Since $\tilde\delta\in[0,1]$, by Assumption \ref{Fepsilon|eta}, $\underline{e}_x\leq F^{-1}_{\epsilon\mid \eta=r(x, \tilde w^*)}(\tilde\delta)\leq \bar e_x$, or equivalently, $\underline{e}_x\leq s(r(x, \tilde w^*),\tilde\delta)\leq \bar e_x$, and the inequalities are strict if $\underline{e}_x\notin\mathbb{S}_{\epsilon\mid X=x}$ or $\bar{e}_x\notin\mathbb{S}_{\epsilon\mid X=x}$, respectively. Thus $\tilde s\in \mathbb{S}_{\epsilon\mid X=x}$. Combining results from (i) and (ii) yields the desired conclusion.
\end{proof}

\begin{proof}[\proofname \ of Lemma \ref{dcquantile}]
\begin{align*}
    \delta &= F_{Y|X=x, W^*=w^*}\left(F^{-1}_{Y|X=x, W^*=w^*}(\delta)\right)\\
    &=\frac{\int_{-\infty}^{F^{-1}_{Y|X=x, W^*=w^*}(\delta)}f_{Y,X,W^*}(y,x,w^*)dy}{\int_{-\infty}^\infty f_{Y,X,W^*}(y,x,w^*)dy}\\
    \implies\delta f_{X,W^*}(x,w^*) &=\int_{-\infty}^{F^{-1}_{Y|X=x, W^*=w^*}(\delta)}f_{Y,X,W^*}(y,x,w^*)dy
\end{align*}
Taking derivatives on both sides with respect to $w^*$ yields
\begin{footnotesize}
\begin{align*}
    &f_{Y,X,W^*}\left(F^{-1}_{Y|X=x, W^*=w^*}(\delta), x,w^*\right)\frac{\partial F^{-1}_{Y|X=x, W^*=w^*}(\delta)}{\partial w^*}+\int_{-\infty}^{F^{-1}_{Y|X=x, W^*=w^*}(\delta)}\frac{\partial f_{Y,X,W^*}(y,x,w^*)}{\partial w^*}dy \\
    &= \delta \frac{\partial f_{X,W^*}(x,w^*)}{\partial w^*}
\end{align*}
\end{footnotesize}
This implies
\begin{footnotesize}
\begin{align*}
  \frac{\partial F^{-1}_{Y|X=x, W^*=w^*}(\delta)}{\partial w^*} = \frac{\delta \frac{\partial f_{X,W^*}(x,w^*)}{\partial w^*}-\int_{-\infty}^{F^{-1}_{Y|X=x, W^*=w^*}(\delta)}\frac{\partial f_{Y,X,W^*}(y,x,w^*)}{\partial w^*}dy}{f_{Y,X,W^*}\left(F^{-1}_{Y|X=x, W^*=w^*}(\delta),x,w^*\right)}.
\end{align*}
\end{footnotesize}
The conclusion for $\frac{\partial F^{-1}_{Y|X=x, W^*=w^*}(\delta)}{\partial x}$ follows the same logic.
\end{proof}

\begin{proof}[\proofname \ of Lemma \ref{convolution thm}]
Assumption \ref{posdenom}, \ref{cpctcontdiff}, and \ref{E<infty} ensure the existence of 
\begin{align*}
    &\phi_{f_{Y,X,W_2}(y,x,\cdot)}(t)=E\left[e^{\mathbf{i}tW_2}\mid Y=y, X=x\right]f_{Y,X}(y,x)\\
    &\phi_{f^{(\lambda_{2,1},\lambda_{2,2})}_{Y,X,W_2}(y,x,\cdot)}(t)=\int e^{\mathbf{i}tw_2}\frac{\partial^{\lambda_2} f_{Y,X,W_2}(y,x,w_2)}{\partial y^\lambda_{2,1}\partial x^\lambda_{2,2}}dw_2\quad\text{and thus}\\
    &g_{\lambda_1,\lambda_{2,1}, \lambda_{2,2}}(y,x,w^*,h_1)\equiv\int \frac{1}{h_1}K\left(\frac{\tilde w^*-w^*}{h_1}\right)g_{\lambda_1,\lambda_{2,1}, \lambda_{2,2}}(y,x,\tilde w^*)d\tilde w^*\\
    &=\int \frac{1}{h_1}K\left(\frac{\tilde w^*-w^*}{h_1}\right)\left(\frac{1}{2\pi}\int\left(-\mathbf{i}t\right)^{\lambda_1} e^{-\mathbf{i}t\tilde w^*}\frac{\phi_{f^{(\lambda_{2,1},\lambda_{2,2})}_{Y,X,W_2}(y,x, \cdot)}(t)\phi_{W^*}(t)}{\phi_{W_2}(t)} dt\right)d\tilde w^*.
\end{align*}
Denote $k_{h_1}(v)\equiv\frac{1}{h_1}K\left(\frac{v}{h_1}\right)$, then for all $x\in\mathbb{R}$, $g_{\lambda_1,\lambda_{2,1}, \lambda_{2,2}}(y,x,w^*,h_1)$ is the convolution between $k_{h_1}(\cdot)$ and $g_{\lambda_1,\lambda_{2,1}, \lambda_{2,2}}(y,x,\cdot)$. By the convolution theorem, it's equal to the inverse Fourier transform of the product of the Fourier transformation of $k_{h_1}(\cdot)$ and the Fourier transformation of $g_{\lambda_1,\lambda_{2,1}, \lambda_{2,2}}(y,x,\cdot)$. The Fourier transformation of $k_{h_1}(\cdot)$ is $\phi_{K}(ht)$, and the Fourier transformation of $g_{\lambda_1,\lambda_{2,1}, \lambda_{2,2}}(y,x,\cdot)$ is $(-\mathbf{i}t)^{\lambda_1}\frac{\phi_{f^{(\lambda_{2,1},\lambda_{2,2})}_{Y,X,W_2}(y,x,\cdot)}(t)\phi_{W^*}(t)}{\phi_{W_2}(t)}$.
\end{proof}

\begin{proof} [\proofname \ of Lemma \ref{decomposition}]
Write $\hat{\theta}_{}(\zeta) \equiv \hat{E}\left[W_1 e^{\mathrm{i} \zeta W_{2}}\right]$, $\delta \hat{\theta}_{}(t) \equiv \hat{\theta}_{}(t)-\theta_{}(t)$,  $\delta\hat\phi_{Z}(t)=\hat\phi_{Z}(t)-\phi_{Z}(t)$ for random variable $Z$, and $\delta\hat\phi_{f^{(\lambda_{2,1}, \lambda_{2,2})}_{Y,X,W_2}(y,x,\cdot)}(t)=\hat\phi_{f^{(\lambda_{2,1}, \lambda_{2,2})}_{Y,X,W_2}(y,x,\cdot)}(t)-\phi_{f^{(\lambda_{2,1}, \lambda_{2,2})}_{Y,X,W_2}(y,x,\cdot)}(t)$. Denote $q_{\lambda_2}(t)\equiv \frac{\phi_{f^{(\lambda_{2,1}, \lambda_{2,2})}_{Y,X,W_2}(y,x,\cdot)}( t)}{\phi_{W_2}(t)}$, $\hat q_{\lambda_2}(t)\equiv \frac{\hat\phi_{f^{(\lambda_{2,1}, \lambda_{2,2})}_{Y,X,W_2}(y,x,\cdot)}( t)}{\hat\phi_{W_2}(t)}$ and $\delta\hat q_{\lambda_2}(t)\equiv \hat q_{\lambda_2}(t)-q_{\lambda_2}(t)$, then $\delta\hat q_{\lambda_2}(t)$ can be written as
\begin{footnotesize}
\begin{align}
    \delta \hat{q}_{\lambda_2}(t)&= \left(\frac{\delta\hat\phi_{f^{(\lambda_{2,1}, \lambda_{2,2})}_{Y,X,W_2}(y,x,\cdot)}( t)}{\phi_{W_2}(t)}-\frac{\phi_{f^{(\lambda_{2,1}, \lambda_{2,2})}_{Y,X,W_2}(y,x,\cdot)}(t)\delta \hat{\phi}_{W_2}(t)}{\left(\phi_{W_2}(t)\right)^{2}}\right)\left(1+\frac{\delta \hat{\phi}_{W_2}(t)}{\phi_{W_2}(t)}\right)^{-1} \quad \text { or } \notag \\
    \delta \hat{q}_{\lambda_2}(t)&= \delta_{1} \hat{q}_{\lambda_2}(t)+\delta_{2} \hat{q}_{\lambda_2}(t), \quad \text { with } \notag \\
    \delta_{1} \hat{q}_{\lambda_2}(t) &\equiv  \frac{\delta \hat\phi_{f^{(\lambda_{2,1}, \lambda_{2,2})}_{Y,X,W_2}(y,x,\cdot)}(t)}{\phi_{W_2}(t)}-\frac{\phi_{f^{(\lambda_{2,1}, \lambda_{2,2})}_{Y,X,W_2}(y,x,\cdot)}(t) \delta \hat{\phi}_{W_2}(t)}{\left(\phi_{W_2}(t)\right)^{2}} \notag \\
    \delta_{2} \hat{q}_{\lambda_2}(t) &\equiv  \frac{\phi_{f^{(\lambda_{2,1}, \lambda_{2,2})}_{Y,X,W_2}(y,x,\cdot)}(t)}{\phi_{W_2}(t)}\left(\frac{\delta \hat{\phi}_{W_2}(t)}{\phi_{W_2}(t)}\right)^{2}\left(1+\frac{\delta \hat{\phi}_{W_2}(t)}{\phi_{W_2}(t)}\right)^{-1} \notag \\
    &\quad-\frac{\delta \hat\phi_{f^{(\lambda_{2,1}, \lambda_{2,2})}_{Y,X,W_2}(y,x,\cdot)}(t)}{\phi_{W_2}(t)} \frac{\delta \hat{\phi}_{W_2}(t)}{\phi_{W_2}(t)}\left(1+\frac{\delta \hat{\phi}_{W_2}(t)}{\phi_{W_2}(t)}\right)^{-1}.\label{hatphiXW2/hatphiW2}
\end{align}
\end{footnotesize}
Similarly, I denote $q_{W_1}(\xi)\equiv \frac{\theta(\xi)}{\phi_{W_2}(\xi)}$, $\hat q_{W_1}(\xi)\equiv \frac{\hat \theta(\xi)}{\hat \phi_{W_2}(\xi)}$, and $\delta\hat q_{W_1}(\xi)\equiv\hat q_{W_1}(\xi)-q_{W_1}(\xi)$. Then $\delta\hat q_{W_1}(\xi)$ can be written as 
\begin{align}
    \delta\hat q_{W_1}(\xi)&=\left(\frac{\delta \hat{\theta}(\xi)}{\phi_{W_2}(\xi)}-\frac{\theta(\xi) \delta \hat{\phi}_{W_2}(\xi)}{\left(\phi_{W_2}(\xi)\right)^{2}}\right)\left(1+\frac{\delta\hat\phi_{W_2}(\xi)}{\phi_{W_2}(\xi)}\right)^{-1}\quad \text{or}\notag \\
    \delta\hat q_{W_1}(\xi)&=\delta_1\hat q_{W_1}(\xi)+\delta_2\hat q_{W_1}(\xi),\quad \text{with}\notag\\
    \delta_1\hat q_{W_1}(\xi)&\equiv  \frac{\delta \hat{\theta}(\xi)}{\phi_{W_2}(\xi)}-\frac{\theta(\xi) \delta \hat{\phi}_{W_2}(\xi)}{\left(\phi_{W_2}(\xi)\right)^{2}} \notag\\
    \delta_{2} \hat{q}_{W_1}(\xi) &\equiv \frac{\theta(\xi)}{\phi_{W_2}(\xi)}\left(\frac{\delta \hat{\phi}_{W_2}(\xi)}{\phi_{W_2}(\xi)}\right)^{2}\left(1+\frac{\delta \hat{\phi}_{W_2}(\xi)}{\phi_{W_2}(\xi)}\right)^{-1} \notag\\
    &-\frac{\delta \hat\theta(\xi)}{\phi_{W_2}(\xi)} \frac{\delta \hat{\phi}_{W_2}(\xi)}{\phi_{W_2}(\xi)}\left(1+\frac{\delta \hat{\phi}_{W_2}(\xi)}{\phi_{W_2}(\xi)}\right)^{-1}. \label{delta hat q W1}
\end{align}
For $Q(t)=\int_0^t \frac{\mathbf{i}\theta(\xi)}{\phi_{W_2}(\xi)}d\xi$, $\delta \hat Q(t)=\int_0^t \frac{\mathbf{i}\hat\theta(\xi)}{\hat\phi_{W_2}(\xi)}d\xi-Q(t)$ and some random function $\delta\bar Q(t)$ such that $|\delta\bar Q(t)|\leq |\delta\hat Q(t)|$ for all $t$,
\begin{align}\label{expQ}
    \exp\left(Q(t)+\delta \hat Q(t)\right) = \exp(Q(t))\left(1+\delta Q(t)+\frac{1}{2}[\exp(\delta \bar Q(t))]\left(\delta \hat Q(t)\right)^2\right)
\end{align}
Then I can state a useful representation for 
\begin{align}
    &\quad\frac{\hat\phi_{f^{(\lambda_{2,1}, \lambda_{2,2})}_{Y,X,W_2}(y,x,\cdot)}(t)\hat\phi_{W^*}(t)}{\hat\phi_{W_2}(t)}\notag\\
    &=\left(q_{\lambda_2}(t)+\delta\hat q_{\lambda_2}(t)\right)\exp(Q(t))\left(1+\delta\hat Q(t)+\frac{1}{2}\exp\left(\delta\bar Q(t)\right)\left(\delta \hat Q(t)\right)^2\right)\notag\\
    &=q_{\lambda_2}(t)\exp(Q(t))+q_{\lambda_2}(t)\exp(Q(t))\left(\delta\hat Q(t)+\frac{1}{2}\exp\left(\delta\bar Q(t)\right)\left(\delta \hat Q(t)\right)^2\right)\notag\\
    &\quad+\delta\hat q_{\lambda_2}(t)\exp(Q(t))\left(1+\delta\hat Q(t)+\frac{1}{2}\exp\left(\delta\bar Q(t)\right)\left(\delta \hat Q(t)\right)^2\right)\notag\\
    &=q_{\lambda_2}(t)\exp(Q(t))+q_{\lambda_2}(t)\exp(Q(t))\left[\int_0^t \mathbf{i} \delta_1\hat q_{W_1}(\xi)d\xi\right]+\delta_1\hat q_{\lambda_2}(t)\exp(Q(t))\notag\\
    &\quad+q_{\lambda_2}(t)\exp(Q(t))\left[\int_0^t \mathbf{i} \delta_2\hat q_{W_1}(\xi)d\xi\right]+q_{\lambda_2}(t)\exp(Q(t))\left[\frac{1}{2}\exp\left(\delta\bar Q(t)\right)\left(\delta \hat Q(t)\right)^2\right]\notag\\
    &\quad+\delta_2\hat q_{\lambda_2}(t)\exp(Q(t))+\delta\hat q_{\lambda_2}(t)\exp(Q(t))\left(\delta\hat Q(t)+\frac{1}{2}\exp\left(\delta\bar Q(t)\right)\left(\delta \hat Q(t)\right)^2\right).\label{phif phiWover phiW2}
\end{align}
Plug (\ref{phif phiWover phiW2}) into $\hat g_{\lambda_1,\lambda_{2,1},\lambda_{2,2}}(y,x,w^*,h)$, I get
\begin{align*}
    &\quad\hat g_{\lambda_1,\lambda_{2,1},\lambda_{2,2}}(y,x,w^*,h)-g_{\lambda_1,\lambda_{2,1},\lambda_{2,2}}(y,x,w^*,h)\\
    &=Dg_{\lambda_1,\lambda_{2,1},\lambda_{2,2}}(y,x,w^*,h)+ Rg_{\lambda_1,\lambda_{2,1},\lambda_{2,2}}(y,x,w^*,h)
\end{align*}
where $Dg_{\lambda_1,\lambda_{2,1},\lambda_{2,2}}(y,x,w^*,h)$ is the term linear in $\delta\hat \theta(t)$, $\delta \hat \phi_{W_2}$, and  $\delta\hat\phi_{f^{(\lambda_{2,1}, \lambda_{2,2})}_{Y,X,W_2}(y,x,\cdot)}(t)$, and 
$Rg_{\lambda_1,\lambda_{2,1},\lambda_{2,2}}(y,x,w^*,h)$ is the higher-order remainder term. More specifically,
\begin{footnotesize}
\begin{align*}
    &Dg_{\lambda_1,\lambda_{2,1},\lambda_{2,2}}(y,x,w^*,h)\equiv \frac{1}{2\pi}\int\left(-\mathbf{i}t\right)^{\lambda_1} e^{-\mathbf{i}tw^*} \phi_K(h_{1n}t)\\
    &\quad\quad\quad\quad\quad\quad\quad\quad\quad\quad\quad\quad\quad\quad\left(q_{\lambda_2}(t)\exp(Q(t))\left[\int_0^t \mathbf{i} \delta_1\hat q_{W_1}(\xi)d\xi\right]+\delta_1\hat q_{\lambda_2}(t)\exp(Q(t))\right) dt\\
    &Rg_{\lambda_1,\lambda_{2,1},\lambda_{2,2}}(y,x,w^*,h)\equiv \frac{1}{2\pi}\int\left(-\mathbf{i}t\right)^{\lambda_1} e^{-\mathbf{i}tw^*} \phi_K(h_{1n}t)\\
    &\quad\quad\quad\quad\quad\quad\quad\quad\quad\quad\quad\quad\quad\quad\Bigg\{q_{\lambda_2}(t)\exp(Q(t))\left[\int_0^t \mathbf{i} \delta_2\hat q_{W_1}(\xi)d\xi\right]\\
    &\quad\quad\quad\quad\quad\quad\quad\quad\quad\quad\quad\quad\quad\quad+q_{\lambda_2}(t)\exp(Q(t))\left[\frac{1}{2}\exp\left(\delta\bar Q(t)\right)\left(\delta \hat Q(t)\right)^2\right]\notag\\
    &\quad\quad\quad\quad\quad\quad\quad\quad\quad\quad\quad\quad\quad\quad+\delta_2\hat q_{\lambda_2}(t)\exp(Q(t))\\
    &\quad\quad\quad\quad\quad\quad\quad\quad\quad\quad\quad\quad\quad\quad+\delta\hat q_{\lambda_2}(t)\exp(Q(t))\left(\delta\hat Q(t)+\frac{1}{2}\exp\left(\delta\bar Q(t)\right)\left(\delta \hat Q(t)\right)^2\right)\Bigg\}dt.
\end{align*}
\end{footnotesize}
\indent Define $\bar g_{\lambda_1,\lambda_{2,1},\lambda_{2,2}}(y,x,w^*,h)\equiv g_{\lambda_1,\lambda_{2,1},\lambda_{2,2}}(y,x,w^*,h)+Dg_{\lambda_1,\lambda_{2,1},\lambda_{2,2}}(y,x,w^*,h)$ 
Then I have the following expression:
\begin{footnotesize}
\begin{align*}
    &\quad\hat g_{\lambda_1,\lambda_{2,1},\lambda_{2,2}}(y,x,w^*,h)-g_{\lambda_1,\lambda_{2,1},\lambda_{2,2}}(y,x,w^*)\\
    &=B_{\lambda_1,\lambda_{2,1},\lambda_{2,2}}(y,x,w^*,h)+L_{\lambda_1,\lambda_{2,1},\lambda_{2,2}}(y,x,w^*,h)+R_{\lambda_1,\lambda_{2,1},\lambda_{2,2}}(y,x,w^*,h)
\end{align*}
\end{footnotesize}
where 
\begin{footnotesize}
\begin{align*}
    &B_{\lambda_1,\lambda_{2,1},\lambda_{2,2}}(y,x,w^*,h)\equiv E\left[\bar g_{\lambda_1,\lambda_{2,1},\lambda_{2,2}}(y,x,w^*,h)\right]-g_{\lambda_1,\lambda_{2,1},\lambda_{2,2}}(y,x,w^*)\\
    &L_{\lambda_1,\lambda_{2,1},\lambda_{2,2}}(y,x,w^*,h)\equiv \bar g_{\lambda_1,\lambda_{2,1},\lambda_{2,2}}(y,x,w^*,h)-E\left[\bar g_{\lambda_1,\lambda_{2,1},\lambda_{2,2}}(y,x,w^*,h)\right]\\
    &R_{\lambda_1,\lambda_{2,1},\lambda_{2,2}}(y,x,w^*,h)\equiv\hat g_{\lambda_1,\lambda_{2,1},\lambda_{2,2}}(y,x,w^*,h)-\bar g_{\lambda_1,\lambda_{2,1},\lambda_{2,2}}(y,x,w^*,h).
\end{align*}
\end{footnotesize}
To see the explicit form of $B_{\lambda_1,\lambda_{2,1},\lambda_{2,2}}(y,x,w^*,h)$, note that
\begin{footnotesize}
\begin{align*}
    &\quad B_{\lambda_1,\lambda_{2,1},\lambda_{2,2}}(y,x,w^*,h)\equiv E\left[\bar g_{\lambda_1,\lambda_{2,1},\lambda_{2,2}}(y,x,w^*,h)\right]-g_{\lambda_1,\lambda_{2,1},\lambda_{2,2}}(y,x,w^*)\\
    &=g_{\lambda_1,\lambda_{2,1},\lambda_{2,2}}(y,x,w^*,h)-g_{\lambda_1,\lambda_{2,1},\lambda_{2,2}}(y,x,w^*)+E\left[Dg_{\lambda_1,\lambda_{2,1},\lambda_{2,2}}(y,x,w^*,h)\right]\\
    &=g_{\lambda_1,\lambda_{2,1},\lambda_{2,2}}(y,x,w^*,h)-g_{\lambda_1,\lambda_{2,1},\lambda_{2,2}}(y,x,w^*)\\
    &\quad+\frac{1}{2\pi}\int\left(-\mathbf{i}t\right)^{\lambda_1} e^{-\mathbf{i}tw^*} \phi_K(h_{1n}t)\frac{E\left[\delta\hat\phi_{f^{(\lambda_{2,1}, \lambda_{2,2})}_{Y,X,W_2}(y,x,\cdot)}(t)\right]\phi_{W^*}(t)}{\phi_{W_2}(t)} dt
\end{align*}
\end{footnotesize}
To see the explicit form of $L_{\lambda_1,\lambda_{2,1},\lambda_{2,2}}(y,x,w^*,h)$, note that
\begin{footnotesize}
\begin{align*}
    &\quad L_{\lambda_1,\lambda_{2,1},\lambda_{2,2}}(y,x,w^*,h)\equiv \bar g_{\lambda_1,\lambda_{2,1},\lambda_{2,2}}(y,x,w^*,h)-E\left[\bar g_{\lambda_1,\lambda_{2,1},\lambda_{2,2}}(y,x,w^*,h)\right]\\
    &=Dg_{\lambda_1,\lambda_{2,1},\lambda_{2,2}}(y,x,w^*,h)-E\left[Dg_{\lambda_1,\lambda_{2,1},\lambda_{2,2}}(y,x,w^*,h)\right]\\
    &=\frac{1}{2\pi}\int\left(-\mathbf{i}t\right)^{\lambda_1} e^{-\mathbf{i}tw^*} \phi_K(h_{1n}t)\\
    &\quad\left(\phi_{f^{(\lambda_{2,1}, \lambda_{2,2})}_{Y,X,W^*}(y,x,\cdot)}(t)\left[\int_0^t \left(\frac{\mathbf{i}\delta \hat{\theta}(\xi)}{\phi_{W_2}(\xi)}-\frac{\mathbf{i}\theta(\xi) \delta \hat{\phi}_{W_2}(\xi)}{\left(\phi_{W_2}(\xi)\right)^{2}}\right)d\xi\right]-\frac{\delta \hat{\phi}_{W_2}(t)}{\phi_{W_2}(t)}\phi_{f^{(\lambda_{2,1}, \lambda_{2,2})}_{Y,X,W^*}(y,x,\cdot)}(t) \right) dt\\
    &\quad+\frac{1}{2\pi}\int\left(-\mathbf{i}t\right)^{\lambda_1} e^{-\mathbf{i}tw^*} \phi_K(h_{1n}t)\frac{\left(\hat\phi_{f^{(\lambda_{2,1}, \lambda_{2,2})}_{Y,X,W_2}(y,x,\cdot)}(t)-E\left[\hat\phi_{f^{(\lambda_{2,1}, \lambda_{2,2})}_{Y,X,W_2}(y,x,\cdot)}(t)\right]\right)\phi_{W^*}(t)}{\phi_{W_2}(t)} dt\\
    &=\frac{1}{2 \pi} \int\left(-\mathbf{i}t\right)^{\lambda_1}  e^{-\mathbf it w^*} \phi_{K}(h_1 t)\phi_{f^{(\lambda_{2,1},\lambda_{2,2})}_{Y,X,W^*}(y,x,\cdot)}(t)\left[\int_{0}^{t}\left(\frac{\mathbf{i} \delta \hat{\theta}_{}(\xi)}{\phi_{W_2}(\xi)}-\frac{\mathbf{i} \theta_{}(\xi) \delta \hat{\phi}_{W_2}(\xi)}{\left(\phi_{W_2}(\xi)\right)^{2}}\right) d \xi\right]  d t\\
    &+\quad\frac{1}{2 \pi}\int \left(-\mathbf{i}t\right)^{\lambda_1} e^{-\mathbf itw^*} \phi_{K}(h_1t)\\
    &\times \left(\frac{\hat\phi_{f^{(\lambda_{2,1}, \lambda_{2,2})}_{Y,X,W_2}(y,x,\cdot)}(t)-E\left[\hat\phi_{f^{(\lambda_{2,1}, \lambda_{2,2})}_{Y,X,W_2}(y,x,\cdot)}(t)\right]}{\phi_{W_2}(t)}\phi_{W^*}(t)-\frac{\delta \hat\phi_{W_2}(t)}{\phi_{W_2}(t)}\phi_{f^{(\lambda_{2,1},\lambda_{2,2})}_{Y,X,W^*}(y,x,\cdot)}(t)\right) dt
\end{align*}
\end{footnotesize}
Then using the identity 
\begin{align*}
    \int_{-\infty}^{\infty} \int_{0}^{\xi} f(\xi, \zeta) d \zeta d \xi &= \int_{0}^{\infty} \int_{\zeta}^{\infty} f(\xi, \zeta) d \xi d \zeta+\int_{-\infty}^{0} \int_{\zeta}^{-\infty} f(\xi, \zeta) d \xi d \zeta \\
    &:=\iint_{\zeta}^{\pm \infty} f(\xi, \zeta) d \xi d \zeta
\end{align*}
for any absolutely integrable function $f$, I obtain
\begin{footnotesize}
\begin{align}
    &\quad L_{\lambda_1,\lambda_{2,1}, \lambda_{2,2}}(y, x, w^*, h) \\
    &= \bar g_{\lambda_1,\lambda_{2,1}, \lambda_{2,2}}(y, x, w^*,h)-g_{\lambda_1,\lambda_{2,1}, \lambda_{2,2}}(y, x, w^*,h)\notag\\
    &=\frac{1}{2 \pi} \int\int_{\xi}^{\pm\infty} \left(-\mathbf{i}t\right)^{\lambda_1} e^{-\mathbf itw^*}\phi_{K}(h_1t)\phi_{f^{(\lambda_{2,1},\lambda_{2,2})}_{Y,X,W^*}(y,x,\cdot)}(t)dt \left(\frac{\mathbf{i} \delta \hat{\theta}_{}(\xi)}{\phi_{W_2}(\xi)}-\frac{\mathbf{i} \theta_{}(\xi) \delta \hat{\phi}_{W_2}(\xi)}{\left(\phi_{W_2}(\xi)\right)^{2}}\right) d \xi \notag\\
    &+\frac{1}{2 \pi}\int  \left(-\mathbf{i}t\right)^{\lambda_1}e^{-\mathbf itw^*} \phi_{K}(h_1t)\times\\
    &\quad\quad\left(\frac{\hat\phi_{f^{(\lambda_{2,1}, \lambda_{2,2})}_{Y,X,W_2}(y,x,\cdot)}(t)-E\left[\hat\phi_{f^{(\lambda_{2,1}, \lambda_{2,2})}_{Y,X,W_2}(y,x,\cdot)}(t)\right]}{\phi_{W_2}(t)}\phi_{W^*}(t)-\frac{\delta \hat\phi_{W_2}(t)}{\phi_{W_2}(t)}\phi_{f^{(\lambda_{2,1},\lambda_{2,2})}_{Y,X,W^*}(y,x,\cdot)}(t)\right) dt\notag\\
    &=\int \Psi_{1, \lambda_1,\lambda_{2,1},\lambda_{2,2}}(\xi, y, x, w^*, h_1)\left(\hat E[W_1e^{\mathbf{i}\xi W_2}]-E[W_1e^{\mathbf{i}\xi W_2}]\right)d\xi\notag\\
    &+\int \Psi_{2,\lambda_1,\lambda_{2,1},\lambda_{2,2}}(\xi, y, x, w^*, h_1)\left(\hat E[e^{\mathbf{i}\xi W_2}]-E[e^{\mathbf{i}\xi W_2}]\right)d\xi\notag\\
    &+\int \Psi_{3,\lambda_1,\lambda_{2,1},\lambda_{2,2}}(\xi, y, x, w^*, h_1)\left(\hat E\left[e^{\mathbf{i}\xi W_2}\frac{1}{h_{2,1}^{1+\lambda_{2,1}}h_{2,2}^{1+\lambda_{2,2}}}G_Y\left( \frac{y-Y}{h_{2,1}}\right)G_X^{(\lambda_{2,2})}\left(\frac{x-X}{h_{2,2}}\right)\right]\right.\notag\\
    &\quad\quad\quad\quad\quad\quad\quad\quad\quad\quad\quad\quad\quad\quad \left.-E\left[e^{\mathbf{i}\xi W_2}\frac{1}{h_{2,1}^{1+\lambda_{2,1}}h_{2,2}^{1+\lambda_{2,2}}}G_Y\left( \frac{y-Y}{h_{2,1}}\right)G_X^{(\lambda_{2,2})}\left(\frac{x-X}{h_{2,2}}\right)\right]\right) d\xi\notag\\
    &=\hat E\Bigg[\int \Psi_{1, \lambda_1,\lambda_{2,1},\lambda_{2,2}}(\xi, y, x, w^*, h_1)\left(W_1e^{\mathbf{i}\xi W_2}-E[W_1e^{\mathbf{i}\xi W_2}]\right)d\xi\notag\\
    &\quad+\int \Psi_{2,\lambda_1,\lambda_{2,1},\lambda_{2,2}}(\xi, y, x, w^*, h_1)\left(e^{\mathbf{i}\xi W_2}-E[e^{\mathbf{i}\xi W_2}]\right)d\xi\notag\\
    &\quad+\int \Psi_{3,\lambda_1,\lambda_{2,1},\lambda_{2,2}}(\xi, y, x, w^*, h_1)\times\left(e^{\mathbf{i}\xi W_2}\frac{1}{h_{2,1}^{1+\lambda_{2,1}}h_{2,2}^{1+\lambda_{2,2}}}G_Y\left( \frac{y-Y}{h_{2,1}}\right)G_X^{(\lambda_{2,2})}\left(\frac{x-X}{h_{2,2}}\right)\right.\notag\\    &\quad\quad\quad\quad\quad\quad\quad\quad\quad\quad\quad\quad\quad\quad\quad\quad \left.-E\left[e^{\mathbf{i}\xi W_2}\frac{1}{h_{2,1}^{1+\lambda_{2,1}}h_{2,2}^{1+\lambda_{2,2}}}G_Y\left( \frac{y-Y}{h_{2,1}}\right)G_X^{(\lambda_{2,2})}\left(\frac{x-X}{h_{2,2}}\right)\right]\right) d\xi \Bigg]\notag\\
    &=\hat E[l_{\lambda_1,\lambda_{2,1}, \lambda_{2,2}}(y, x, w^*,h;Y, X, W_1, W_2)],\label{L}
\end{align}
\end{footnotesize}

where 
\begin{footnotesize}
\begin{align*}
    \Psi_{1, \lambda_1,\lambda_{2,1},\lambda_{2,2}}(\xi, y, x, w^*, h_1) &= \frac{1}{2 \pi} \frac{\mathbf{i}}{\phi_{W_2}(\xi)}\int_{\xi}^{\pm\infty} \left(-\mathbf{i}t\right)^{\lambda_1} e^{-\mathbf itw^*} \phi_{K}(h_1t)\phi_{f^{(\lambda_{2,1},\lambda_{2,2})}_{Y,X,W^*}(y,x,\cdot)}(t)dt \\
    \Psi_{2,\lambda_1,\lambda_{2,1},\lambda_{2,2}}(\xi, y, x, w^*, h_1) &= -\frac{1}{2 \pi} \frac{\mathbf{i}\theta(\xi)}{(\phi_{W_2}(\xi))^2}\int_{\xi}^{\pm\infty}  \left(-\mathbf{i}t\right)^{\lambda_1}e^{-\mathbf itw^*} \phi_{K}(h_1 t)\phi_{f^{(\lambda_{2,1},\lambda_{2,2})}_{Y,X,W^*}(y,x,\cdot)}(t)dt\\
    &\quad -\frac{1}{2 \pi} \left(-\mathbf{i}t\right)^{\lambda_1}e^{-\mathbf i\xi w^*} \phi_{K}(h _1\xi)\frac{\phi_{f^{(\lambda_{2,1},\lambda_{2,2})}_{Y,X,W^*}(y,x,\cdot)}(\xi)}{\phi_{W_2}(\xi)}\\
    \Psi_{3,\lambda_1,\lambda_{2,1},\lambda_{2,2}}(\xi, y, x, w^*, h_1) &= \frac{1}{2 \pi}   \left(-\mathbf{i}t\right)^{\lambda_1}e^{-\mathbf i\xi w^*} \phi_{K}(h_1\xi)\frac{\phi_{W^*}(\xi)}{\phi_{W_2}(\xi)}.
\end{align*}
\end{footnotesize}

\end{proof}

\begin{definition}
Write $f(t)\preceq g(t)$ for $f,g: \mathbb{R}\rightarrow\mathbb{R}$ when there is a constant $C>0$, independent of $t$, such that $f(t)\leq Cg(t)$ for all $t\in\mathbb{R}$ (and similarly for $\succeq$). Write $a_n\preceq b_n$ for two sequences $a_n, b_n$ if there exists a constant $C$ independent of $n$ such that $a_n\leq C b_n$ for all $n\in\mathbb{N}$
\end{definition}

\begin{proof} [\proofname\ of Theorem \ref{biasrate}]
By Parseval's identity, I have
\begin{align*}
    &\quad|g_{\lambda_1,\lambda_{2,1}, \lambda_{2,2}}(y, x, w^*,h)-g_{\lambda_1,\lambda_{2,1}, \lambda_{2,2}}(y, x, w^*)|\\
    &=\left|\frac{1}{2\pi}\int \left(-\mathbf{i}t\right)^{\lambda_1}
    e^{-\mathbf itw^*}\left(\phi_{K}(h_1t)-1\right)\frac{\phi_{f^{(\lambda_{2,1},\lambda_{2,2})}_{Y,X,W_2}(y,x,\cdot)}(t)}{\phi_{W_2}(t)}\phi_{W^*}(t)dt\right|\\
    &\leq\frac{1}{2\pi}\int |t|^{\lambda_1}\left|\phi_{K}(h_1t)-1\right|\left|\phi_{f^{(\lambda_{2,1},\lambda_{2,2})}
    _{Y,X,W^*}(y,x,\cdot)}(t)\right| d t\\
    &=\frac{1}{2\pi}\int_{|h_1t| >\bar\zeta_K }|t|^{\lambda_1} \left|\phi_{K}(h_1t)-1\right|\left|\phi_{f^{(\lambda_{2,1},\lambda_{2,2})}_{Y,X,W^*}(y,x,\cdot)}(t)\right| d t\\
    &\preceq \int_{|h_1t| >\bar\zeta_K } |t|^{\lambda_1}\left|\phi_{f^{(\lambda_{2,1},\lambda_{2,2})}_{Y,X,W^*}(y,x,\cdot)}(t)\right| d t
\end{align*}
where I have used Assumption \ref{phik} to ensure $\phi_{K}(t)=1$ for $|t| \leq\bar\zeta _K$ and $\sup_{\zeta}|\phi_{K}(\zeta)|<\infty$. Then by Assumption \ref{smoothness1} and Lemma 7 in \cite{schennach2004nonparametric}, 
\begin{align*}
    &\quad|g_{\lambda_1,\lambda_{2,1}, \lambda_{2,2}}(y, x, w^*,h)-g_{\lambda_1,\lambda_{2,1}, \lambda_{2,2}}(y, x, w^*)|\\
    &\preceq \int_{|h_1t| >\bar\zeta_K} |t|^{\lambda_1}(1+|t|)^{\gamma_{\phi}} \exp \left(\alpha_{\phi}|t|^{\beta_{\phi}}\right) d t\\
    &=O\left(\left(\frac{\bar\zeta_K}{h_1}\right)^{1+\gamma_{\phi}+\lambda_1}\exp\left(\alpha_{\phi}\left(\frac{\bar\zeta_K}{h_1}\right)^{\beta_{\phi}}\right)\right)=O\left(\left(h_1^{-1}\right)^{1+\gamma_{\phi}+\lambda_1}\exp\left(\alpha_{\phi}\bar\zeta_K^{\beta_{\phi}}\left(h_1^{-1}\right)^{\beta_{\phi}}\right)\right).
\end{align*}
Next I show that the second term in $B_{\lambda_1,\lambda_{2,1},\lambda_{2,2}}(y, x, w^*, h)$ is of a smaller order. To do this, I first state a useful representation of 
$E\left[\hat\phi_{f^{(\lambda_{2,1}, \lambda_{2,2})}_{Y,X,W^*}(y,x,\cdot)}(\xi)\right]$
\begin{footnotesize}
\begin{align*}
    &\quad E\left[\hat\phi_{f^{(\lambda_{2,1}, \lambda_{2,2})}_{Y,X,W_2}(y,x,\cdot)}(\xi)\right]\\
    &=E\left[e^{\mathbf{i}\xi W_2 }\frac{1}{h_{2,1}^{1+\lambda_{2,1}}h_{2,2}^{1+\lambda_{2,2}}}G^{(\lambda_{2,1})}_Y\left(\frac{y-Y}{h_{2n,1}}\right)G^{(\lambda_{2,2})}_X\left(\frac{x-X}{h_{2n,2}}\right)\right]\\
    &=\int\int e^{\mathbf{i}\xi\tilde w}\frac{1}{h_{2,1}^{1+\lambda_{2,1}}h_{2,2}^{1+\lambda_{2,2}}}G^{(\lambda_{2,1})}_Y\left(\frac{y-\tilde y}{h_{2n,1}}\right)G^{(\lambda_{2,2})}_X \left(\frac{x-\tilde x}{h_{2,2}}\right) f_{Y,X,W_2}(\tilde y, \tilde x, \tilde w) d\tilde y d\tilde xd\tilde w\\
    &=\int \phi_{f^{(0,0)}_{Y,X,W_2}(\tilde y,\tilde x,\cdot)}(\xi)\frac{1}{h_{2,1}^{1+\lambda_{2,1}}h_{2,2}^{1+\lambda_{2,2}}}G^{(\lambda_{2,1})}_Y\left(\frac{y-\tilde y}{h_{2n,1}}\right)G^{(\lambda_{2,2})}_X \left(\frac{x-\tilde x}{h_{2,2}}\right)d\tilde y d\tilde x
\end{align*}
\end{footnotesize}
If $\lambda_{2,1}=1, \lambda_{2,2}=0$,
\begin{footnotesize}
\begin{align*}
    &\quad E\left[\hat\phi_{f^{(1,0)}_{Y,X,W_2}(y,x,\cdot)}(\xi)\right]\\
    &=-\int \left(\int \phi_{f^{(0,0)}_{Y,X,W_2}(\tilde y, \tilde x,\cdot)}(\xi)\frac{1}{h_{2,1}} dG_Y\left(\frac{y-\tilde y}{h_{2,1}}\right) \right)\frac{1}{h_{2,2}}G_X \left(\frac{x-\tilde x}{h_{2,2}}\right) d\tilde x\\
    &=\int \left(\left.-\phi_{f^{(0,0)}_{Y,X,W_2}(\tilde y,\tilde x,\cdot)}(\xi)\frac{1}{h_
    {2,1}}G_Y\left(\frac{\tilde y-y}{h_{2,1}}\right)\right|_{-\infty}^\infty\right.\\
    &\quad\quad\quad\quad\left.+\int \frac{1}{h_{2,1}}G_Y\left(\frac{y-\tilde y}{h_{2,1}}\right)  \phi_{f^{(1,0)}_{Y,X,W_2}(\tilde y, \tilde x,\cdot)}(\xi)d \tilde y\right)\frac{1}{h_{2,2}}G_X \left(\frac{x-\tilde x}{h_{2,2}}\right) d\tilde x\\
    &=\int \int \phi_{f^{(1,0)}_{Y,X,W_2}(\tilde y, \tilde x,\cdot)}(\xi)\frac{1}{h_{2,1}}G_Y\left(\frac{y-\tilde y}{h_{2,1}}\right)   \frac{1}{h_{2,2}}G_X \left(\frac{x-\tilde x}{h_{2,2}}\right) d\tilde y d\tilde x
\end{align*}
\end{footnotesize}
For $\lambda_{2,1}=0, \lambda_{2,2}=1$, I have similar results, so 
\begin{footnotesize}
\begin{align*}
     &\quad E\left[\hat\phi_{f^{(\lambda_{2,1}, \lambda_{2,2})}_{Y,X,W_2}(y,x,\cdot)}(\xi)\right]=\int \int \phi_{f^{(\lambda_{2,1}, \lambda_{2,2})}_{Y,X,W_2}(\tilde y, \tilde x,\cdot)}(\xi)\frac{1}{h_{2,1}}G_Y\left(\frac{y-\tilde y}{h_{2,1}}\right)   \frac{1}{h_{2,2}}G_X \left(\frac{x-\tilde x}{h_{2,2}}\right) d\tilde y d\tilde x.
\end{align*}
\end{footnotesize}
Then note that
\begin{scriptsize}
\begin{align}
    &\quad \left|E\left[\delta\hat\phi_{f^{(\lambda_{2,1}, \lambda_{2,2})}_{Y,X,W_2}(y,x,\cdot)}(\xi)\right]\frac{\phi_{W^*}(\xi)}{\phi_{W_2}(\xi)}\right|\notag\\
    &=\left|\int\int  \phi_{f^{(\lambda_{2,1}, \lambda_{2,2})}_{Y,X,W^*}(\tilde y, \tilde x,\cdot)}(\xi)\frac{1}{h_{2,1}}G_Y\left(\frac{y-\tilde y}{h_{2,1}}\right)   \frac{1}{h_{2,2}}G_X \left(\frac{x-\tilde x}{h_{2,2}}\right) d\tilde y d\tilde x-\phi_{f^{(\lambda_{2,1}, \lambda_{2,2})}_{Y,X,W^*}(y,x,\cdot)}(\xi)\right|\notag\\
    &=\left|\int\int \phi_{f^{(\lambda_{2,1}, \lambda_{2,2})}_{Y,X,W^*}(\tilde y, \tilde x,\cdot)}(\xi)\left(\frac{1}{(2\pi)^2}\int\int e^{-\mathbf{i}t_1(y-\tilde y)-\mathbf{i}t_2(x-\tilde x)} \phi_{G_Y}(t_1h_{2,1})\phi_{G_X}(t_2h_{2,2})dt_1dt_2\right) d\tilde y d\tilde x\right.\notag\\
    &\quad\left.-\frac{1}{(2\pi)^2}\int\int e^{-\mathbf{i}t_1y-\mathbf{i}t_2x}\left(\int\int e^{\mathbf{i}t_1y+\mathbf{i}t_2x} \phi_{f^{(\lambda_{2,1}, \lambda_{2,2})}_{Y,X,W^*}(y,x,\cdot)}(\xi) dy dx \right) dt_1 dt_2\right|\notag\\
    &=\left|\frac{1}{(2\pi)^2}\int\int e^{-\mathbf{i}t_1y-\mathbf{i}t_2x}\left(\phi_{G_Y}(t_1h_{2,1})\phi_{G_X}(t_2h_{2,2})-1\right)\left(\int\int e^{\mathbf{i}t_1\tilde y+\mathbf{i}t_2\tilde x} \phi_{f^{(\lambda_{2,1}, \lambda_{2,2})}_{Y,X,W^*}(\tilde y,\tilde x,\cdot)}(\xi) d\tilde y d\tilde x \right) dt_1 dt_2\right|\notag\\
    &=\left|\frac{1}{(2\pi)^2}\int_{\frac{\bar\xi_{G_X}}{h_{2,2}}}\int_{\frac{\bar\xi_{G_Y}}{h_{2,1}}} e^{-\mathbf{i}t_1y-\mathbf{i}t_2x}\phi_{f^{(\lambda_{2,1}, \lambda_{2,2})}}(t_1, t_2, \xi) dt_1 dt_2\right|\notag\\
    &\leq\frac{1}{(2\pi)^2}\int_{\frac{\bar\xi_{G_X}}{h_{2,2}}}\int_{\frac{\bar\xi_{G_Y}}{h_{2,1}}} \left|\phi_{f^{(\lambda_{2,1}, \lambda_{2,2})}}(t_1, t_2, \xi)\right| dt_1 dt_2\notag\\
    &\preceq \left|\phi_{W^*}(\xi)\right|\int_{\frac{\bar\xi_{G_X}}{h_{2,2}}}\int_{\frac{\bar\xi_{G_Y}}{h_{2,1}}} (1+|t_1|)^{\gamma_{f_1}}(1+|t_2|)^{\gamma_{f_2}}\exp\left(\alpha_{f_1}|t_1|^{\beta_{f_1}}+\alpha_{f_2}|t_2|^{\beta_{f_2}}\right) dt_1 dt_2 \notag\\
    &\preceq \left|\phi_{W^*}(\xi)\right| O\left(\left(\frac{\bar\xi_{G_Y}}{h_{2,1}}\right)^{1+\gamma_{f_1}}\left(\frac{\bar\xi_{G_X}}{h_{2,2}}\right)^{1+\gamma_{f_2}}\exp\left(\alpha_{f_1}\left(\frac{\bar\xi_{G_Y}}{h_{2,1}}\right)^{\beta_{f_1}}+\alpha_{f_2}\left(\frac{\bar\xi_{G_X}}{h_{2,2}}\right)^{\beta_{f_2}}\right)\right)\notag\\
    &\preceq \left|\phi_{W^*}(\xi)\right|o(1)\label{Edeltahatphi},
\end{align}
\end{scriptsize}
so
\begin{footnotesize}
\begin{align*}
    &\quad \left|\frac{1}{2\pi}\int\left(-\mathbf{i}t\right)^{\lambda_1} e^{-\mathbf{i}tw^*} \phi_K(h_{1n}t)\frac{E\left[\delta\hat\phi_{f^{(\lambda_{2,1}, \lambda_{2,2})}_{Y,X,W_2}(v, \cdot)}(t)\right]\phi_{W^*}(t)}{\phi_{W_2}(t)} dt\right|\\
    &\preceq o(1)\int_0^{h_{1n}^{-1}} |t|^{\lambda_1} \frac{\left|\phi_{W^*}(t)\right|}{\left|\phi_{W_2}(t)\right|}\left|\phi_{W_2}(t)\right| dt\\
    &\preceq o(1)\int_0^{h_{1n}^{-1}} (1+|t|)^{\lambda_1+\gamma_{\phi}-\gamma_2} \exp \left(\alpha_{\phi}|t|^{\beta_{\phi}}\right)\exp \left(-\alpha_{2}|t|^{\beta_{2}}\right)  dt\\
    &\preceq o(1)O\left(\left(h_1^{-1}\right)^{1+\gamma_\phi+\lambda_1-\gamma_2}\exp \left(\alpha_{\phi}\left(h_1^{-1}\right)^{\beta_{\phi}}\right)\exp \left(-\alpha_{2}\left(h_1^{-1}\right)^{\beta_{2}}\right)\right)\\
    &=o\left(\left(h_1^{-1}\right)^{1+\gamma_{\phi}+\lambda_1}\exp\left(\alpha_{\phi}\bar\zeta_K^{\beta_{\phi}}\left(h_1^{-1}\right)^{\beta_{\phi}}\right)\right).
\end{align*}
\end{footnotesize}
This completes the proof.
\end{proof}
\begin{lemma}\label{orderPsi}Suppose the conditions of Lemma \ref{decomposition} hold. For each $\xi$ and $h_1$, let  $\Psi^+_{l,\lambda_1,\lambda_{2,1},\lambda_{2,2}}(\xi,h_1)\equiv\sup_{(y,x,w^*)\in\mathbb{S}_{(Y,X,W^*)}}|\Psi_{l,\lambda_1,\lambda_{2,1},\lambda_{2,2}
}(\xi, y, x, w^*, h_1)|$ for $l=1,2,3$. Define 
\begin{footnotesize}
\begin{align*}
    \Psi_{\lambda_1,\lambda_{2,1}, \lambda_{2,2}}^+(h)\equiv\sum_{l=1}^2\int \Psi_{l,\lambda_1,\lambda_{2,1},\lambda_{2,2}}^+(\xi, h_1)d\xi+(h_{2,1}^{-1})^{1+\lambda_{2,1}}(h_{2,2}^{-1})^{1+\lambda_{2,2}}\int\Psi_{3,\lambda_1,\lambda_{2,1},\lambda_{2,2}}^+(\xi, h_1)d\xi.
\end{align*}        
\end{footnotesize}
If Assumption \ref{smoothness1} and \ref{smoothness2} also hold, then for $h_1>0$
\begin{scriptsize}
\begin{align*}
    &\quad\Psi_{\lambda_1,\lambda_{2,1}, \lambda_{2,2}}^+(h)\\
    &=O\left(\max\left\{(h_1^{-1})^{1+\gamma_*},(h_{2,1}^{-1})^{1+\lambda_{2,1}}(h_{2,2}^{-1})^{1+\lambda_{2,2}}\right\}(h_1^{-1})^{1-\gamma_2+\gamma_{\phi}+\lambda_1}\exp\left(\left(\alpha_{\phi}\mathbf{1}\{\beta_{\phi}=\beta_2\}-\alpha_2\right)\left(h_1^{-1}\right)^{\beta_2}\right)\right).
\end{align*}    
\end{scriptsize}
\end{lemma}
\begin{proof}
\begin{footnotesize}
\begin{align*}
    &\quad\Psi_{1, \lambda_1,\lambda_{2,1},\lambda_{2,2}}^+(\xi, h_1)\\
    &=\sup_{(y,x,w^*)\in\mathbb{S}_{(Y,X,W^*)}}|\Psi_{1, \lambda_1,\lambda_{2,1},\lambda_{2,2}}(\xi, y, x, w^*, h_1)|\\
    &\preceq \sup_{(y,x,w^*)\in\mathbb{S}_{(Y,X,W^*)}}\frac{1}{\left|\phi_{W_2}(\xi)\right|}\int_{\xi}^{\pm\infty}  |t|^{\lambda_1}\left|e^{-\mathbf i t w^*}\right| \left|\phi_{K}(h_1t)\right|\left(\sup_{(y,x)\in\mathbb{S}_{(Y,X)}}\left|\phi_{f^{(\lambda_{2,1},\lambda_{2,2})}
    _{Y,X,W^*}(y,x,\cdot)}(t)\right|\right)dt  \\
    &\preceq \frac{1}{\left|\phi_{W_2}(\xi)\right|}\int_{\xi}^{\pm\infty}   |t|^{\lambda_1}\left|\phi_{K}(h_1t)\right|\left(\sup_{(y,x)\in\mathbb{S}_{(Y,X)}}\left|\phi_{f^{(\lambda_{2,1},\lambda_{2,2})}
    _{Y,X,W^*}(y,x,\cdot)}(t)\right|\right)dt  
\end{align*}    
\end{footnotesize}
so that 
\begin{footnotesize}
\begin{align*}
    \int \Psi_{1, \lambda_1,\lambda_{2,1},\lambda_{2,2}}^+(\xi, h_1)d\xi&\preceq\int\left[\frac{1}{\left|\phi_{W_2}(\xi)\right|}\mathbf{1}\{|\xi|\leq h_1^{-1}\}\int_{\xi}^{h_1^{-1}}|t|^{\lambda_1}\left|\phi_{f^{(\lambda_{2,1},\lambda_{2,2})}_{Y,X,W^*}(y,x,\cdot)}(t)\right|dt\right]d\xi \\
    &\preceq\int\left[\frac{1}{\left|\phi_{W_2}(\xi)\right|}\mathbf{1}\{|\xi|\leq h_1^{-1}\}\int_{\xi}^{h_1^{-1}}|t|^{\lambda_1}(1+|t|)^{\gamma_{\phi}} \exp \left(\alpha_{\phi}|t|^{\beta_{\phi}}\right)dt\right]d\xi\\
    &\preceq(h_1^{-1})^{2-\gamma_2+\gamma_{\phi}+\lambda_1}\exp\left(-\alpha_2\left(h_1^{-1}\right)^{\beta_2}\right)\exp\left(\alpha_{\phi}\left(h_1^{-1}\right)^{\beta_{\phi}}\right)
\end{align*}    
\end{footnotesize}
where in the third $\preceq$, I used the assumption $1-\gamma_2+\gamma_\phi+\lambda_1>0$ when $\beta_2=0$ and invoked Lemma 7 and Lemma 8 in \cite{schennach2004nonparametric}.
\begin{align*}
    &\quad\Psi^+_{2,\lambda_1,\lambda_{2,1},\lambda_{2,2}}(\xi,  h_1)\\
    &=\sup_{(y,x,w^*)\in\mathbb{S}_{(Y,X,W^*)}}\left|\Psi_{2,\lambda_1,\lambda_{2,1},\lambda_{2,2}}(\xi, y, x, w^*, h_1)\right|\\ 
    &\preceq \sup_{(y,x,w^*)\in\mathbb{S}_{(Y,X,W^*)}} \frac{|\theta(\xi)|}{|\phi_{W_2}(\xi)|^2}\int_{\xi}^{\pm\infty} |t|^{\lambda_1} |e^{-\mathbf itw^*}| |\phi_{K}(h_1t)|\left|\phi_{f^{(\lambda_{2,1},\lambda_{2,2})}
    _{Y,X,W^*}(y,x,\cdot)}(t)\right|dt \\
    &\quad +\sup_{(y,x,w^*)\in\mathbb{S}_{(Y,X,W^*)}}|\xi|^{\lambda_1}|e^{-\mathbf i\xi w^*}| |\phi_{K}(h_1\xi)|\frac{\left|\phi_{f^{(\lambda_{2,1},\lambda_{2,2})}
    _{Y,X,W^*}(y,x,\cdot)}(\xi)\right|}{|\phi_{W_2}(\xi)|}\\
    &\preceq \sup_{(y,x)\in\mathbb{S}_{(Y,X)}} \Bigg\{\frac{|\theta(\xi)|}{|\phi_{W_2}(\xi)|^2}\int_{\xi}^{\pm\infty}  |t|^{\lambda_1}|\phi_{K}(h_1t)|\left|\phi_{f^{(\lambda_{2,1},\lambda_{2,2})}
    _{Y,X,W^*}(y,x,\cdot)}(t)\right|dt\\
    &\quad\quad\quad\quad\quad\quad\quad\quad+|\xi|^{\lambda_1}|\phi_{K}(h_1\xi)|\frac{\left|\phi_{f^{(\lambda_{2,1},\lambda_{2,2})}
    _{Y,X,W^*}(y,x,\cdot)}(t)\right|}{|\phi_{W_2}(\xi)|}\Bigg\}\\
    &\preceq\frac{1}{|\phi_{W_2}(\xi)|}\left(\frac{|\theta(\xi)|}{|\phi_{W_2}(\xi)|}\int_{\xi}^{\pm\infty} |t|^{\lambda_1} |\phi_{K}(h_1t)|\left(\sup_{(y,x)\in\mathbb{S}_{(Y,X)}}\left|\phi_{f^{(\lambda_{2,1},\lambda_{2,2})}
    _{Y,X,W^*}(y,x,\cdot)}(t)\right|\right)dt \right.\\
    &\quad + |\xi|^{\lambda_1}|\phi_{K}(h_1\xi)|\left(\sup_{(y,x)\in\mathbb{S}_{(Y,X)}}\left|\phi_{f^{(\lambda_{2,1},\lambda_{2,2})}
    _{Y,X,W^*}(y,x,\cdot)}(\xi)\right| \right) \Bigg)\\
    &\preceq\frac{1}{|\phi_{W_2}(\xi)|}\left(\frac{|\phi'_{W^*}(\xi)|}{|\phi_{W^*}(\xi)|}\int_{\xi}^{\pm\infty} |t|^{\lambda_1} |\phi_{K}(h_1t)|\left(\sup_{(y,x)\in\mathbb{S}_{(Y,X)}}\left|\phi_{f^{(\lambda_{2,1},\lambda_{2,2})}
    _{Y,X,W^*}(y,x,\cdot)}(t)\right|\right)dt \right.\\
    &\quad + |\xi|^{\lambda_1}|\phi_{K}(h_1\xi)|\left(\sup_{(y,x)\in\mathbb{S}_{(Y,X)}}\left|\phi_{f^{(\lambda_{2,1},\lambda_{2,2})}
_{Y,X,W^*}(y,x,\cdot)}(\xi)\right| \right) \Bigg)
\end{align*}
where I used the fact that
\begin{align*}
    \frac{\theta(\xi)}{\phi_{W_2}(\xi)}&=\frac{E\left[W_{1} e^{\mathrm{i} \xi W_{2}}\right]}{E\left[e^{\mathrm{i} \xi W_{2}}\right]}=\frac{E\left[\left(W^*+\Delta W_{1}\right) e^{\mathrm{i} \xi\left(W^*+\Delta W_{2}\right)}\right]}{E\left[e^{\mathrm{i} \xi\left(W^*+\Delta W_{2}\right)}\right]} \\
&=\frac{E\left[W^* e^{\mathrm{i} \xi\left(W^*+\Delta W_{2}\right)}\right]+E\left[E\left[\Delta W_{1} \mid W^*, \Delta W_{2}\right] e^{\mathrm{i} \xi\left(W^*+\Delta W_{2}\right)}\right]}{E\left[e^{\mathrm{i} \xi\left(W^*+\Delta_{2}\right)}\right]} \\
&=\frac{E\left[W^* e^{\mathrm{i} \xi\left(W^*+\Delta W_{2}\right)}\right]}{E\left[e^{\mathrm{i} \xi\left(W^*+\Delta W_{2}\right)}\right]}=\frac{E\left[W^* e^{\mathrm{i} \xi W^*}\right] }{E\left[e^{\mathrm{i} \xi W^*}\right]} \\
&=\frac{-\mathbf{i}(d / d \xi) E\left[e^{\mathrm{i} \xi W^*}\right]}{E\left[e^{\mathrm{i} \xi W^*}\right]}=-\mathbf{i} \frac{(d / d \xi) \phi_{W^*}(\xi)}{\phi_{W^*}(\xi)}.
\end{align*}
Integrating $\Psi^+_{2,\lambda_1,\lambda_{2,1},\lambda_{2,2}}(\xi, h_1)$ with respect to $\xi$ and using Assumption \ref{smoothness1} and \ref{smoothness2} gives
\begin{footnotesize}
\begin{align*}
    &\quad\int \Psi^+_{2,\lambda_1,\lambda_{2,1},\lambda_{2,2}}(\xi, h_1) d\xi\\
    &\preceq \int \frac{1}{|\phi_{W_2}(\xi)|}\left(\frac{|\phi'_{W^*}(\xi)|}{|\phi_{W^*}(\xi)|}\int_{\xi}^{\pm\infty} \mathbf{1}\left\{|t|\leq h_1^{-1} \right\}|t|^{\lambda_1} \left(\sup_{(y,x)\in\mathbb{S}_{(Y,X)}}\left|\phi_{f^{(\lambda_{2,1},\lambda_{2,2})}
    _{Y,X,W^*}(y,x,\cdot)}(t)\right|\right)dt \right.\\
    &\quad \left.+ \mathbf{1}\left\{ |\xi|\leq h_1^{-1}\right\} |\xi|^{\lambda_1} \left(\sup_{(y,x)\in\mathbb{S}_{(Y,X)}}\left|\phi_{f^{(\lambda_{2,1},\lambda_{2,2})}
    _{Y,X,W^*}(y,x,\cdot)}(\xi)\right|\right)\right) d\xi\\
    &\preceq \int \frac{1}{|\phi_{W_2}(\xi)|}\left(\frac{|\phi'_{W^*}(\xi)|}{|\phi_{W^*}(\xi)|}\mathbf{1}\left\{|\xi|\leq h_1^{-1} \right\}\int_{|\xi|}^{h_1^{-1}}  |t|^{\lambda_1} \left(\sup_{(y,x)\in\mathbb{S}_{(Y,X)}}\left|\phi_{f^{(\lambda_{2,1},\lambda_{2,2})}
    _{Y,X,W^*}(y,x,\cdot)}(t)\right|\right)dt \right.\\
    &\quad + \mathbf{1}\left\{|\xi|\leq h_1^{-1}\right\} |\xi|^{\lambda_1} \left(\sup_{(y,x)\in\mathbb{S}_{(Y,X)}}\left|\phi_{f^{(\lambda_{2,1},\lambda_{2,2})}
    _{Y,X,W^*}(y,x,\cdot)}(\xi)\right|\right)\bigg) d\xi\\
    &\preceq \int (1+|\xi|)^{-\gamma_2}\exp(-\alpha_2|\xi|^{\beta_2})\mathbf{1}\left\{|\xi|\leq h_1^{-1} \right\}\\
    &\quad \times \left(\left(1+|\xi|\right)^{\gamma_*}\int_{0}^{h_1^{-1}}  |t|^{\lambda_1} (1+|t|)^{\gamma_{\phi}} \exp \left(\alpha_{\phi}|t|^{\beta_{\phi}}\right)dt + |\xi|^{\lambda_1} (1+|\xi|)^{\gamma_{\phi}}\exp\left(\alpha_{\phi}|\xi|^{\beta_{\phi}}\right)\right) d\xi\\
    &\preceq (1+h_1^{-1})^{1-\gamma_2}\exp\left(-\alpha_2\left(h_1^{-1}\right)^{\beta_2}\right)\\
    &\quad \times \bigg((1+h_1^{-1})^{\gamma_{*}}(1+h_1^{-1})^{1+\gamma_{\phi}+\lambda_1}\exp\left(\alpha_{\phi}\left(h_1^{-1}\right)^{\beta_{\phi}}\right) + (1+h_1^{-1})^{\gamma_{\phi}+\lambda_1}\exp\left(\alpha_{\phi}\left(h_1^{-1}\right)^{\beta_{\phi}}\right)\bigg)\\
    &\preceq (h_1^{-1})^{2-\gamma_2+\gamma_{\phi}+\gamma_*+\lambda_1}\exp\left(-\alpha_2\left(h_1^{-1}\right)^{\beta_2}\right)\exp\left(\alpha_{\phi}\left(h_1^{-1}\right)^{\beta_{\phi}}\right)
\end{align*}    
\end{footnotesize}
\begin{align*}
    \Psi^+_{3,\lambda_1,\lambda_{2,1},\lambda_{2,2}}(\xi,h_1)&=\sup_{(y,x,w^*)\in\mathbb{S}_{(Y,X,W^*)}}\left|\Psi_{3,\lambda_1,\lambda_{2,1},\lambda_{2,2}}(\xi, y, x, w^*, h_1)\right|\\
    &\preceq\sup_{(y,x,w^*)\in\mathbb{S}_{(Y,X,W^*)}} \frac{|\phi_{W^*}(\xi)|}{|\phi_{W_2}(\xi)|}  |\xi|^{\lambda_1}|e^{-\mathbf i\xi w^*}| |\phi_{K}(h_1\xi)|\\
    &=\frac{|\phi_{W^*}(\xi)|}{|\phi_{W_2}(\xi)|} |\xi|^{\lambda_1} |\phi_{K}(h_1\xi)|.
\end{align*}
so that
\begin{align*}
    &\quad (h_{2,1}^{-1})^{1+\lambda_{2,1}}(h_{2,2}^{-1})^{1+\lambda_{2,2}}\int \Psi^+_{3,\lambda_1,\lambda_{2,1},\lambda_{2,2}}(\xi,h_1) d\xi\\
    &\preceq (h_{2,1}^{-1})^{1+\lambda_{2,1}}(h_{2,2}^{-1})^{1+\lambda_{2,2}}\int\frac{|\phi_{W^*}(\xi)|}{|\phi_{W_2}(\xi)|}|\xi|^{\lambda_1}\mathbf{1}\left\{|\xi|\leq h_1^{-1}\right\}d\xi\\
    &\preceq (h_{2,1}^{-1})^{1+\lambda_{2,1}}(h_{2,2}^{-1})^{1+\lambda_{2,2}}\int_0^{h_1^{-1}}\frac{|\phi_{W^*}(\xi)|}{|\phi_{W_2}(\xi)|}|\xi|^{\lambda_1}d\xi\\
    &\preceq (h_{2,1}^{-1})^{1+\lambda_{2,1}}(h_{2,2}^{-1})^{1+\lambda_{2,2}}(h_1^{-1})^{1-\gamma_2+\gamma_{\phi}+\lambda_1}\exp\left(-\alpha_2\left(h_1^{-1}\right)^{\beta_2}\right)\exp\left(\alpha_{\phi}\left(h_1^{-1}\right)^{\beta_{\phi}}\right)
\end{align*}
Collecting together these rates delivers the conclusion of the lemma. 
\end{proof}
\begin{lemma}\label{uniform rate hat phi f}
Suppose the conditions for Lemma \ref{decomposition} hold, then 
\begin{footnotesize}
\begin{align*}
    &\sup_{(y,x)\in\mathbb{S}_{(Y,X)}}\left|\hat E\left[e^{\mathbf{i}\xi W_2}\frac{1}{h_{2,1}^{1+\lambda_{2,1}}h_{2,2}^{1+\lambda_{2,2}}}G_Y^{(\lambda_{2,1})}\left(\frac{y-Y}{h_{2,1}}\right)G_X^{(\lambda_{2,2})}\left(\frac{x-X}{h_{2,2}}\right)\right]\right.\\
    &\quad\quad\quad\quad\quad\quad\left.-E\left[e^{\mathbf{i}\xi W_2}\frac{1}{h_{2,1}^{1+\lambda_{2,1}}h_{2,2}^{1+\lambda_{2,2}}}G_Y^{(\lambda_{2,1})}\left(\frac{y-Y}{h_{2,1}}\right)G_X^{(\lambda_{2,2})}\left(\frac{x-X}{h_{2,2}}\right)\right]\right|\preceq \frac{1}{h_{2,1}^{1+\lambda_{2,1}}h_{2,2}^{1+\lambda_{2,2}}\sqrt{n}}.
\end{align*}
\end{footnotesize}
\end{lemma}
\begin{proof}
Denote the Fourier transform of $f_{Y,X,W_2}(y,x,w^*)$ as $\phi_{f_2}(t_1, t_2, \xi)$.
\begin{footnotesize}
\begin{align*}
   &\quad E\left[e^{\mathbf{i}\xi W_2}\frac{1}{h_{2,1}^{1+\lambda_{2,1}}h_{2,2}^{1+\lambda_{2,2}}}G_Y^{(\lambda_{2,1})}\left(\frac{y-Y}{h_{2,1}}\right)G_X^{(\lambda_{2,2})}\left(\frac{x-X}{h_{2,2}}\right)\right]\\
   &=\int\int\int e^{\mathbf{i}\xi \tilde w}\frac{1}{h_{2,1}^{1+\lambda_{2,1}}h_{2,2}^{1+\lambda_{2,2}}}G_Y^{(\lambda_{2,1})}\left(\frac{y-\tilde y}{h_{2,1}}\right)G_X^{(\lambda_{2,2})}\left(\frac{x-\tilde x}{h_{2,2}}\right)f_{Y,X,W_2}(\tilde y, \tilde x, \tilde w)d\tilde w d\tilde y d\tilde x\\
   &=\int\int\int e^{\mathbf{i}\xi \tilde w}\left(\frac{1}{(2\pi)^2}\int \int e^{-\mathbf{i}t_1(x-\tilde x)-\mathbf{i}t_2(y-\tilde y)} (\mathbf{i}t_1)^{\lambda_{2,1}} (\mathbf{i}t_2)^{\lambda_{2,2}}\phi_{G_X}(t_1h_{2,1})\phi_{G_Y}(t_2h_{2,2})dt_1 dt_2\right)\times\\
   &\quad\quad\quad\quad\quad\quad\quad\quad\quad\quad\quad\quad\quad\quad\quad\quad\quad\quad\quad\quad\quad\quad\quad\quad\quad\quad\quad\quad\quad\quad\quad\quad\quad f_{Y,X,W_2}(\tilde y, \tilde x, \tilde w)d\tilde w d\tilde y d\tilde x\\
   &=\frac{1}{(2\pi)^2}\int\int (\mathbf{i}t_1)^{\lambda_{2,1}} (\mathbf{i}t_2)^{\lambda_{2,2}} e^{-\mathbf{i}t_1 y-\mathbf{i}t_2 x}\phi_{G_Y}(t_1h_{2,1})\phi_{G_X}(t_2h_{2,2})\phi_{f_2}(t_1, t_2, \xi)dt_1 dt_2
\end{align*}
\end{footnotesize}
so
\begin{footnotesize}
\begin{align*}
    &\quad\sup_{(y,x)\in\mathbb{S}_{(Y,X)}}\left|\hat E\left[e^{\mathbf{i}\xi W_2}\frac{1}{h_{2,1}^{1+\lambda_{2,1}}h_{2,2}^{1+\lambda_{2,2}}}G_Y^{(\lambda_{2,1})}\left(\frac{y-Y}{h_{2,1}}\right)G_X^{(\lambda_{2,2})}\left(\frac{x-X}{h_{2,2}}\right)\right]\right.\\
    &\quad\quad\quad\left.-E\left[e^{\mathbf{i}\xi W_2}\frac{1}{h_{2,1}^{1+\lambda_{2,1}}h_{2,2}^{1+\lambda_{2,2}}}G_Y^{(\lambda_{2,1})}\left(\frac{y-Y}{h_{2,1}}\right)G_X^{(\lambda_{2,2})}\left(\frac{x-X}{h_{2,2}}\right)\right]\right|\\
    &=\sup_{(y,x)\in\mathbb{S}_{(Y,X)}}\left|\frac{1}{(2\pi)^2}\int\int (\mathbf{i}t_1)^{\lambda_{2,1}} (\mathbf{i}t_2)^{\lambda_{2,2}} e^{-\mathbf{i}t_1 y-\mathbf{i}t_2 x}\phi_{G_Y}(t_1h_{2,1})\phi_{G_X}(t_2h_{2,2})\times\right.\\
    &\quad\quad\quad\quad\quad\quad\quad\quad\quad\quad\left.\left(\frac{1}{n}\sum_{j=1}^n e^{\mathbf{i}t_1 Y_j+\mathbf{i}t_2 X_j+\mathbf{i}\xi W_{2,j}}-\phi_{f_2}(t_1, t_2, \xi)\right)dt_1 dt_2\right|\\
    &\leq\frac{1}{(2\pi)^2}\int\int |t_1|^{\lambda_{2,1}} |t_2|^{\lambda_{2,2}} \left|\phi_{G_Y}(t_1h_{2,1})\phi_{G_X}(t_2h_{2,2})\right|\left|\frac{1}{n}\sum_{j=1}^n e^{\mathbf{i}t_1 Y_j+\mathbf{i}t_2 X_j+\mathbf{i}\xi W_{2,j}}-\phi_{f_2}(t_1, t_2, \xi)\right| dt_2.
\end{align*}
\end{footnotesize}
Then since 
\begin{footnotesize}
\begin{align*}
    &\quad E\left[\left(\frac{1}{n}\sum_{j=1}^n e^{\mathbf{i}t_1 Y_j+\mathbf{i}t_2 X_j+\mathbf{i}\xi W_{2,j}}-\phi_{f_2}(t_1, t_2, \xi)\right)^2\right]\\
    &=\frac{1}{n}E\left[\left(e^{\mathbf{i}t_1 Y_j+\mathbf{i}t_2 X_j+\mathbf{i}\xi W_{2,j}}-\phi_{f_2}(t_1, t_2, \xi)\right)^2\right]\\
    &\leq \frac{1}{n}E\left[\left(e^{\mathbf{i}t_1 Y_j+\mathbf{i}t_2 X_j+\mathbf{i}\xi W_{2,j}}\right)^2\right]=O\left(\frac{1}{n}\right),
\end{align*}
\end{footnotesize}
I have
\begin{footnotesize}
\begin{align*}
    &\quad\sup_{(y,x)\in\mathbb{S}_{(Y,X)}}\left|\hat E\left[e^{\mathbf{i}\xi W_2}\frac{1}{h_{2,1}^{1+\lambda_{2,1}}h_{2,2}^{1+\lambda_{2,2}}}G_Y^{(\lambda_{2,1})}\left(\frac{y-Y}{h_{2,1}}\right)G_X^{(\lambda_{2,2})}\left(\frac{x-X}{h_{2,2}}\right)\right]\right.\\
    &\quad\quad\quad\left.-E\left[e^{\mathbf{i}\xi W_2}\frac{1}{h_{2,1}^{1+\lambda_{2,1}}h_{2,2}^{1+\lambda_{2,2}}}G_Y^{(\lambda_{2,1})}\left(\frac{y-Y}{h_{2,1}}\right)G_X^{(\lambda_{2,2})}\left(\frac{x-X}{h_{2,2}}\right)\right]\right|\\
    &\preceq \frac{1}{\sqrt{n}}\int\int |t_1|^{\lambda_{2,1}} |t_2|^{\lambda_{2,2}} \left|\phi_{G_Y}(t_1h_{2,1})\phi_{G_X}(t_2h_{2,2})\right| dt_1dt_2\\
    &\preceq \frac{1}{\sqrt{n}}\int_0^{h_{2,2}^{-1}}\int_0^{h_{2,1}^{-1}} |t_1|^{\lambda_{2,1}} |t_2|^{\lambda_{2,2}} dt_1dt_2\preceq \frac{1}{h_{2,1}^{1+\lambda_{2,1}}h_{2,2}^{1+\lambda_{2,2}}\sqrt{n}}.
\end{align*}
\end{footnotesize}
\end{proof}

\begin{lemma}\label{CLTgeneral}
For a finite integer J, let $\left\{P_{n,j}(a)\right\}$ be a sequence of nonrandom real-valued continuously differentiable functions of random vectors $a$, for $j=1,...,J$ respectively. Let $A$ be a random vector. If for each $n$, $\sigma_n = \sqrt{var \left(\sum_{j=1}^J P_{n,j}(A)\right)}$ exist, and $\sigma_n>0$ for $n$ sufficiently large, then
\begin{align*}
    \sigma_n^{-1}n^{1/2}\left(\hat E\left[\sum_{j=1}^JP_{n,j}(A)\right]\right)\xrightarrow{d}N(0,1).
\end{align*}
\end{lemma}
\begin{proof}
Denote the centered version of the summands as $Z_{n,j}(A)\equiv P_{n,j}(A)-E\left[P_{n,j}(A)\right]$. Then $\sigma_n^2\equiv E\left[\left(\sum_{j=1}^JZ_{n,j}(A)\right)^2\right]$ It's sufficient to verify that the Lindeberg condition holds:
\begin{align*}
    \sum_{i=1}^n E\left[\mathbf 1\left\{\left|\frac{\sum_{j=1}^JZ_{n,j}(A)}{\sqrt{n}\sigma_n}\right|>\epsilon\right\}\left|\frac{\sum_{j=1}^JZ_{n,j}(A)}{\sqrt{n}\sigma_n}\right|^2\right]\rightarrow 0, \text{\ as\ } n\rightarrow \infty,
\end{align*}
for all $\epsilon>0$. Note that as $n\rightarrow \infty$
\begin{align*}
    E\left[\left|\frac{\sum_{j=1}^JZ_{n,j}(A)}{\sqrt{n} \sigma_n}\right|^2\right]=\frac{1}{n}E\left[\frac{\left|\sum_{j=1}^JZ_{n,j}(A)\right|^2}{ E\left[\left(\sum_{j=1}^JZ_{n,j}(A)\right)^2\right]}\right]\rightarrow 0
\end{align*}
By Markov Inequality, for any $\epsilon>0$, as $n\rightarrow \infty$, $Pr\left(\left|\frac{\sum_{j=1}^JZ_{n,j}(A)}{\sqrt{n}\sigma_n}\right|>\epsilon\right)\rightarrow 0$, which implies $\mathbf 1\left\{\left|\frac{\sum_{j=1}^JZ_{n,j}(A)}{\sqrt{n}\sigma_n}\right|>\epsilon\right\}=o_p(1)$. Note that $ \mathbf 1\left\{\left|\frac{\sum_{j=1}^JZ_{n,j}(A)}{\sqrt{n}\sigma_n}\right|>\epsilon\right\}\left|\frac{\sum_{j=1}^JZ_{n,j}(A)}{\sigma_n}\right|^2\leq \left|\frac{\sum_{j=1}^JZ_{n,j}(A)}{ \sigma_n}\right|^2$ and that $E\left[\left|\frac{\sum_{j=1}^JZ_{n,j}(A)}{ \sigma_n}\right|^2\right]=1<\infty$. By Dominated Convergence Theorem
\begin{align*}
    E\left[\mathbf 1\left\{\left|\frac{\sum_{j=1}^JZ_{n,j}(A)}{\sqrt{n}\sigma_n}\right|>\epsilon\right\}\left|\frac{\sum_{j=1}^JZ_{n,j}(A)}{\sigma_n}\right|^2\right]\rightarrow 0, \text{\ as\ } n\rightarrow \infty.
\end{align*}
The conclusion follows.
\end{proof}

\begin{proof}[\proofname\ of Theorem \ref{Omegaprop}]
(i) The fact that $E[L_{\lambda_1, \lambda_{2,1}, \lambda_{2,2}}(y,x,w^*,h)]=0$ follows directly from Eq. (\ref{L}). Next, with fixed value of $h$, Assumption \ref{phik} and \ref{finite variance} ensures the existence and finiteness of 
\begin{align*}
    E[L_{\lambda_1, \lambda_{2,1}, \lambda_{2,2}}^2(y,x,w^*,h)]&=E\left[\left(\hat E\left[l_{\lambda_1, \lambda_{2,1}, \lambda_{2,2}}(y,x,w^*,h;Y, X, W_1, W_2)\right]\right)^2\right]\\&=n^{-1}E\left[\left(l_{\lambda_1, \lambda_{2,1}, \lambda_{2,2}}(y,x,w^*,h;Y, X, W_1, W_2)\right)^2\right]
    \\&=n^{-1}\Omega_{\lambda_1, \lambda_{2,1}, \lambda_{2,2}}(y,x,w^*,h).
\end{align*}
From Eq. (\ref{L}), I have
\begin{footnotesize}
\begin{align}
    &\quad\Omega_{\lambda_1, \lambda_{2,1}, \lambda_{2,2}}(y,x,w^*,h) \notag\\
    &\equiv E\left[n\left(\bar g_{\lambda_1, \lambda_{2,1}, \lambda_{2,2}}(y,x,w^*,h)-g_{\lambda_1, \lambda_{2,1}, \lambda_{2,2}}(y,x,w^*,h)\right)^2\right]\notag\\
    &=E\Bigg[\Bigg(\int \Psi_{1,\lambda_1, \lambda_{2,1}, \lambda_{2,2}
    }(\xi, y, x, w^*, h_1)n^{1/2}\delta\hat\theta(\xi)d\xi\notag\\
    &\quad\quad\quad+\int \Psi_{2,\lambda_1, \lambda_{2,1}, \lambda_{2,2}
    }(\xi, y, x, w^*, h_1)n^{1/2}\delta\hat\phi_{W_2}(\xi)d\xi\notag\\
    &\quad\quad\quad+\int \Psi_{3,\lambda_1, \lambda_{2,1}, \lambda_{2,2}
    }(\xi, y, x, w^*, h_1)n^{1/2}\delta\hat\phi_{f^{(\lambda_{2,1}, \lambda_{2,2})}
    _{Y,X,W_2}(y,x,\cdot)}(\xi) d\xi\Bigg)^2\Bigg]\notag\\
    &=\int\int \Psi_{1,\lambda_1, \lambda_{2,1}, \lambda_{2,2}
    }(\xi, y, x, w^*, h_1) E\left[n\delta\hat\theta(\xi)\delta\hat\theta^{\dagger}(\zeta)\right] \left(\Psi_{1, \lambda_1, \lambda_{2,1}, \lambda_{2,2}}(\zeta, y, x, w^*, h_1)\right)^\dagger d\xi d\zeta \notag\\
    &\quad +\int\int \Psi_{2,\lambda_1, \lambda_{2,1}, \lambda_{2,2}
    }(\xi, y, x, w^*, h_1) E\left[n\delta\hat\phi_{W_2}(\xi)\delta\hat\phi_{W_2}^{\dagger}(\zeta)\right] \left(\Psi_{2,\lambda_1, \lambda_{2,1}, \lambda_{2,2}}(\zeta, y, x, w^*, h_1)\right)^\dagger d\xi d\zeta\notag\\
    &\quad +\int\int (h_{2,1}^{-2})^{1+\lambda_{2,1}}(h_{2,2}^{-2})^{1+\lambda_{2,2}}\Psi_{3,\lambda_1, \lambda_{2,1}, \lambda_{2,2}
    }(\xi, y, x, w^*, h_1)\times  \notag \\
    &\quad\quad\quad\quad\quad\quad E\left[nh_{2,1}^{2+2\lambda_{2,1}}h_{2,2}^{2+2\lambda_{2,2}}\delta\hat\phi_{f^{(\lambda_{2,1}, \lambda_{2,2})}
    _{Y,X,W_2}(y,x,\cdot)}(\xi)\delta\hat\phi^\dagger_{f^{(\lambda_{2,1}, \lambda_{2,2})}
    _{Y,X,W_2}(y,x,\cdot)}(\zeta)\right]\times\notag\\
    &\quad\quad\quad\quad\quad\quad\left(\Psi_{3,\lambda_1, \lambda_{2,1}, \lambda_{2,2}
    }(\zeta, y, x, w^*, h_1)\right)^\dagger  d\xi d\zeta\notag\\
    &\quad+\int\int \Psi_{1,\lambda_1, \lambda_{2,1}, \lambda_{2,2}
    }(\xi, y, x, w^*, h_1) E\left[n\delta\hat\theta(\xi)\delta\hat\phi_{W_2}^{\dagger}(\zeta)\right] \left(\Psi_{2,\lambda_1, \lambda_{2,1}, \lambda_{2,2}}(\zeta, y, x, w^*, h_1)\right)^\dagger d\xi d\zeta \notag\\
    &\quad +\int\int \Psi_{1,\lambda_1, \lambda_{2,1}, \lambda_{2,2}
    }(\xi, y, x, w^*, h_1) E\left[n\delta\hat\theta(\xi)h_{2,1}^{1+\lambda_{2,1}}h_{2,2}^{1+\lambda_{2,2}}\delta\hat\phi^\dagger_{f^{(\lambda_{2,1}, \lambda_{2,2})}
    _{Y,X,W_2}(y,x,\cdot)}(\zeta)\right]\times\notag\\
    &\quad\quad\quad\quad\quad\quad\quad\quad\quad\quad\quad\quad\quad\quad\quad\quad\quad\quad (h_{2,1}^{-1})^{1+\lambda_{2,1}}(h_{2,2}^{-1})^{1+\lambda_{2,2}} \left(\Psi_{3,\lambda_1, \lambda_{2,1}, \lambda_{2,2}
    }(\zeta, y, x, w^*, h_1)\right)^\dagger  d\xi d\zeta\notag\\
    &\quad +\int\int \Psi_{2,\lambda_1, \lambda_{2,1}, \lambda_{2,2}
    }(\xi, y, x, w^*, h_1) E\left[n\delta\hat\phi_{W_2}(\xi)h_{2,1}^{1+\lambda_{2,1}}h_{2,2}^{1+\lambda_{2,2}}\delta\hat\phi^\dagger_{f^{(\lambda_{2,1}, \lambda_{2,2})}
    _{Y,X,W_2}(y,x,\cdot)}(\zeta)\right]\times\notag\\ &\quad\quad\quad\quad\quad\quad\quad\quad\quad\quad\quad\quad\quad\quad\quad\quad\quad\quad (h_{2,1}^{-1})^{1+\lambda_{2,1}}(h_{2,2}^{-1})^{1+\lambda_{2,2}}\left(\Psi_{3,\lambda_1, \lambda_{2,1}, \lambda_{2,2}
    }(\zeta, y, x, w^*, h_1)\right)^\dagger d\xi d\zeta\notag\\
    &\quad+\int\int \Psi_{2,\lambda_1, \lambda_{2,1}, \lambda_{2,2}
    }(\xi, y, x, w^*, h_1) E\left[n\delta\hat\phi_{W_2}(\zeta)\delta\hat\theta^{\dagger}(\xi)\right] \left(\Psi_{1, \lambda_1, \lambda_{2,1}, \lambda_{2,2}}(\zeta, y, x, w^*, h_1)\right)^\dagger d\xi d\zeta \notag\\
    &\quad +\int\int  (h_{2,1}^{-1})^{1+\lambda_{2,1}}(h_{2,2}^{-1})^{1+\lambda_{2,2}}\Psi_{3,\lambda_1, \lambda_{2,1}, \lambda_{2,2}
    }(\zeta, y, x, w^*, h_1)\times\notag\\
    &\quad\quad\quad\quad E\left[nh_{2,1}^{1+\lambda_{2,1}}h_{2,2}^{1+\lambda_{2,2}}\delta\hat\phi_{f^{(\lambda_{2,1}, \lambda_{2,2})}
    _{Y,X,W_2}(y,x,\cdot)}(\zeta)\delta\hat\theta^{\dagger}(\xi)\right]\times\left(\Psi_{1, \lambda_1, \lambda_{2,1}, \lambda_{2,2}}(\xi, y, x, w^*, h_1)\right)^\dagger d\xi d\zeta\notag\\
    &\quad +\int\int  (h_{2,1}^{-1})^{1+\lambda_{2,1}}(h_{2,2}^{-1})^{1+\lambda_{2,2}}\Psi_{3,\lambda_1, \lambda_{2,1}, \lambda_{2,2}}(\zeta, y, x, w^*, h_1)\times\notag\\ 
    &\quad\quad\quad\quad E\left[nh_{2,1}^{1+\lambda_{2,1}}h_{2,2}^{1+\lambda_{2,2}}\delta\phi_{f^{(\lambda_{2,1}, \lambda_{2,2})}_{Y,X,W_2}(y,x,\cdot)}(\zeta)\delta\hat\phi_{W_2}^{\dagger}(\xi)\right]\times\left(\Psi_{2,\lambda_1, \lambda_{2,1}, \lambda_{2,2}}(\xi, y, x, w^*, h_1)\right)^\dagger d\xi d\zeta \label{Omega}
\end{align}
\end{footnotesize}
Note that I have
\begin{footnotesize}
\begin{align*}
    E\left[n\delta\hat\theta(\xi)\delta\hat\theta^{\dagger}(\zeta)\right]&=E\left[n\left(\hat\theta(\xi)-\theta(\xi)\right)\left(\hat\theta^{\dagger}(\zeta)-\theta^{\dagger}(\zeta)\right)\right]\\
    &=E\left[\left(W_1e^{\mathbf{i}\xi W_2}-\theta(\xi)\right)\left(W_1e^{-\mathbf{i}\zeta W_2}-\theta(-\zeta)\right)\right]\\
    &=E\left[W_1^2e^{\mathbf{i}(\xi-\zeta) W_2}\right]-\theta(\xi)E\left[W_1e^{-\mathbf{i}\zeta W_2}\right]-E\left[W_1e^{\mathbf{i}\xi W_2}\right]\theta^{\dagger}(\zeta)+\theta(\xi)\theta^{\dagger}(\zeta)\\
    &=E\left[W_1^2e^{\mathbf{i}(\xi-\zeta) W_2}\right]-\theta(\xi)\theta(-\zeta),
\end{align*}
\begin{align*}
    E\left[n\delta\hat\theta(\xi)\delta\hat\phi_{W_2}^{\dagger}(\zeta)\right]&=E\left[n\left(\hat\theta(\xi)-\theta(\xi)\right)\left(\hat\phi_{W_2}^{\dagger}(\zeta)-\phi_{W_2}^{\dagger}(\zeta)\right)\right]\\
    &=E\left[\left(W_1e^{\mathbf{i}\xi W_2}-\theta(\xi)\right)\left(e^{-\mathbf{i}\zeta W_2}-\phi_{W_2}^{\dagger}(\zeta)\right)\right]\\
    &=E\left[W_1e^{\mathbf{i}\xi W_2}e^{-\mathbf{i}\zeta W_2}\right]-\theta(\xi)E\left[e^{-\mathbf{i}\zeta W_2}\right]-E\left[W_1e^{\mathbf{i}\xi W_2}\right]\phi^{\dagger}_{W_2}(\zeta)+\theta(\xi)\phi^{\dagger}_{W_2}(\zeta)\\
    &=E\left[W_1e^{\mathbf{i}(\xi-\zeta) W_2}\right]-\theta(\xi)\phi_{W_2}(-\zeta),
\end{align*}
\begin{align*}
    &\quad E\left[n\delta\hat\theta(\xi)h_{2,1}^{1+\lambda_{2,1}}h_{2,2}^{1+\lambda_{2,2}}\delta\hat\phi^\dagger_{f^{(\lambda_{2,1}, \lambda_{2,2})}_{Y,X,W_2}(y,x,\cdot)}(\zeta)\right]\\
    &=E\left[\left(W_1e^{\mathbf{i}\xi W_2}-\theta(\xi)\right)\times\right.\\
    &\quad\quad\quad\quad\left.\left(e^{-\mathbf{i}\zeta W_2}G_Y\left(\frac{y-Y}{h_{2,1}}\right)G^{(\lambda_{2,1}, \lambda_{2,2})}_X\left(\frac{x-X}{h_{2,2}}\right)-h_{2,1}^{1+\lambda_{2,1}}h_{2,2}^{1+\lambda_{2,2}}\phi^\dagger_{f^{(\lambda_{2,1}, \lambda_{2,2})}_{Y,X,W_2}(y,x,\cdot)}(\zeta)\right)\right]\\
    &=E\left[W_1e^{\mathbf{i}(\xi-\zeta) W_2}G_Y\left(\frac{y-Y}{h_{2,1}}\right)G^{(\lambda_{2,1}, \lambda_{2,2})}_X\left(\frac{x-X}{h_{2,2}}\right)\right]-\theta(\xi)h_{2,1}^{1+\lambda_{2,1}}h_{2,2}^{1+\lambda_{2,2}}\phi_{f^{(\lambda_{2,1}, \lambda_{2,2})}_{Y,X,W_2}(y,x,\cdot)}(-\zeta),\\
    \\
    &\quad E\left[n\delta\hat\phi_{W_2}(\xi)\delta\hat\phi_{W_2}^{\dagger}(\zeta)\right]=E\left[\left(e^{\mathbf{i}\xi W_2}-\phi_{W_2}(\xi)\right)\left(e^{-\mathbf{i}\zeta W_2}-\phi_{W_2}^{\dagger}(\zeta)\right)\right]\\
    &=E\left[e^{\mathbf{i}(\xi-\zeta) W_2}\right]-\phi_{W_2}(\xi)\phi_{W_2}( -\zeta),
\end{align*}
\begin{align*}
    &\quad E\left[n\delta\hat\phi_{W_2}(\xi)h_{2,1}^{1+\lambda_{2,1}}h_{2,2}^{1+\lambda_{2,2}}\delta\hat\phi^\dagger_{f^{(\lambda_{2,1}, \lambda_{2,2})}_{Y,X,W_2}(y,x,\cdot)}(\zeta)\right]\\
    &=E\left[\left(e^{\mathbf{i}\xi W_2}-\phi_{W_2}(\xi)\right)\times\right.\\
    &\quad\quad\quad\quad\left.\left(e^{-\mathbf{i}\zeta W_2}G_Y\left(\frac{y-Y}{h_{2,1}}\right)G^{(\lambda_{2,1}, \lambda_{2,2})}_X\left(\frac{x-X}{h_{2,2}}\right)-h_{2,1}^{1+\lambda_{2,1}}h_{2,2}^{1+\lambda_{2,2}}\phi^\dagger_{f^{(\lambda_{2,1}, \lambda_{2,2})}_{Y,X,W_2}(y,x,\cdot)}(\zeta)\right)\right]\\
    &=E\left[e^{\mathbf{i}(\xi-\zeta) W_2}G_Y\left(\frac{y-Y}{h_{2,1}}\right)G^{(\lambda_{2,1}, \lambda_{2,2})}_X\left(\frac{x-X}{h_{2,2}}\right)\right]-\phi_{W_2}(\xi) h_{2,1}^{1+\lambda_{2,1}}h_{2,2}^{1+\lambda_{2,2}}\phi_{f^{(\lambda_{2,1}, \lambda_{2,2})}_{Y,X,W_2}(y,x,\cdot)}(-\zeta),
\end{align*}
\begin{align*}
    &\quad E\left[nh_{2,1}^{1+\lambda_{2,1}}h_{2,2}^{1+\lambda_{2,2}}\delta \hat\phi_{f^{(\lambda_{2,1}, \lambda_{2,2})}_{Y,X,W_2}(y,x,\cdot)}(\xi)h_{2,1}^{1+\lambda_{2,1}}h_{2,2}^{1+\lambda_{2,2}}\delta \hat\phi^\dagger_{f^{(\lambda_{2,1}, \lambda_{2,2})}_{Y,X,W_2}(y,x,\cdot)}(\zeta)\right]\\
    &=E\left[\left(e^{\mathbf{i}\xi W_2}G_Y\left(\frac{y-Y}{h_{2,1}}\right)G^{(\lambda_{2,1}, \lambda_{2,2})}_X\left(\frac{x-X}{h_{2,2}}\right)-h_{2,1}^{1+\lambda_{2,1}}h_{2,2}^{1+\lambda_{2,2}}\phi_{f^{(\lambda_{2,1}, \lambda_{2,2})}_{Y,X,W_2}(y,x,\cdot)}(\xi)\right)\times\right.\\
    &\quad\quad\quad\quad\quad\quad\quad\quad\left.\left(e^{-\mathbf{i}\zeta W_2}G_Y\left(\frac{y-Y}{h_{2,1}}\right)G^{(\lambda_{2,1}, \lambda_{2,2})}_X\left(\frac{x-X}{h_{2,2}}\right)-h_{2,1}^{1+\lambda_{2,1}}h_{2,2}^{1+\lambda_{2,2}}\phi^\dagger_{f^{(\lambda_{2,1}, \lambda_{2,2})}_{Y,X,W_2}(y,x,\cdot)}(\zeta)\right)\right]\\
    &=E\left[e^{\mathbf{i}(\xi-\zeta)W_2}\left(G_Y\left(\frac{y-Y}{h_{2,1}}\right)G^{(\lambda_{2,1}, \lambda_{2,2})}_X\left(\frac{x-X}{h_{2,2}}\right)\right)^2\right]\\
    &\quad\quad\quad\quad-h_{2,1}^{2+2\lambda_{2,1}}h_{2,2}^{2+2\lambda_{2,2}}\phi_{f^{(\lambda_{2,1}, \lambda_{2,2})}_{Y,X,W_2}(y,x,\cdot)}(\xi)\phi_{f^{(\lambda_{2,1}, \lambda_{2,2})}_{Y,X,W_2}(y,x,\cdot)}(-\zeta).
\end{align*}
\end{footnotesize}
By Assumption \ref{finite variance},
\begin{footnotesize}
\begin{align*}
    \left| E\left[n\delta\hat\theta(\xi)\delta\hat\theta^{\dagger}(\zeta)\right]\right|&=\left|E\left[W_1^2e^{\mathbf{i}(\xi-\zeta) W_2}\right]-\theta(\xi)\theta(-\zeta)\right|\\
    &\leq E\left[W_1^2\left|e^{\mathbf{i}(\xi-\zeta) W_2}\right|\right]+E\left[\left|W_1\right|\left|e^{\mathbf{i}\xi W_2}\right|\right]E\left[\left|W_1\right|\left|e^{-\mathbf{i}\zeta W_2}\right|\right]\\
    &\leq E\left[W_1^2\right]+E\left[\left|W_1\right|\right]E\left[\left|W_1\right|\right]\preceq 1,
\end{align*}
\begin{align*}
    \left|E\left[n\delta\hat\theta(\xi)\delta\hat\phi_{W_2}^{\dagger}(\zeta)\right]\right|&=\left|\theta(\xi-\zeta)-\theta(\xi)\phi_{W_2}(-\zeta)\right|\\
    &\leq E\left[\left|W_1\right|\left|e^{\mathbf{i}(\xi-\zeta) W_2}\right|\right]+E\left[\left|W_1\right|\left|e^{\mathbf{i}\xi W_2}\right|\right]E\left[\left|e^{-\mathbf{i}\zeta W_2}\right|\right]\\
    &\leq 2E[|W_1|]\preceq 1,
\end{align*}
\end{footnotesize}
\begin{scriptsize}
\begin{align*}
    &\quad\sup_{x\in\mathbb{S}_X}\left|E\left[n\delta\hat\theta(\xi)h_{2,1}^{1+\lambda_{2,1}}h_{2,2}^{1+\lambda_{2,2}}\delta\hat\phi^\dagger_{f^{(\lambda_{2,1}, \lambda_{2,2})}_{Y,X,W_2}(y,x,\cdot)}(\zeta)\right]\right|\\
    &=\sup_{x\in\mathbb{S}_X}\left|E\left[W_1e^{\mathbf{i}(\xi-\zeta) W_2}G_Y\left(\frac{y-Y}{h_{2,1}}\right)G^{(\lambda_{2,1}, \lambda_{2,2})}_X\left(\frac{x-X}{h_{2,2}}\right)\right]-\theta(\xi)h_{2,1}^{1+\lambda_{2,1}}h_{2,2}^{1+\lambda_{2,2}}\phi_{f^{(\lambda_{2,1}, \lambda_{2,2})}_{Y,X,W_2}(y,x,\cdot)}(-\zeta)\right|\\
    &\leq E\left[\left|W_1\right|\left|e^{\mathbf{i}(\xi-\zeta) W_2}\right|\sup_{x\in\mathbb{S}_X}\left|G_Y\left(\frac{y-Y}{h_{2,1}}\right)G^{(\lambda_{2,1}, \lambda_{2,2})}_X\left(\frac{x-X}{h_{2,2}}\right)\right|\right]\\
    &\quad\quad\quad\quad+E\left[\left|W_1\right|\left|e^{\mathbf{i}\xi W_2}\right|\right]E\left[\left|e^{-\mathbf{i}\zeta W_2}\right|\sup_{x\in\mathbb{S}_X}\left|G_Y\left(\frac{y-Y}{h_{2,1}}\right)G^{(\lambda_{2,1}, \lambda_{2,2})}_X\left(\frac{x-X}{h_{2,2}}\right)\right|\right]\\
    &\leq 2E[|W_1|]\preceq 1,
\end{align*}    
\end{scriptsize}
where the last line follows by Assumption \ref{phik}.
Following the same logic, I have that
\begin{align*}
    &\left|E\left[n\delta\hat\phi_{W_2}(\xi)\delta\hat\phi^\dagger_{W_2}(\zeta)\right] \right|\preceq 1,\\
    &\sup_{x\in\mathbb{S}_X}\left|E\left[n\delta\hat\phi_{W_2}(\xi)h_{2,1}^{1+\lambda_{2,1}}h_{2,2}^{1+\lambda_{2,2}}\delta\hat\phi^\dagger_{f^{(\lambda_{2,1}, \lambda_{2,2})}_{Y,X,W_2}(y,x,\cdot)}(\zeta)\right]\right|\preceq 1,\\
    &\sup_{x\in\mathbb{S}_X}\left|E\left[nh_{2,1}^{2+2\lambda_{2,1}}h_{2,2}^{2+2\lambda_{2,2}}\delta\hat\phi_{f^{(\lambda_{2,1}, \lambda_{2,2})}_{Y,X,W_2}(y,x,\cdot)}(\xi)\delta\hat\phi^\dagger_{f^{(\lambda_{2,1}, \lambda_{2,2})}_{Y,X,W_2}(y,x,\cdot)}(\zeta)\right]\right|\preceq 1,\\
    &\left|E\left[n\delta\hat\phi_{W_2}(\zeta)\delta\hat\theta^{\dagger}(\xi)\right]\right| \preceq 1,\\
    &\sup_{x\in\mathbb{S}_X}\left|E\left[nh_{2,1}^{1+\lambda_{2,1}}h_{2,2}^{1+\lambda_{2,2}}\delta\hat\phi_{f^{(\lambda_{2,1}, \lambda_{2,2})}_{Y,X,W_2}(y,x,\cdot)}(\zeta)\delta\hat\theta^{\dagger}(\xi)\right]\right|\preceq 1,\\
    &\sup_{x\in\mathbb{S}_X}\left|E\left[nh_{2,1}^{1+\lambda_{2,1}}h_{2,2}^{1+\lambda_{2,2}}\delta\hat\phi_{f^{(\lambda_{2,1}, \lambda_{2,2})}_{Y,X,W_2}(y,x,\cdot)}(\zeta)\delta\hat\phi^{\dagger}_{W_2}(\xi)\right]\right| \preceq 1.
\end{align*}
It follows from Equation (\ref{Omega}) that
\begin{scriptsize}
\begin{align*}
    &\quad\Omega_{\lambda_1, \lambda_{2,1}, \lambda_{2,2}}(y,x,w^*,h)\\
    &\leq \int\int \left|\Psi_{1, \lambda_1, \lambda_{2,1}, \lambda_{2,2}}(\xi, y, x, w^*, h_1)\right|  \left|\left(\Psi_{1, \lambda_1, \lambda_{2,1}, \lambda_{2,2}}(\zeta, y, x, w^*, h_1)\right)^\dagger\right| d\xi d\zeta \notag\\
    &\quad +\int\int \left|\Psi_{2,\lambda_1, \lambda_{2,1}, \lambda_{2,2}}(\xi, y, x, w^*, h_1)\right| \left|\left(\Psi_{2,\lambda_1, \lambda_{2,1}, \lambda_{2,2}}(\zeta, y, x, w^*, h_1)\right)^\dagger\right|  d\xi d\zeta\notag\\
    &\quad +\int\int \left|(h_{2,1}^{-1})^{1+\lambda_{2,1}}(h_{2,2}^{-1})^{1+\lambda_{2,2}}\Psi_{3,\lambda_1, \lambda_{2,1}, \lambda_{2,2}}(\xi, y, x, w^*, h_1)\right|\times\\
    &\quad\quad\quad\quad\quad\quad\left|\left((h_{2,1}^{-1})^{1+\lambda_{2,1}}(h_{2,2}^{-1})^{1+\lambda_{2,2}}\Psi_{3,\lambda_1, \lambda_{2,1}, \lambda_{2,2}}(\zeta, y, x, w^*, h_1)\right)^\dagger\right|   d\xi d\zeta\notag\\
    &\quad+\int\int \left|\Psi_{1, \lambda_1, \lambda_{2,1}, \lambda_{2,2}}(\xi, y, x, w^*, h_1)\right| \left|\left(\Psi_{2,\lambda_1, \lambda_{2,1}, \lambda_{2,2}}(\zeta, y, x, w^*, h_1)\right)^\dagger\right|  d\xi d\zeta \notag\\
    &\quad +\int\int \left|\Psi_{1, \lambda_1, \lambda_{2,1}, \lambda_{2,2}}(\xi, y, x, w^*, h_1)\right| \left|\left((h_{2,1}^{-1})^{1+\lambda_{2,1}}(h_{2,2}^{-1})^{1+\lambda_{2,2}}\Psi_{3,\lambda_1, \lambda_{2,1}, \lambda_{2,2}}(\zeta, y, x, w^*, h_1)\right)^\dagger\right|  d\xi d\zeta\notag\\
    &\quad +\int\int \left|\Psi_{2,\lambda_1, \lambda_{2,1}, \lambda_{2,2}}(\xi, y, x, w^*, h_1)\right|\left|\left((h_{2,1}^{-1})^{1+\lambda_{2,1}}(h_{2,2}^{-1})^{1+\lambda_{2,2}}\Psi_{3,\lambda_1, \lambda_{2,1}, \lambda_{2,2}}(\zeta, y, x, w^*, h_1)\right)^\dagger\right|  d\xi d\zeta\notag\\
    &\quad+\int\int \left|\Psi_{2,\lambda_1, \lambda_{2,1}, \lambda_{2,2}}(\xi, y, x, w^*, h_1)\right| \left|\left(\Psi_{1, \lambda_1, \lambda_{2,1}, \lambda_{2,2}}(\zeta, y, x, w^*, h_1)\right)^\dagger\right| d\xi d\zeta \notag\\
    &\quad +\int\int  \left|(h_{2,1}^{-1})^{1+\lambda_{2,1}}(h_{2,2}^{-1})^{1+\lambda_{2,2}}\Psi_{3,\lambda_1, \lambda_{2,1}, \lambda_{2,2}}(\zeta, y, x, w^*, h_1)\right| \left|\left(\Psi_{1, \lambda_1, \lambda_{2,1}, \lambda_{2,2}}(\xi, y, x, w^*, h_1)\right)^\dagger\right|  d\xi d\zeta\notag\\
    &\quad +\int\int  \left|(h_{2,1}^{-1})^{1+\lambda_{2,1}}(h_{2,2}^{-1})^{1+\lambda_{2,2}}\Psi_{3,\lambda_1, \lambda_{2,1}, \lambda_{2,2}}(\zeta, y, x, w^*, h_1)\right| \left|\left(\Psi_{2,\lambda_1, \lambda_{2,1}, \lambda_{2,2}}(\xi, y, x, w^*, h_1)\right)^\dagger\right|  d\xi d\zeta\\
    &=\left(\int \left|\Psi_{1, \lambda_1, \lambda_{2,1}, \lambda_{2,2}}(\xi, y, x, w^*, h_1)\right|d\xi+\int \left|\Psi_{2,\lambda_1, \lambda_{2,1}, \lambda_{2,2}}(\xi, y, x, w^*, h_1)\right|d\xi\right.\\
    &\quad\quad\quad\quad\quad\quad\quad\quad\quad\quad\quad\quad\quad\quad\left.+\int \left|(h_{2,1}^{-1})^{1+\lambda_{2,1}}(h_{2,2}^{-1})^{1+\lambda_{2,2}}\Psi_{3,\lambda_1, \lambda_{2,1}, \lambda_{2,2}}(\xi, y, x, w^*, h_1)\right| d\xi\right)^2\\
    &\leq \left(\int \Psi_{1, \lambda_1, \lambda_{2,1}, \lambda_{2,2}}^+(\xi,h_1)d\xi+\int \Psi_{2,\lambda_1, \lambda_{2,1}, \lambda_{2,2}}^+(\xi,h_1)d\xi+\int (h_{2,1}^{-1})^{1+\lambda_{2,1}}(h_{2,2}^{-1})^{1+\lambda_{2,2}}\Psi_{3,\lambda_1, \lambda_{2,1}, \lambda_{2,2}}^+(\xi,h) d\xi\right)^2\\
    &=\left(\Psi_{\lambda_1, \lambda_{2,1}, \lambda_{2,2}}^+(h)\right)^2,
\end{align*}
\end{scriptsize}
where $\Psi_{1, \lambda_1, \lambda_{2,1}, \lambda_{2,2}}^+(\xi,h_1), \Psi_{2,\lambda_1, \lambda_{2,1}, \lambda_{2,2}}^+(\xi,h_1), \Psi_{3,\lambda_1, \lambda_{2,1}, \lambda_{2,2}}^+(\xi,h)$, and $\Psi_{\lambda_1, \lambda_{2,1}, \lambda_{2,2}}^+(h)$ are as defined in Lemma \ref{orderPsi} and 
\begin{scriptsize}
\begin{align*}
       &\quad\Psi_{\lambda_1, \lambda_{2,1}, \lambda_{2,2}}^+(h)\\
       &=O\left(\max\left\{(h_1^{-1})^{1+\gamma_*},(h_{2,1}^{-1})^{1+\lambda_{2,1}}(h_{2,2}^{-1})^{1+\lambda_{2,2}}\right\}(h_1^{-1})^{1-\gamma_2+\gamma_{\phi}+\lambda_1}\exp\left(\left(\alpha_{\phi}\mathbf{1}\{\beta_{\phi}=\beta_2\}-\alpha_2\right)\left(h_1^{-1}\right)^{\beta_2}\right)\right).
\end{align*} 
\end{scriptsize}
Hence I have proved Eq. (\ref{sqrtsupOmega})\\
\indent Next I show Eq. (\ref{supL}). From Eq. (\ref{L}) I have
\begin{scriptsize}
\begin{align*}
    &\quad\sup_{(y, x, w^*)\in\mathbb{S}_{(Y,X,W^*)}}\left|L_{\lambda_1, \lambda_{2,1}, \lambda_{2,2}}(y,x,w^*,h)\right|\\&=\sup_{(y, x, w^*)\in\mathbb{S}_{(Y,X,W^*)}} \Bigg|\int \Psi_{1, \lambda_1, \lambda_{2,1}, \lambda_{2,2}}(\xi, y, x, w^*, h_1)\left(\hat E[W_1e^{\mathbf{i}\xi W_2}]-E[W_1e^{\mathbf{i}\xi W_2}]\right)d\xi\\
    &\quad\quad\quad\quad+\int \Psi_{2,\lambda_1, \lambda_{2,1}, \lambda_{2,2}}(\xi, y, x, w^*, h_1)\left(\hat E[e^{\mathbf{i}\xi W_2}]-E[e^{\mathbf{i}\xi W_2}]\right)d\xi\notag\\
    &\quad\quad\quad\quad+\int \Psi_{3,\lambda_1, \lambda_{2,1}, \lambda_{2,2}}(\xi, y, x, w^*, h_1)\times\\
    &\quad\quad\quad\quad\quad\quad\left(\hat E\left[e^{\mathbf{i}\xi W_2}\frac{1}{h_{2,1}^{1+\lambda_{2,1}}h_{2,2}^{1+\lambda_{2,2}}}G_Y\left( \frac{y-Y}{h_{2,1}}\right)G^{(\lambda_{2,1}, \lambda_{2,2})}_X\left(\frac{x-X}{h_{2,2}}\right)\right]\right.\\
    &\quad\quad\quad\quad\quad\quad\quad\quad\left.-E\left[e^{\mathbf{i}\xi W_2}\frac{1}{h_{2,1}^{1+\lambda_{2,1}}h_{2,2}^{1+\lambda_{2,2}}}G_Y\left( \frac{y-Y}{h_{2,1}}\right)G^{(\lambda_{2,1}, \lambda_{2,2})}_X\left(\frac{x-X}{h_{2,2}}\right)\right]\right)d\xi\Bigg|\\
    &\leq \int \left(\sup_{(y, x, w^*)\in\mathbb{S}_{(Y,X,W^*)}}\Psi_{1, \lambda_1, \lambda_{2,1}, \lambda_{2,2}}(\xi, y, x, w^*, h_1)\right)\left|\hat E[W_1e^{\mathbf{i}\xi W_2}]-E[W_1e^{\mathbf{i}\xi W_2}]\right|d\xi\\
    &\quad+\int \left(\sup_{(y, x, w^*)\in\mathbb{S}_{(Y,X,W^*)}}\Psi_{2,\lambda_1, \lambda_{2,1}, \lambda_{2,2}}(\xi, y, x, w^*, h_1)\right)\left|\hat E[e^{\mathbf{i}\xi W_2}]-E[e^{\mathbf{i}\xi W_2}]\right|d\xi\notag\\
    &\quad+\int \left(\sup_{(y, x, w^*)\in\mathbb{S}_{(Y,X,W^*)}}(h_{2,1}^{-1})^{1+\lambda_{2,1}}(h_{2,2}^{-1})^{1+\lambda_{2,2}}\Psi_{3,\lambda_1, \lambda_{2,1}, \lambda_{2,2}}(\xi, y, x, w^*, h_1)\right)\times\\   &\quad\sup_{x\in\mathbb{S}_X}\left|\hat E\left[e^{\mathbf{i}\xi W_2}G_Y\left(\frac{y-Y}{h_{2,1}}\right)G^{(\lambda_{2,1}, \lambda_{2,2})}_X\left(\frac{x-X}{h_{2,2}}\right)\right]-E\left[e^{\mathbf{i}\xi W_2}G_Y\left(\frac{y-Y}{h_{2,1}}\right)G^{(\lambda_{2,1}, \lambda_{2,2})}_X\left(\frac{x-X}{h_{2,2}}\right)\right]\right|d\xi\\
    &=\int \Psi_{1, \lambda_1, \lambda_{2,1}, \lambda_{2,2}}^+(\xi, h_1)\left|\hat E[W_1e^{\mathbf{i}\xi W_2}]-E[W_1e^{\mathbf{i}\xi W_2}]\right|d\xi\\
    &\quad+\int \Psi_{2,\lambda_1, \lambda_{2,1}, \lambda_{2,2}}^+(\xi, h_1)\left|\hat E[e^{\mathbf{i}\xi W_2}]-E[e^{\mathbf{i}\xi W_2}]\right|d\xi\notag\\
    &\quad+\int (h_{2,1}^{-1})^{1+\lambda_{2,1}}(h_{2,2}^{-1})^{1+\lambda_{2,2}}\Psi_{3,\lambda_1, \lambda_{2,1}, \lambda_{2,2}}^+(\xi, h_1)\times \notag\\  
    &\quad\sup_{x\in\mathbb{S}_X}\left|\hat E\left[e^{\mathbf{i}\xi W_2}G_Y\left(\frac{y-Y}{h_{2,1}}\right)G^{(\lambda_{2,1}, \lambda_{2,2})}_X\left(\frac{x-X}{h_{2,2}}\right)\right]-E\left[e^{\mathbf{i}\xi W_2}G_Y\left(\frac{y-Y}{h_{2,1}}\right)G^{(\lambda_{2,1}, \lambda_{2,2})}_X\left(\frac{x-X}{h_{2,2}}\right)\right]\right| d\xi
\end{align*}
\end{scriptsize}
where $\Psi_{1, \lambda_1, \lambda_{2,1}, \lambda_{2,2}}^+(\xi, h_1)$, $\Psi_{2,\lambda_1, \lambda_{2,1}, \lambda_{2,2}}^+(\xi, h_1)$ and $\Psi_{3,\lambda_1, \lambda_{2,1}, \lambda_{2,2}}^+(\xi, h_1)$ are defined in Lemma \ref{orderPsi}. The integrals are finite for two reasons. First, with probability approaching 1, I have $\left|\hat E[W_1e^{\mathbf{i}\xi W_2}]-E[W_1e^{\mathbf{i}\xi W_2}]\right| \leq 2 E[\left|W_1\right|]+\epsilon\leq \infty$, for some positive real number $\epsilon$. This follows from Assumption \ref{finite variance}, the fact that 
\begin{footnotesize}
\begin{align*}
    \left|\hat E[W_1e^{\mathbf{i}\xi W_2}]-E[W_1e^{\mathbf{i}\xi W_2}]\right|&\leq \left|\hat E[W_1e^{\mathbf{i}\xi W_2}]\right|+\left|E[W_1e^{\mathbf{i}\xi W_2}]\right|\\
    &\leq \hat E[\left|W_1e^{\mathbf{i}\xi W_2}\right|]+E[\left|W_1e^{\mathbf{i}\xi W_2}\right|]\\
    &\leq \hat E[\left|W_1\right|]+E[\left|W_1\right|],
\end{align*}
\end{footnotesize}
and that $\hat E[\left|W_1\right|]\xrightarrow[]{p}E[\left|W_1\right|]$. Also, it's easy to show that $\left|\hat E[e^{\mathbf{i}\xi W_2}]-E[e^{\mathbf{i}\xi W_2}]\right|\leq 2+\epsilon<\infty$ and that $\sup_{x\in\mathbb{S}_X}\left|\hat E\left[e^{\mathbf{i}\xi W_2}G_Y\left(\frac{y-Y}{h_{2,1}}\right)G^{(\lambda_{2,1}, \lambda_{2,2})}_X\left(\frac{x-X}{h_{2,2}}\right)\right]-E\left[e^{\mathbf{i}\xi W_2}G_Y\left(\frac{y-Y}{h_{2,1}}\right)G^{(\lambda_{2,1}, \lambda_{2,2})}_X\left(\frac{x-X}{h_{2,2}}\right)\right]\right|\preceq 1$. Second, it follows by Lemma \ref{orderPsi} that $\int \Psi^+_{j,\lambda_1, \lambda_{2,1}, \lambda_{2,2}}(\xi,h_1)d\xi<\infty$ for $j=1,2$ and that  $\int (h_{2,1}^{-1})^{1+\lambda_{2,1}}(h_{2,2}^{-1})^{1+\lambda_{2,2}}\Psi^+_{3,\lambda_1, \lambda_{2,1}, \lambda_{2,2}}(\xi,h_1)d\xi<\infty$.\\
\indent By Assumption \ref{finite variance},
\begin{footnotesize}
\begin{align*}
    E\left[\left(\hat E[W_1e^{\mathbf{i}\xi W_2}]-E[W_1e^{\mathbf{i}\xi W_2}]\right)^2\right]&\leq \frac{1}{n} E\left[\left( W_1e^{\mathbf{i}\xi W_2}-E[W_1e^{\mathbf{i}\xi W_2}]\right)^2\right]\leq \frac{1}{n} E\left[\left( W_1e^{\mathbf{i}\xi W_2}\right)^2\right]\\
    &\leq \frac{1}{n} E\left[ W_1^2\right]=O\left(\frac{1}{n}\right),
\end{align*}
\end{footnotesize}
which implies $\left|\hat E[W_1e^{\mathbf{i}\xi W_2}]-E[W_1e^{\mathbf{i}\xi W_2}]\right|=O_p(n^{-1/2})$. Similarly, I also have $\left|\hat E[e^{\mathbf{i}\xi W_2}]-E[e^{\mathbf{i}\xi W_2}]\right|=O_p(n^{-1/2})$.\\
\indent By the conclusions above and Lemma \ref{uniform rate hat phi f}, I have 
\begin{scriptsize}
\begin{align*}
    &\quad \sup_{(y, x, w^*)\in\mathbb{S}_{(Y,X,W^*)}}\left|L_{\lambda_1, \lambda_{2,1}, \lambda_{2,2}}(y,x,w^*,h)\right|\\
    &\leq \int \Psi_{1, \lambda_1, \lambda_{2,1}, \lambda_{2,2}}^+(\xi, h_1)\left|\hat E[W_1e^{\mathbf{i}\xi W_2}]-E[W_1e^{\mathbf{i}\xi W_2}]\right|d\xi\\
    &\quad+\int \Psi_{2,\lambda_1, \lambda_{2,1}, \lambda_{2,2}}^+(\xi, h_1)  \left|\hat E[e^{\mathbf{i}\xi W_2}]-E[e^{\mathbf{i}\xi W_2}]\right|d\xi\notag\\
    &\quad+\int \Psi_{3,\lambda_1, \lambda_{2,1}, \lambda_{2,2}}^+(\xi, h_1)\times\\  
    &\quad\quad\quad \sup_{(y,x)\in\mathbb{S}_{(Y,X)}}\left|\hat E\left[e^{\mathbf{i}\xi W_2}\frac{1}{h_{2,1}^{1+\lambda_{2,1}}h_{2,2}^{1+\lambda_{2,2}}}G_Y^{(\lambda_{2,1})}\left(\frac{y-Y}{h_{2,1}}\right)G_X^{(\lambda_{2,2})}\left(\frac{x-X}{h_{2,2}}\right)\right]\right.\\
    &\quad\quad\quad\quad\quad\quad\quad\quad\quad\quad\left.-E\left[e^{\mathbf{i}\xi W_2}\frac{1}{h_{2,1}^{1+\lambda_{2,1}}h_{2,2}^{1+\lambda_{2,2}}}G_Y^{(\lambda_{2,1})}\left(\frac{y-Y}{h_{2,1}}\right)G_X^{(\lambda_{2,2})}\left(\frac{x-X}{h_{2,2}}\right)\right]\right|d\xi\\
    &\preceq n^{-1/2} \times\\
    &\quad\left(\int \Psi_{1, \lambda_1, \lambda_{2,1}, \lambda_{2,2}}^+(\xi, h_1)d\xi+\int \Psi_{2,\lambda_1, \lambda_{2,1}, \lambda_{2,2}}^+(\xi, h_1)d\xi+\int (h_{2,1}^{-1})^{1+\lambda_{2,1}}(h_{2,2}^{-1})^{1+\lambda_{2,2}}\Psi_{3,\lambda_1, \lambda_{2,1}, \lambda_{2,2}}^+(\xi, h_1) d\xi\right)\\
    &=n^{-1/2} \Psi^+_{\lambda_1, \lambda_{2,1}, \lambda_{2,2}}(h),
\end{align*}
\end{scriptsize}
where 
\begin{scriptsize}
\begin{align*}
    &\quad\Psi^+_{\lambda_1, \lambda_{2,1}, \lambda_{2,2}}(h)\\
    &=O\left(\max\left\{(h_1^{-1})^{1+\gamma_*},(h_{2,1}^{-1})^{1+\lambda_{2,1}}(h_{2,2}^{-1})^{1+\lambda_{2,2}}\right\}(h_1^{-1})^{1-\gamma_2+\gamma_{\phi}+\lambda_1}\exp\left(\left(\alpha_{\phi}\mathbf{1}\{\beta_{\phi}=\beta_2\}-\alpha_2\right)\left(h_1^{-1}\right)^{\beta_2}\right)\right),
\end{align*}
\end{scriptsize}
as shown in Lemma \ref{orderPsi}. \\
\indent (ii) Next I show asymptotic normality.  I apply Lemma \ref{CLTgeneral} to 
\begin{footnotesize}
\begin{align*}
    &l_{\lambda_1, \lambda_{2,1}, \lambda_{2,2}}(y, x, w^*,h_n;Y, X, W_1, W_2)\\
    &=\int \Psi_{1, \lambda_1, \lambda_{2,1}, \lambda_{2,2}}(\xi, y, x, w^*, h_1)\left(W_1e^{\mathbf{i}\xi W_2}-E[W_1e^{\mathbf{i}\xi W_2}]\right)d\xi\notag\\
    &\quad+\int \Psi_{2,\lambda_1, \lambda_{2,1}, \lambda_{2,2}}(\xi, y, x, w^*, h_1)\left(e^{\mathbf{i}\xi W_2}-E[e^{\mathbf{i}\xi W_2}]\right)d\xi\notag\\
    &\quad+\int \Psi_{3,\lambda_1, \lambda_{2,1}, \lambda_{2,2}}(\xi, y, x, w^*, h_1)\times\\
    &\quad\quad\quad\quad\left(\frac{1}{h_{2,1}^{1+\lambda_{2,1}}h_{2,2}^{1+\lambda_{2,2}}}G_Y\left( \frac{y-Y}{h_{2,1}}\right)G^{(\lambda_{2,1}, \lambda_{2,2})}_X\left(\frac{x-X}{h_{2,2}}\right)e^{\mathbf{i}\xi W_2}\right.\\
    &\quad\quad\quad\quad\quad\quad\left.-E\left[\frac{1}{h_{2,1}^{1+\lambda_{2,1}}h_{2,2}^{1+\lambda_{2,2}}}G_Y\left( \frac{y-Y}{h_{2,1}}\right)G^{(\lambda_{2,1}, \lambda_{2,2})}_X\left(\frac{x-X}{h_{2,2}}\right)e^{\mathbf{i}\xi W_2}\right]\right)d\xi,
\end{align*}
\end{footnotesize}
where $A=Y,X,W_1,W_2$ Previous argument ensures that for fixed $h$, $\Omega_{\lambda_1, \lambda_{2,1}, \lambda_{2,2}}(y, x, w^*,h)<\infty$. The desired conclusion follows.
\end{proof}

\begin{lemma}\label{Slemma6}(\cite{schennach2004estimation} Lemma 6) Let $A$ and $W_2$ be random variables satisfying $E[|A|^2]<\infty$ and $E[|A||W_2|]<\infty$ and let $(A_j,W_{2,j})_{j=1,...,n}$ be a corresponding IID sample. Then for any $u,U\geq 0$ and $\epsilon>0$,
\begin{align*}
    \sup_{\zeta\in[-Un^u, Un^u]}\left|\hat E[A\exp\left(\mathbf{i}\zeta W_2\right)]-E[A\exp\left(\mathbf{i}\zeta W_2\right)]\right|=o_p\left(n^{-1/2+\epsilon}\right).
\end{align*}
\end{lemma}

\begin{proof}[\proofname\ of Theorem \ref{Remainder}]
Plug (\ref{hatphiXW2/hatphiW2}) (\ref{delta hat q W1}) and (\ref{expQ}) into
\begin{footnotesize}
\begin{align*}
    &\quad\hat g_{\lambda_1,\lambda_{2,1}, \lambda_{2,2}}(y, x, w^*, h)-g_{\lambda_1,\lambda_{2,1}, \lambda_{2,2}}(y, x, w^*, h_1)\\
    &=\frac{1}{2\pi}\int \left(-\mathbf{i}t\right)^{\lambda_2}e^{-\mathbf itw^*}\phi_{K}(h_1t)\\
    &\quad\quad\times\left[\frac{\hat\phi_{f^{(\lambda_{2,1},\lambda_{2,2})}_{Y,X,W_2}(y,x,\cdot)}(t)}{\hat\phi_{W_2}(t)}\exp\left(\int_0^{t}\frac{\mathbf{i}\hat\theta(\xi)}{\hat \phi_{W_2}(\xi)}d\xi\right)-\frac{\phi_{f^{(\lambda_{2,1},\lambda_{2,2})}_{Y,X,W_2}(y,x,\cdot)}(t)}{\phi_{W_2}(t)}\exp\left(\int_0^{t}\frac{\mathbf{i}\theta(\xi)}{ \phi_{W_2}(\xi)}d\xi\right)\right]dt.
\end{align*}
\end{footnotesize}
and remove terms linear in $\delta \hat{\theta}_{}(t)$, $\delta \hat{\phi}_{W_{2}}(\xi)$, and $\delta \hat\phi_{f^{(\lambda_{2,1},\lambda_{2,2})}_{Y,X,W^*}(y,x,\cdot)}$. For notation simplicity, I write $h$ instead of $h_n$ here. I can then find that $\left|\hat g_{\lambda_1,\lambda_{2,1}, \lambda_{2,2}}(y, x, w^*, h)-g_{\lambda_1,\lambda_{2,1}, \lambda_{2,2}}(y, x, w^*, h_1)\right|\preceq \frac{1}{2\pi}\sum_{l=1}^7R_{l,\lambda_1,\lambda_{2,1},\lambda_{2,2}}$, where
\begin{footnotesize}
\begin{align*}
    R_{1, \lambda_1,\lambda_{2,1},\lambda_{2,2}}&= \int_0^\infty|t|^{\lambda_1} \left|\phi_{K}(h_1t)\right|\left|\delta_{1} \hat{q}_{\lambda_2}(t)\right|\left|\phi_{W^*}(t)\right|\left(\int_0^t \left|\delta_1\hat q_{W_1}(\xi)\right|d\xi\right) dt\\
    R_{2,\lambda_1,\lambda_{2,1},\lambda_{2,2}}&= \int_0^\infty|t|^{\lambda_1} \left|\phi_{K}(h_1t)\right|\left|\delta_{2} \hat{q}_{\lambda_2}(t)\right|\left|\phi_{W^*}(t)\right|dt\\
    R_{3,\lambda_1,\lambda_{2,1},\lambda_{2,2}}  &= \int_0^\infty|t|^{\lambda_1} \left|\phi_{K}(h_1t)\right|\left|\delta_{2} \hat{q}_{\lambda_2}(t)\right|\left|\phi_{W^*}(t)\right|\left(\int_0^t \left|\delta_1\hat q_{W_1}(\xi)\right|d\xi\right)  dt\\
    R_{4,\lambda_1,\lambda_{2,1},\lambda_{2,2}} &= \int_0^\infty|t|^{\lambda_1} \left|\phi_{K}(h_1t)\right|\left|q_{\lambda_2}( t)\right|\left|\phi_{W^*}(t)\right|\left(\int_0^t \left|\delta_2\hat q_{W_1}(\xi)\right|d\xi\right)  dt\\
    R_{5,\lambda_1,\lambda_{2,1},\lambda_{2,2}}&= \int_0^\infty|t|^{\lambda_1} \left|\phi_{K}(h_1t)\right|\left|\delta \hat q_{\lambda_2}(t)\right|\left|\phi_{W^*}(t)\right|\left(\int_0^t \left|\delta_2\hat q_{W_1}(\xi)\right|d\xi\right)  dt\\
    R_{6,\lambda_1,\lambda_{2,1},\lambda_{2,2}}&= \int_0^\infty|t|^{\lambda_1} \left|\phi_{K}(h_1t)\right|\left|q_{\lambda_2}(t)\right|\left|\phi_{W^*}(t)\right|\frac{1}{2}\exp\left(\left|\delta \bar Q(t)\right|\right)\left(\int_0^t \left|\delta\hat q_{W_1}(\xi)\right|d\xi\right)^2  dt\\
    R_{7,\lambda_1,\lambda_{2,1},\lambda_{2,2}}&= \int_0^\infty|t|^{\lambda_1} \left|\phi_{K}(h_1t)\right|\left|\delta \hat q_{\lambda_2}(t)\right|\left|\phi_{W^*}(t)\right|\frac{1}{2}\exp\left(\left|\delta \bar Q(t)\right|\right)\left(\int_0^t \left|\delta\hat q_{W_1}(\xi)\right|d\xi\right)^2 dt,
\end{align*}    
\end{footnotesize}
where $q_{\lambda_2}(t)$, $\delta\hat q_{\lambda_2}(t)$, $\delta_1\hat q_{\lambda_2}(t)$, $\delta_2\hat q_{\lambda_2}(t)$, $q_{W_1}(t)$, $\delta \hat q_{W_1}(t)$, $\delta_1 \hat q_{W_1}(t)$, $\delta_2 \hat q_{W_1}(t)$, and $\delta\bar Q(t)$ were defined in the proof of Lemma \ref{decomposition}.\\
\indent Lemma \ref{Slemma6} gives that for any $\epsilon>0$,
\begin{scriptsize}
\begin{align*}
    &\quad\sup_{(y,x)\in\mathbb{S}_(Y,X)}\sup_{\ \xi\in[-h_{1n}^{-1},h_{1n}^{-1}]} \left|\hat E\left[e^{\mathbf{i}\xi W_2 }\frac{1}{h_{2,1}^{1+\lambda_{2,1}}h_{2,2}^{1+\lambda_{2,2}}}G^{(\lambda_{2,1})}_Y\left(\frac{y-Y}{h_{2n,1}}\right)G^{(\lambda_{2,2})}_X\left(\frac{x-X}{h_{2n,2}}\right)\right]\right.\\
    &\quad\quad\quad\quad\quad\quad\quad\quad\quad\quad\left.-E\left[e^{\mathbf{i}\xi W_2 }\frac{1}{h_{2,1}^{1+\lambda_{2,1}}h_{2,2}^{1+\lambda_{2,2}}}G^{(\lambda_{2,1})}_Y\left(\frac{y-Y}{h_{2n,1}}\right)G^{(\lambda_{2,2})}_X\left(\frac{x-X}{h_{2n,2}}\right)\right]\right|\\
    &=\left(h_{2,1}^{-1}\right)^{1+\lambda_{2,1}}(h_{2,2}^{-1})^{1+\lambda_{2,2}}\sup_{x\in\mathbb{S}_X} \left|G_Y^{(\lambda_{2,1})}\left(\frac{y-Y}{h_{2n,1}}\right)G^{(\lambda_{2,2})}_X\left(\frac{x-X}{h_{2n,2}}\right)\right| \sup_{\ \xi\in[-h_{1n}^{-1},h_{1n}^{-1}]} \left|\hat E\left[e^{\mathbf{i}\xi W_2 }\right]-E\left[e^{\mathbf{i}\xi W_2 }\right]\right|\\
    &=o_p\left(\left(h_{2,1}^{-1}\right)^{1+\lambda_{2,1}}(h_{2,2}^{-1})^{1+\lambda_{2,2}}n^{-1/2+\epsilon}\right). 
\end{align*}
\end{scriptsize}
By (\ref{Edeltahatphi}) I have that 
\begin{scriptsize}
\begin{align*}
    &\quad\sup_{(y,x)\in\mathbb{S}_(Y,X)}\sup_{\ \xi\in[-h_{1n}^{-1},h_{1n}^{-1}]}\left|E\left[e^{\mathbf{i}\xi W_2 }\frac{1}{h_{2,1}^{1+\lambda_{2,1}}h_{2,2}^{1+\lambda_{2,2}}}G^{(\lambda_{2,1})}_Y\left(\frac{y-Y}{h_{2n,1}}\right)G^{(\lambda_{2,2})}_X\left(\frac{x-X}{h_{2n,2}}\right)\right]-\phi_{f^{(\lambda_2)}_{Y,X,W_2}(y,x,\cdot)}(\xi)\right|\\
    &\preceq \sup_{\ \xi\in[-h_{1n}^{-1},h_{1n}^{-1}]}\left|\phi_{W_2}(\xi)\right|O\left(\left(\frac{\bar\xi_{G_Y}}{h_{2,1}}\right)^{1+\gamma_{f_1}}\left(\frac{\bar\xi_{G_X}}{h_{2,2}}\right)^{1+\gamma_{f_2}}\exp\left(\alpha_{f_1}\left(\frac{\bar\xi_{G_Y}}{h_{2,1}}\right)^{\beta_{f_1}}+\alpha_{f_2}\left(\frac{\bar\xi_{G_X}}{h_{2,2}}\right)^{\beta_{f_2}}\right)\right)\\
    &\preceq O\left(\left(\frac{\bar\xi_{G_Y}}{h_{2,1}}\right)^{1+\gamma_{f_1}}\left(\frac{\bar\xi_{G_X}}{h_{2,2}}\right)^{1+\gamma_{f_2}}\exp\left(\alpha_{f_1}\left(\frac{\bar\xi_{G_Y}}{h_{2,1}}\right)^{\beta_{f_1}}+\alpha_{f_2}\left(\frac{\bar\xi_{G_X}}{h_{2,2}}\right)^{\beta_{f_2}}\right)\right)\\
    &=o\left(\left(h_{2,1}^{-1}\right)^{1+\lambda_{2,1}}(h_{2,2}^{-1})^{1+\lambda_{2,2}}n^{-1/2+\epsilon}\right),
\end{align*}
\end{scriptsize}
where the last inequality holds under Assumption \ref{bandwidth2}. Combining results from above, I have that for any $\epsilon>0$,
\begin{footnotesize}
\begin{align*}
    &\quad\sup_{(y,x)\in\mathbb{S}_(Y,X)}\sup_{\ \xi\in[-h_{1n}^{-1},h_{1n}^{-1}]} \left|\hat\phi_{f^{(\lambda_2)}_{Y,X,W_2}(y,x,\cdot)}(\xi)-\phi_{f^{(\lambda_2)}_{Y,X,W_2}(y,x,\cdot)}(\xi)\right| \\
    &=\sup_{(y,x)\in\mathbb{S}_(Y,X)}\sup_{\ \xi\in[-h_{1n}^{-1},h_{1n}^{-1}]} \left|\hat E\left[e^{\mathbf{i}\xi W_2 }\frac{1}{h_{2,1}^{1+\lambda_{2,1}}h_{2,2}^{1+\lambda_{2,2}}}G^{(\lambda_{2,1})}_Y\left(\frac{y-Y}{h_{2n,1}}\right)G^{(\lambda_{2,2})}_X\left(\frac{x-X}{h_{2n,2}}\right)\right]\right.\\ &\quad\quad\quad\quad\quad\quad\quad\quad\quad\quad\quad\quad\quad\quad\left.-E\left[e^{\mathbf{i}\xi W_2 }\frac{1}{h_{2,1}^{1+\lambda_{2,1}}h_{2,2}^{1+\lambda_{2,2}}}G^{(\lambda_{2,1})}_Y\left(\frac{y-Y}{h_{2n,1}}\right)G^{(\lambda_{2,2})}_X\left(\frac{x-X}{h_{2n,2}}\right)\right]\right|\\
    &\quad+\sup_{(y,x)\in\mathbb{S}_(Y,X)}\sup_{\ \xi\in[-h_{1n}^{-1},h_{1n}^{-1}]} \left|E\left[e^{\mathbf{i}\xi W_2 }\frac{1}{h_{2,1}^{1+\lambda_{2,1}}h_{2,2}^{1+\lambda_{2,2}}}G^{(\lambda_{2,1})}_Y\left(\frac{y-Y}{h_{2n,1}}\right)G^{(\lambda_{2,2})}_X\left(\frac{x-X}{h_{2n,2}}\right)\right]\right.\\
    &\quad\quad\quad\quad\quad\quad\quad\quad\quad\quad\quad\quad\quad\quad \left.-\phi_{f^{(\lambda_2)}_{Y,X,W_2}(y,x,\cdot)}(\xi)\right|\\
    &=o_p\left(\left(h_{2,1}^{-1}\right)^{1+\lambda_{2,1}}(h_{2,2}^{-1})^{1+\lambda_{2,2}}n^{-1/2+\epsilon}\right).
\end{align*}
\end{footnotesize}
I define $\Upsilon(h_1)$ and $\hat\Phi_n$ as
\begin{footnotesize}
\begin{align*}
    \Upsilon(h_{1n}) &\equiv (1+h_{1n}^{-1})\left(\sup_{\xi\in[-h^{-1},h^{-1}]}\frac{|\phi_{W^*}^{'}(\xi)|}{|\phi_{W^*}(\xi)|}\right)\left(\sup_{\xi\in[-h_{1n}^{-1},h_{1n}^{-1}]} \left|\phi_{W_2}(\xi)\right|^{-1}\right)\\
    &=O\left((1+h_{1n}^{-1})^{1+\gamma_{*}-\gamma_2}\exp\left(-\alpha_2\left(h_{1n}^{-1}\right)^{\beta_2}\right)\right)\\
    \hat\Phi_{n,\lambda_2}&\equiv\max\Bigg\{\sup_{\xi\in[-h_{1n}^{-1},h_{1n}^{-1}]} \left|\hat\theta(\xi)-\theta(\xi)\right|, \sup_{\xi\in[-h_{1n}^{-1},h_{1n}^{-1}]} \left|\hat\phi_{W_2}(\xi)-\phi_{W_2}(\xi)\right|,\\ &\quad\quad\quad\quad\quad\quad\sup_{(y,x)\in\mathbb{S}_(Y,X)}\sup_{\ \xi\in[-h_{1n}^{-1},h_{1n}^{-1}]} \left|\hat\phi_{f^{(\lambda_{2,1},\lambda_{2,2})}_{Y,X,W_2}(y,x,\cdot)}(\xi)-\phi_{f^{(\lambda_{2,1},\lambda_{2,2})}_{Y,X,W_2}(y,x,\cdot)}(\xi)\right|\Bigg\}\\
    &=o_p\left((h_{2,1}^{-1})^{1+\lambda_{2,1}}(h_{2,2}^{-1})^{1+\lambda_{2,2}}n^{-1/2+\epsilon}\right) 
\end{align*}
\end{footnotesize}
for any $\epsilon>0$. The latter order of magnitude follows from Lemma \ref{Slemma6} and Assumption \ref{bandwidth2} and \ref{finite W2 mmts}. Then $R_{1,\lambda_1, \lambda_{2,1}, \lambda_{2,2}}-R_{7,\lambda_1, \lambda_{2,1}, \lambda_{2,2}}$ can be bounded in terms of $\Psi^+_{\lambda_1,\lambda_{2,1}, \lambda_{2,2}}(h_{1n}), \Upsilon(h_{1n})$, and $\hat\Phi_{n,\lambda_2}$. Note that under Assumption \ref{bandwidth2},
\begin{footnotesize}
\begin{align*}
    &\quad\sup_{\xi\in[-h_{1n}^{-1},h_{1n}^{-1}]} \frac{\hat\Phi_{n,\lambda_2}}{\left|\phi_{W_2}(\xi)\right|}\\
    &\preceq \hat\Phi_{n,\lambda_2} \Upsilon(h_{1n})\\
    &=o_p\left((h_{2,1}^{-1})^{1+\lambda_{2,1}}(h_{2,2}^{-1})^{1+\lambda_{2,2}}n^{-1/2+\epsilon}\right)O\left((1+h_{1n}^{-1})^{1+\gamma_{*}-\gamma_2}\exp\left(-\alpha_2\left(h_{1n}^{-1}\right)^{\beta_2}\right)\right)\\
    &=o_p(1).
\end{align*}
\end{footnotesize}
Now I have
\begin{scriptsize}
\begin{align*}
    R_1&\leq 2\int_0^\infty |t|^{\lambda_1}\left|\phi_{K}(h_1t)\right|\left(\frac{1}{\left|\phi_{W_2}(t)\right|}+\frac{\left|\phi_{f^{(\lambda_{2,1},\lambda_{2,2})}_{Y,X,W_2}(y,x,\cdot)}(t)\right|}{\left|\phi_{W_2}( t)\right|^2}\right)\hat\Phi_{n,\lambda_2}\left|\phi_{W^*}(t)\right|\left(\int_0^t \left|\delta_1\hat q_{W_1}(\xi)\right|d\xi\right) dt\\
    &\preceq\Upsilon(h_1)\hat\Phi_{n,\lambda_2}\int_0^\infty|t|^{\lambda_1} \left|\phi_{K}(h_1t)\right|\left(1+\frac{\left|\phi_{f^{(\lambda_{2,1},\lambda_{2,2})}_{Y,X,W_2}(y,x,\cdot)}(t)\right|}{\left|\phi_{W_2}( t)\right|}\right)\left|\phi_{W^*}(t)\right|\left(\int_0^t \left|\delta_1\hat q_{W_1}(\xi)\right|d\xi\right) d\tau dt \\
    &= \Upsilon(h_1)\hat\Phi_{n,\lambda_2}\int_0^\infty\Bigg[\int_\xi^\infty |t|^{\lambda_1}\left|\phi_{K}(h_1t)\right|\left(1+\frac{\left|\phi_{f^{(\lambda_{2,1},\lambda_{2,2})}_{Y,X,W_2}(y,x,\cdot)}(t)\right|}{\left|\phi_{W_2}( t)\right|}\right)\left|\phi_{W^*}(t)\right| dt\Bigg] \left|\delta_1\hat q_{W_1}(\xi)\right|d\xi\\
    &=\Upsilon(h_1)\hat\Phi_{n,\lambda_2}\int_0^\infty\Bigg[\int_\xi^\infty |t|^{\lambda_1}\left|\phi_{K}(h_1t)\right|\left(\left|\phi_{W^*}(t)\right|+\left|\phi_{f^{(\lambda_{2,1},\lambda_{2,2})}_{Y,X,W^*}(y,x,\cdot)}(t)\right|\right) dt\Bigg] \left|\delta_1\hat q_{W_1}(\xi)\right|d\xi\\
    &\preceq\Upsilon(h_1)\hat\Phi_{n,\lambda_2}^2\int_0^\infty\Bigg[\int_\xi^\infty|t|^{\lambda_1}\left|\phi_{K}(h_1t)\right|\left(\left|\phi_{W^*}(t)\right|+\left|\phi_{f^{(\lambda_{2,1},\lambda_{2,2})}_{Y,X,W^*}(y,x,\cdot)}(t)\right|\right)  dt\Bigg]\\
    &\quad\quad\quad\quad\quad\quad\quad\quad\times\left(1+\frac{\left|\theta(\xi)\right|}{\left|\phi_{W_2}(\xi)\right|}\right)\frac{1}{\left|\phi_{W_2}(\xi)\right|}d\xi\\
    &\preceq \Upsilon(h_1)\hat\Phi_{n,\lambda_2}^2 \Psi^+_{\lambda_1,\lambda_{2,1}, \lambda_{2,2}}(h)\\
    &=o_p\left(h_{2,1}^{-2}(h_{2,2}^{-1})^{2+2\lambda_2}n^{-1+2\epsilon}(h_{1n}^{-1})^{1+\gamma_{*}-\gamma_2}\exp\left(-\alpha_2\left(h_{1n}^{-1}\right)^{\beta_2}\right)\right)\times\\
    &O\left(\max\left\{(h_1^{-1})^{1+\gamma_*},(h_{2,1}^{-1})^{1+\lambda_{2,1}}(h_{2,2}^{-1})^{1+\lambda_{2,2}}\right\}(h_1^{-1})^{1-\gamma_2+\gamma_{\phi}+\lambda_1}\exp\left(\left(\alpha_{\phi}\mathbf{1}\{\beta_{\phi}=\beta_2\}-\alpha_2\right)\left(h_1^{-1}\right)^{\beta_2}\right)\right).
\end{align*}
\end{scriptsize}
\indent If further Assumption \ref{bandwidth2} holds, I just need to show that
\begin{align*}
    o_p\left(h_{2,1}^{-2}(h_{2,2}^{-1})^{2+2\lambda_2}n^{-1/2+2\epsilon}(h_{1n}^{-1})^{1+\gamma_{*}-\gamma_2}\exp\left(-\alpha_2\left(h_{1n}^{-1}\right)^{\beta_2}\right)\right)=o_p(1).
\end{align*}
If $\beta_2\neq 0$, $h_{1n}^{-1}=O\left(\left(\ln n\right)^{1/\beta_2-\eta}\right)$,  $h_{2n,2}^{-1}=O\left(n^{(8+4\lambda_2)^{-1}-\eta}\right)$, so that
\begin{footnotesize}
\begin{align*}
    &\quad o_p\left(n^{-1/2+2\epsilon}h_{2,1}^{-2}(h_{2,2}^{-1})^{2+2\lambda_2}(h_{1n}^{-1})^{1+\gamma_{*}-\gamma_2}\exp\left(-\alpha_2\left(h_{1n}^{-1}\right)^{\beta_2}\right)\right)\\
    &=o_p\left(n^{-1/2+2\epsilon}n^{1/2-2\eta(2+\lambda_2)}(h_{1n}^{-1})^{1+\gamma_{*}-\gamma_2}\exp\left(-\alpha_2\left(h_{1n}^{-1}\right)^{\beta_2}\right)\right)\\
    &=o_p\left(n^{2\epsilon-2\eta(2+\lambda_2)}(\left(\ln n\right)^{1/\beta_2-\eta})^{1+\gamma_{*}-\gamma_2}\exp\left(-\alpha_2\left(\ln n\right)^{1-\eta\beta_2}\right)\right)\\
    &=o_p\left(\exp\left[-\alpha_2\left(\ln n\right)^{1-\eta\beta_2}+\left(2\epsilon-2\eta(2+\lambda_2)\right)\ln n+\left(1+\gamma_*-\gamma_2\right)\left(1/\beta_2-\eta\right)\ln\left(\ln n\right)\right]\right)\\
    &=o_p(1),
\end{align*}    
\end{footnotesize}
where the equality follows since $\left(\ln n\right)^{1-\eta\beta_2}$ and $\ln\left(\ln n\right)$ are dominated by $\ln n$ and by picking $\epsilon<\eta(2+\lambda_2)$.\\
\indent If $\beta_2=0$, $h_{1n}^{-1}=O\left(n^{(4+4\gamma_*-4\gamma_2)^{-1}-\eta}\right)$, $h_{2n,2}^{-1}=O\left(n^{(16+8\lambda_2)^{-1}-\eta}\right)$ and with $\epsilon<\eta$,
\begin{align*}
    &\quad o_p\left(n^{-1/2+2\epsilon}(h_{1n}^{-1})^{1+\gamma_{*}-\gamma_2}(h_{2,1}^{-2}(h_{2,2}^{-1})^{2+2\lambda_2}\right)\\
    &= o_p\left(n^{-1/2+2\epsilon}(n^{(4+4\gamma_*-4\gamma_2)^{-1}-\eta})^{1+\gamma_{*}-\gamma_2}n^{1/4-2\eta(2+\lambda_2)}\right)\\
    &=o_p\left(n^{2\epsilon-\eta(5+\gamma_*-\gamma_2+2\lambda_2)}\right)\\
    &=o_p(1).
\end{align*}
The remaining terms are similarly bounded
\begin{footnotesize}
\begin{align*}
    &\quad R_{2,\lambda_1, \lambda_{2,1}, \lambda_{2,2}}\\
    &\leq \int_0^\infty|t|^{\lambda_1} \left|\phi_{K}(h_1t)\right|\Bigg|\frac{\left|\phi_{f^{(\lambda_{2,1},\lambda_{2,2})}_{Y,X,W_2}(y,x,\cdot)}(t)\right|}{\left|\phi_{W_2}(t)\right|^2}\frac{1}{\left|\phi_{W_2}(t)\right|}\hat\Phi^2_{n,\lambda_2}\left|1+o_p(1)\right|^{-1}\\
    &\quad\quad+\frac{1}{|\phi_{W_2}(t)|^2}\hat\Phi^2_{n,\lambda_2}\left|1+o_p(1)\right|^{-1}\Bigg|\left|\phi_{W^*}(t)\right|  dt\\
    &\preceq \Upsilon(h_1)\hat\Phi^2_{n,\lambda_2}\left|1+o_p(1)\right|^{-1}\int_0^\infty |t|^{\lambda_1}\left|\phi_{K}(h_1t)\right|\frac{1}{\left|\phi_{W_2}(t)\right|}\left|\frac{\left|\phi_{f^{(\lambda_{2,1},\lambda_{2,2})}_{Y,X,W_2}(y,x,\cdot)}(t)\right|}{\left|\phi_{W_2}(t)\right|}+1\right|\left|\phi_{W^*}(t)\right| dt\\
    &\preceq \Upsilon(h_1)\hat\Phi^2_{n,\lambda_2}\left|1+o_p(1)\right|^{-1}\Bigg(\int_0^\infty|t|^{\lambda_1} \left|\phi_{K}(h_1t)\right|\frac{\left|\phi_{f^{(\lambda_{2,1},\lambda_{2,2})}_{Y,X,W^*}(y,x,\cdot)}(t)\right|}{\left|\phi_{W_2}(t)\right|} dt\\
    &\quad\quad\quad\quad\quad\quad\quad\quad\quad\quad\quad\quad\quad\quad\quad\quad\quad\quad +\int_0^\infty|t|^{\lambda_1} \left|\phi_{K}(h_1t)\right|\frac{\left|\phi_{W^*}(t)\right|}{\left|\phi_{W_2}(t)\right|} dt\Bigg)\\
    &\preceq \Upsilon(h_1)\hat\Phi^2_{n,,\lambda_2}\Psi_{\lambda_1,\lambda_{2,1}, \lambda_{2,2}}^+(h)\left(1+o_p(1)\right)^{-1};\\
    &\quad R_{3,\lambda_1, \lambda_{2,1}, \lambda_{2,2}}\\
    &\preceq \Upsilon(h_1)\hat\Phi_{n,\lambda_2}R_{2,\lambda_1, \lambda_{2,1}, \lambda_{2,2}}=o_p(1)R_{2,\lambda_1, \lambda_{2,1}, \lambda_{2,2}};\\
    &\quad R_{4,\lambda_1, \lambda_{2,1}, \lambda_{2,2}} \\
    &= \int_0^\infty|t|^{\lambda_1} \left|\phi_{K}(h_1t)\right|\left|\phi_{f^{(\lambda_{2,1},\lambda_{2,2})}_{Y,X,W^*}(y,x,\cdot)}(t)\right|\left(\int_0^t \left|\delta_2\hat q_{W_1}(\xi)\right|d\xi\right)  dt\\
    &\preceq\Upsilon(h_1)\hat\Phi^2_{n,\lambda_2}\left|1+o_p(1)\right|^{-1}\int_0^\infty \frac{\int_\xi^\infty|t|^{\lambda_1}\left|\phi_{K}(h_1t)\right|\left|\phi_{f^{(\lambda_{2,1},\lambda_{2,2})}_{Y,X,W^*}(y,x,\cdot)}(t)\right|d t}{\left|\phi_{W_2}(\xi)\right|}d\xi\\
    &=\Upsilon(h_1)\hat\Phi^2_{n,\lambda_2}\Psi_{\lambda_1,\lambda_{2,1}, \lambda_{2,2}}^+(h)(1+o_p(1));\\
    &\quad R_{5,\lambda_1, \lambda_{2,1}, \lambda_{2,2}}\\
    &\leq \int_0^\infty|t|^{\lambda_1} \left|\phi_{K}(h_1t)\right|\left(\frac{1}{\left|\phi_{W_2}(t)\right|}+\frac{\left|\phi_{f^{(\lambda_{2,1},\lambda_{2,2})}_{Y,X,W_2}(y,x,\cdot)}(t)\right|}{\left|\phi_{W_2}(t)\right|^2}\right)\\&\quad\times\hat\Phi_{n,\lambda_2}\left|1+o_p(1)\right|^{-1}\left|\phi_{W^*}(t)\right|\left(\int_0^t \left|\delta_2\hat q_{W_1}(\xi)\right|d\xi\right) dt\\
    &\preceq \Upsilon(h_1)\hat\Phi_{n,\lambda_2}\left|1+o_p(1)\right|^{-1}\int_0^\infty \left|\phi_{K}(h_1t)\right|\left(\left|\phi_{W^*}(t)\right|+\left|\phi_{f^{(\lambda_{2,1},\lambda_{2,2})}_{Y,X,W^*}(y,x,\cdot)}(t)\right|\right)\\&\quad\times\left(\int_0^t \left|\delta_2\hat q_{W_1}(\xi)\right|d\xi\right) d\tau dt\\
    &\preceq \Upsilon(h_1)\hat\Phi_{n,\lambda_2}\left(1+o_p(1)\right)R_{4,\lambda_1, \lambda_{2,1}, \lambda_{2,2}}=o_p(1)R_{4,\lambda_1, \lambda_{2,1}, \lambda_{2,2}};\\
    &\quad R_{6,\lambda_1, \lambda_{2,1}, \lambda_{2,2}}\\
    &\leq \int_0^\infty|t|^{\lambda_1} \left|\phi_{K}(h_1t)\right|\left|\phi_{f^{(\lambda_{2,1},\lambda_{2,2})}_{Y,X,W^*}(y,x,\cdot)}(t)\right|\frac{1}{2}\exp\left(\int_0^t\left|\delta \hat q_{W_1}(\xi)\right|d\xi\right)\left(\int_0^t \left|\delta\hat q_{W_1}(\xi)\right|d\xi\right)^2 dt\\
    &\leq \frac{1}{2}\exp\left(o_p(1)\right)\int_0^\infty |t|^{\lambda_1}\left|\phi_{K}(h_1 t)\right|\left|\phi_{f^{(\lambda_{2,1},\lambda_{2,2})}_{Y,X,W^*}(y,x,\cdot)}(t)\right|\left(\int_0^t \left|\delta\hat q_{W_1}(\xi)\right|d\xi\right)^2 dt\\
    &\preceq\frac{1}{2}\exp\left(o_p(1)\right) \Upsilon(h_1)\hat\Phi_{n,\lambda_2}^2\left|1+o_p(1)\right|^{-1}\int_0^\infty |t|^{\lambda_1}\left|\phi_{K}(h_1 t)\right|\left|\phi_{f^{(\lambda_{2,1},\lambda_{2,2})}_{Y,X,W^*}(y,x,\cdot)}(t)\right|\\
    &\quad\quad\quad\quad\quad\quad\quad\quad\quad\quad\quad\quad\quad\quad\quad\quad\times\left(\int_0^t \frac{1}{\left|\phi_{W_2}(\xi)\right|}+\frac{\left|\theta(\xi)\right|}{\left|\phi_{W_2}(\xi)\right|^2}d\xi\right)dt\\
    &\preceq\frac{1}{2}\exp\left(o_p(1)\right) \Upsilon(h_1)\hat\Phi_{n,\lambda_2}^2\left|1+o_p(1)\right|^{-1}\int_0^\infty \left(\int_{\xi}^\infty|t|^{\lambda_1}\left|\phi_{K}(h_1 t)\right|\left|\phi_{f^{(\lambda_{2,1},\lambda_{2,2})}_{Y,X,W^*}(y,x,\cdot)}(t)\right|dt\right)\\  &\quad\quad\quad\quad\quad\quad\quad\quad\quad\quad\quad\quad\quad\quad\quad\quad\times \left(\frac{1}{\left|\phi_{W_2}(\xi)\right|}+\frac{\left|\theta(\xi)\right|}{\left|\phi_{W_2}(\xi)\right|^2}\right)d\xi\\
    &\preceq O_p(1)\Upsilon(h_1)\hat\Phi_{n,\lambda_2}^2\Psi_{\lambda_1,\lambda_{2,1}, \lambda_{2,2}}^+(h)\\
    &\quad R_{7,\lambda_1, \lambda_{2,1}, \lambda_{2,2}}\\
    &\leq \int_0^\infty|t|^{\lambda_1} \left|\phi_{K}(h_1t)\right|\left(1+\frac{\left|\phi_{f^{(\lambda_{2,1},\lambda_{2,2})}_{Y,X,W_2}(y,x,\cdot)}(t)\right|}{\left|\phi_{W_2}(t)\right|}\right)\Upsilon(h_1)\hat\Phi_{n,\lambda_2}\left|1+o_p(1)\right|^{-1}\left|\phi_{W^*}(t)\right|\\
    &\quad\quad\times\frac{1}{2}\exp\left(\int_0^t\left|\delta \hat q_{W_1}(\xi)\right|d\xi\right)\left(\int_0^t \left|\delta\hat q_{W_1}(\xi)\right|d\xi\right)^2 dt\\
    &\preceq \Upsilon(h_1)\hat\Phi_{n,\lambda_2}\left|1+o_p(1)\right|^{-1}\int_0^\infty |t|^{\lambda_1}\left|\phi_{K}(h_1t)\right|\left(\left|\phi_{W^*}(t)\right|+\left|\phi_{f^{(\lambda_{2,1},\lambda_{2,2})}_{Y,X,W^*}(y,x,\cdot)}(t)\right|\right)\\
    &\quad\quad\times\frac{1}{2}\exp\left(\int_0^t\left|\delta \hat q_{W_1}(\xi)\right|d\xi\right)\left(\int_0^t \left|\delta\hat q_{W_1}(\xi)\right|d\xi\right)^2 dt\\
    &\preceq \Upsilon(h_1)\hat\Phi_{n,\lambda_2}\left|1+o_p(1)\right|^{-1} R_{6,\lambda_1, \lambda_{2,1}, \lambda_{2,2}} = o_p(1)R_{6,\lambda_1, \lambda_{2,1}, \lambda_{2,2}}.
\end{align*}
\end{footnotesize}
\end{proof}

\begin{proof}[\proofname\ of Theorem \ref{unif convergent rate}]
In Lemma \ref{file7}, let $Z=(X, W^*)$. By (\ref{pmpx}), for any value $(y, x,w^*)\in \mathbb{S}_{Y, X, W^*}$, and any value $\delta\in[0,1]$, I have that $\mathbf{WLAR}(x)=\tilde\Gamma(f)$ and that $\rho(y,x)=\tilde\Xi(f)$. (i) By Lemma \ref{file7},
\begin{align*}
    &\sup_{(y,x,w^*)\in \mathbb  S_{\tau}}\left|\widehat{\rho(y,x)}-\rho(y,x)\right|\\
    &\leq \frac{O_p\left(\epsilon_{n,0}\right)+O_p\left(\tilde\epsilon_{n,0,0,0}\right)}{\tau^2}+O_p\left(\epsilon_{n,0}\right)+O_p\left(\tilde\epsilon_{n,0,0,1}\right)+\frac{O_p\left(\epsilon_{n,1}\right)+O_p\left(\tilde\epsilon_{n,1,0,0}\right)}{\tau^2}\\
    &\leq O_p\left(\tilde\epsilon_{n,0,0,1}\right)+\frac{O_p\left(\epsilon_{n,1}\right)+O_p\left(\tilde\epsilon_{n,1,0,0}\right)}{\tau^2}.
\end{align*}
The last inequality follows since $\epsilon_{n,0}$ is of smaller order than $\epsilon_{n,1}$, $\tilde\epsilon_{n,0,0,0}$ is of smaller order than $\tilde\epsilon_{n,1,0,0}$. Choose $\tau_n$ such that $\tau_n>0, \tau_n\rightarrow 0$ as $n\rightarrow\infty$, and that $\frac{\epsilon_{n,1}}{\tau^2}\rightarrow 0$ and $\frac{\tilde\epsilon_{n,1,0,0}}{\tau^2}\rightarrow 0$, I can get the desired result.\\
(ii) A direct application of Lemma \ref{file8} gives the desired result. 
\end{proof}
\subsection{Appendix C: Asymptotic Normality}
For asymptotic normality, I need to place a lower bound on $\Omega_{\lambda_1,\lambda_{2,1}, \lambda_{2,2}}(y,x,w^*.h_n)$ relative to $B_{\lambda_1,\lambda_{2,1}, \lambda_{2,2}}(y,x,w^*.h_n)$ and $R_{\lambda_1,\lambda_{2,1}, \lambda_{2,2}}(y,x,w^*.h_n)$. The following assumption is stated at a high level. More primitive sufficient conditions can be derived using techniques of \cite{schennach2004nonparametric}. Combining this assumption with the results from Theorem \ref{biasrate},\ref{Omegaprop} and \ref{Remainder} yields a corollary establishing the asymptotic normality of $\hat g_{\lambda_1,\lambda_{2,1}, \lambda_{2,2}} (y,x,w^*,h_n)$.
\begin{assumption}\label{n^1/2Omega^-1/2BR}(High Level)
For given $\lambda_1, \lambda_{2,1}, \lambda_{2,2}\in\{0,1\}$ and given $j\in\{1,2\}$, $h_n\rightarrow 0$ at a rate such that for each $(y,x,w^*)\in\mathbb{S}_{(Y,X,W^*)}$ such that $\Omega_{\lambda_1,\lambda_{2,1}, \lambda_{2,2}}(y,x,w^*,h_n)>0$ for all $n$ sufficiently large, I have $n^{1/2}\left(\Omega_{\lambda_1,\lambda_{2,1}, \lambda_{2,2}}(y,x,w^*,h_n)\right)^{-1/2}\left|B_{\lambda_1,\lambda_{2,1}, \lambda_{2,2}}(y,x,w^*,h_n)\right|\xrightarrow{p}0$, and $n^{1/2}\left(\Omega_{\lambda_1,\lambda_{2,1}, \lambda_{2,2}}(y,x,w^*,h_n)\right)^{-1/2}\left|R_{\lambda_1,\lambda_{2,1}, \lambda_{2,2}}(y,x,w^*,h_n)\right|\xrightarrow{p}0$.
\end{assumption}

\begin{corollary}\label{cltg}
If the conditions of Theorem \ref{Omegaprop} and Assumption \ref{n^1/2Omega^-1/2BR} hold, then for each $(y,x,w^*)\in\mathbb{S}_{\tau}$ such that $\Omega_{\lambda_1,\lambda_{2,1}, \lambda_{2,2}}(y,x,w^*,h_n)>0$ for all $n$ sufficiently large I have
\begin{align*}
    &\quad n^{1/2}\left(\Omega_{\lambda_1,\lambda_{2,1}, \lambda_{2,2}} (y,x,w^*,h_n)\right)^{-1/2} \left(\hat g_{\lambda_1,\lambda_{2,1}, \lambda_{2,2}} (y,x,w^*,h_n)-g_{\lambda_1,\lambda_{2,1}, \lambda_{2,2}}(y,x,w^*)\right)\\
    &\xrightarrow{d}N(0,1).
\end{align*}
\end{corollary}
In the next few lemmas, I show the asymptotic normality of $\widehat{\rho(y,x)}$ under some high-level assumptions. The asymptotic normality of $\widehat{\mathbf{WLAR}(x)}$ can be shown following the same logic. I first state a Lemma in a similar fashion as Lemma \ref{decomposition}
\begin{lemma}\label{decomp pmpx}
For any value $(y,x,w^*)\in \mathbb{S}_{\tau}$, I have that
\begin{align*}
    \widehat{\rho(y,x)}-\rho(y,x)=BD\tilde\Xi_n+LD\tilde\Xi_n+RD\tilde\Xi_n+R\tilde\Xi_n,
\end{align*}
where $R\tilde\Xi_n$ corresponds to the nonlinear part of functional $\tilde\Xi$; $BD\tilde\Xi_n$ and $RD\tilde\Xi_n$ are the bias term contained in the linear part of functional $\tilde\Xi$; $LD\tilde\Xi_n\equiv\hat E\left[lD\tilde\Xi_n\right]$ is the linear term\footnote{Linear in terms of $\hat E\left[W_1e^{i\xi W_2}\right]$, $\hat E\left[e^{i\xi W_2}\right]$ and $\hat E\left[e^{i\xi W_2}\frac{1}{h_{2,1}^{1+\lambda_{2,1}}h_{2,2}^{1+\lambda_{2,2}}}G_Y\left(\frac{y-Y}{h_{2,1}}\right)G_X^{(\lambda_{2,2})}\left(\frac{x-X}{h_{2,2}}\right)\right]$} contained in the linear part of functional $\tilde\Xi$, and is equal to $\hat E\left[\sum_{k=1}^{11}\tilde s_k(y,x,w^*)lD\tilde\Xi_{n,k}\right]$ in which
\begin{align*}
    lD\tilde\Xi_{n,1} &=\int l_{0,0,0}(y,x,w^*,h_n)dy \\
    lD\tilde\Xi_{n,2}&=\int_{-\infty}^{y} l_{0,0,0}(y,x,w^*,h_n )dy\\
    lD\tilde\Xi_{n,3}&=\int l_{0,0,1}(y,x,w^*,h_n) dy\\
    lD\tilde\Xi_{n,4}&=\int_{-\infty}^{y} l_{0,0,1}(y,x,w^*,h_n) dy\\
    lD\tilde\Xi_{n,5}&=l_{0,0,0}(y,x,w^*,h_n)\\
    lD\tilde\Xi_{n,6}&=\int l_{1,0,0}(y,x,w^*,h_n) dy\\
    lD\tilde\Xi_{n,7}&=\int_{-\infty}^{y} l_{1,0,0}(y,x,w^*,h_n) dy\\
    lD\tilde\Xi_{n,8}&=\int \tilde l_{0,0}(s,w^*,h_n) ds\\
    lD\tilde\Xi_{n,9}&=\int_{-\infty}^x \tilde l_{0,0}(s,w^*,h_n) ds\\
    lD\tilde\Xi_{n,10}&=\int \tilde l_{1,0}(s,w^*,h_n) ds\\
    lD\tilde\Xi_{n,11}&=\int_{-\infty}^x \tilde l_{1,0}(s,w^*,h_n) ds.
\end{align*}
\begin{scriptsize}
\begin{align*}
    &\tilde s_1(y,x,w^*) =-\frac{F_{Y\mid X=x, W^*=w^*}(y)}{f_{X,W^*}(x,w^*)f_{Y,X,W^*}(y,x,w^*)}\left(\frac{\partial f_{X,W^*}(x,w^*)}{\partial x}+\frac{\partial f_{X,W^*}(x,w^*)}{\partial w^*}\frac{1}{\tilde\Psi_{(1)}(f)}\right)\\
    &\quad\quad\quad\quad\quad\quad-\frac{F_{X\mid W^*=w^*}(x)\frac{\partial f_{W^*}(w^*)}{\partial w^*}-\int_{-\infty}^x\frac{\partial f_{X,W^*}(s,w^*)}{\partial w^*}ds}{\tilde f^2(x,w^*)}\times\frac{\Psi_{(2)}(f)}{\tilde\Psi^2_{(1)}(f)}\\
    &\tilde s_2(y,x,w^*) =\frac{1}{f_{X,W^*}(x,w^*)f_{Y,X,W^*}(y,x,w^*)}\left(\frac{\partial f_{X,W^*}(x,w^*)}{\partial x}+\frac{\partial f_{X,W^*}(x,w^*)}{\partial w^*}\frac{1}{\tilde\Psi_{(1)}(f)}\right)\\
    &\tilde s_3(y,x,w^*) =\frac{F_{Y\mid X=x, W^*=w^*}(y)}{f_{Y,X,W^*}(y,x,w^*)}\\
    &\tilde s_4(y,x,w^*) =-\frac{1}{f_{Y,X,W^*}(y,x,w^*)}\\
    &\tilde s_5(y,x,w^*) =-\frac{1}{f^2(y,x,w^*)}\left[F_{Y\mid X=x, W^*=w^*}(y)\frac{\partial f_{X,W^*}(x,w^*)}{\partial x}-\int_{-\infty}^{y}\frac{\partial f_{Y,X,W^*}(y,x,w^*)}{\partial x}dy\right.\\
    &\quad\quad\quad\quad\quad\quad+\left.\left(F_{Y\mid X=x, W^*=w^*}(y)\frac{\partial f_{X,W^*}(x,w^*)}{\partial w^*}-\int_{-\infty}^{y}\frac{\partial f_{Y,X,W^*}(y,x,w^*)}{\partial w^*}dy\right)\frac{1}{\tilde \Psi_{(1)}(f)}\right]\\
    &\tilde s_6(y,x,w^*) =\frac{F_{Y\mid X=x, W^*=w^*}(y)}{f_{Y,X,W^*}(y,x,w^*)}\frac{1}{\tilde \Psi_{(1)}(f)}\\
    &\tilde s_7(y,x,w^*) =-\frac{1}{f_{Y,X,W^*}(y,x,w^*)}\frac{1}{\tilde \Psi_{(1)}(f)}\\
    &\tilde s_8(y,x,w^*) = -\frac{F_{X\mid W^*=w^*}(x)\frac{\partial f_{W^*}(w^*)}{\partial w^*}}{f_{W^*}(w^*)f_{X,W^*}(x,w^*)}\times \frac{\Psi_{(2)}(f)}{\tilde\Psi^2_{(1)}(f)}\\
    &\tilde s_9(y,x,w^*) = \frac{\frac{\partial f_{W^*}(w^*)}{\partial w^*}}{f_{W^*}(w^*)f_{X,W^*}(x,w^*)}\times \frac{\Psi_{(2)}(f)}{\tilde\Psi^2_{(1)}(f)}\\
    &\tilde s_{10}( y,x,w^*) = \frac{F_{X\mid W^*=w^*}(x)}{f_{X,W^*}(x,w^*)}\times \frac{\Psi_{(2)}(f)}{\tilde\Psi^2_{(1)}(f)}\\
    &\tilde s_{11}( y,x,w^*) =-\frac{1}{f_{X,W^*}(x,w^*)}\times\frac{\Psi_{(2)}(f)}{\tilde\Psi^2_{(1)}(f)},
\end{align*}

\end{scriptsize}
where $\alpha(f)\equiv F^{-1}_{Y\mid X=x,W^*=w^*}(\delta)$, $\Phi_{(2)}(f)\equiv \frac{\partial F_{Y|X=x, W^*=w^*}^{-1}(\delta)}{\partial w^*}$, \\
$\Psi_{(2)}(f)\equiv- \left.\frac{\partial F_{Y|X=x, W^*=w^*}(y)}{\partial w^*}\right/\frac{\partial F_{Y|X=x, W^*=w^*}(y)}{\partial y}$, and $\tilde\Psi_{(1)}(f)=- \left.\frac{\partial F_{X|W^*=w^*}(x)}{\partial w^*}\right/\frac{\partial F_{X|W^*=w^*}(x)}{\partial x}$.
\end{lemma}
\begin{proof}[\proofname\ of Lemma \ref{decomp pmpx}]
\indent By (\ref{pmpxbar}), for any value $(y, x,w^*)\in \mathbb{S}_{\tau}$, I have that
\begin{align*}
    \widehat{\rho(y,x)}-\rho(y,x)=D\tilde\Xi_n+R\tilde\Xi_n.
\end{align*}
By Lemma \ref{file7},  
\begin{align*}
    D\tilde\Xi_n = \sum_{k=1}^{11}\tilde s_k(x,w^*,\delta)D\tilde\Xi_{n,k},
\end{align*}
where $\tilde s_k(x,w^*,\delta)$ is as stated in the Lemma.
and
\begin{align*}
    &D\tilde\Xi_{n,1}=\int \left(\hat{g}_{0,0,0}(y,x,w^*,h_n)-g_{0,0,0}(y,x,w^*,h_n)\right) dy\\
    &D\tilde\Xi_{n,2}=\int_{-\infty}^{y} \left(\hat{g}_{0,0,0}(y,x,w^*,h_n)-g_{0,0,0}(y,x,w^*,h_n)\right) dy\\
    &D\tilde\Xi_{n,3}=\int \left(\hat{g}_{0,1,1}(y,x,w^*,h_n)-g_{0,1,1}(y,x,w^*,h_n)\right) dy\\
    &D\tilde\Xi_{n,4}=\int_{-\infty}^{y} \left(\hat{g}_{0,1,1}(y,x,w^*,h_n)-g_{0,1,1}(y,x,w^*,h_n)\right) dy\\
    &D\tilde\Xi_{n,5}=\hat{g}_{0,0,0}(y,x,w^*,h_n)-g_{0,0,0}(y,x,w^*,h_n)\\
    &D\tilde\Xi_{n,6}=\int \left(\hat{g}_{1,0,0}(y,x,w^*,h_n)-g_{1,0,0}(y,x,w^*,h_n)\right) dy\\
    &D\tilde\Xi_{n,7}=\int_{-\infty}^{y} \left(\hat{g}_{1,0,0}(y,x,w^*,h_n)-g_{1,0,0}(y,x,w^*,h_n)\right) dy\\
    &D\tilde\Xi_{n,8}=\int\left(\hat{\tilde g}_{0,0}(s,w^*,h_n)-\tilde g_{0,0}(s,w^*,h_n)\right) ds\\
    &D\tilde\Xi_{n,9}=\int_{-\infty}^x \left(\hat{\tilde g}_{0,0}(s,w^*,h_n)-\tilde g_{0,0}(s,w^*,h_n)\right) ds\\
    &D\tilde\Xi_{n,10}=\int\left(\hat{\tilde g}_{1,0}(s,w^*,h_n)-\tilde g_{1,0}(s,w^*,h_n)\right) ds\\
    &D\tilde\Xi_{n,11}=\int_{-\infty}^x \left(\hat{\tilde g}_{1,0}(s,w^*,h_n)-\tilde g_{1,0}(s,w^*,h_n)\right) ds.
\end{align*}
By Lemma \ref{decomposition}, a first order expansion of $D\tilde\Xi_n$ can be written as
\begin{align*}
    D\tilde\Xi_n &=BD\tilde\Xi_n+LD\tilde\Xi_n+RD\tilde\Xi_n\\&= BD\tilde\Xi_n+\sum_{k=1}^{11}\tilde s_k(x,w^*,\delta)LD\tilde\Xi_{n,k}+RD\tilde\Xi_n
\end{align*}
where 
\begin{footnotesize}
\begin{align*}
    LD\tilde\Xi_{n,1} &=\int L_{0,0,0}(y,x,w^*,h_n) dy= \int \hat E\left[l_{0,0,0}(y,x,w^*,h_n) \right]dy =  \hat E\left[\int l_{0,0,0}(y,x,w^*,h_n)dy \right]\\
    LD\tilde\Xi_{n,2}&=\int_{-\infty}^{y} L_{0,0,0}(y,x,w^*,h_n) dy=\int_{-\infty}^{y} \hat E\left[l_{0,0,0}(y,x,w^*,h_n)\right] dy\\
    &=\hat E\left[\int_{-\infty}^{y} l_{0,0,0}(y,x,w^*,h_n )dy\right]\\
    LD\tilde\Xi_{n,3}&=\int L_{0,0,1}(y,x,w^*,h_n) dy= \int \hat E\left[l_{0,0,1}(y,x,w^*,h_n)\right] dy\\
    &= \hat E\left[\int l_{0,0,1}(y,x,w^*,h_n) dy\right]\\
    LD\tilde\Xi_{n,4}&=\int_{-\infty}^{y} L_{0,0,1}(y,x,w^*,h_n) dy= \int_{-\infty}^{y} \hat E\left[l_{0,0,1}(y,x,w^*,h_n)\right] dy\\
    &= \hat E\left[\int_{-\infty}^{y} l_{0,0,1}(y,x,w^*,h_n) dy\right]\\
    LD\tilde\Xi_{n,5}&=L_{0,0,0}(y,x,w^*,h_n)=\hat E\left[l_{0,0,0}(y,x,w^*,h_n)\right]\\
    LD\tilde\Xi_{n,6}&=\int L_{1,0,0}(y,x,w^*,h_n) dy= \int \hat E\left[l_{1,0,0}(y,x,w^*,h_n)\right] dy\\
    &= \hat E\left[\int l_{1,0,0}(y,x,w^*,h_n) dy\right]\\
    LD\tilde\Xi_{n,7}&=\int_{-\infty}^{y} L_{1,0,0}(y,x,w^*,h_n) dy= \int_{-\infty}^{y} \hat E\left[l_{1,0,0}(y,x,w^*,h_n)\right] dy\\
    &= \hat E\left[\int_{-\infty}^{y} l_{1,0,0}(y,x,w^*,h_n) dy\right]\\
    LD\tilde\Xi_{n,8}&=\int\tilde L_{0,0}(s,w^*,h_n) ds=\int \hat E\left[\tilde l_{0,0}(s,w^*,h_n)\right] ds=\hat E\left[\int \tilde l_{0}(s,w^*,h_n) ds\right]\\
    LD\tilde\Xi_{n,9}&=\int_{-\infty}^x \tilde L_{0,0}(s,w^*,h_n) ds=\int_{-\infty}^x \hat E\left[\tilde l_{0,0}(s,w^*,h_n)\right] ds=\hat E\left[\int_{-\infty}^x \tilde l_{0,0}(s,w^*,h_n) ds\right]\\
    LD\tilde\Xi_{n,10}&=\int \tilde L_{1,0}(s,w^*,h_n)ds=\int \hat E\left[\tilde l_{1,0}(s,w^*,h_n)\right] ds=\hat E\left[\int \tilde l_{1,0}(s,w^*,h_n) ds\right]\\
    LD\tilde\Xi_{n,11}&=\int_{-\infty}^x \tilde L_{1,0}(s,w^*,h_n)ds=\int_{-\infty}^x \hat E\left[\tilde l_{1,0}(s,w^*,h_n)\right] ds=\hat E\left[\int_{-\infty}^x \tilde l_{1,0}(s,w^*,h_n) ds\right].
\end{align*}    
\end{footnotesize}
Then I can write
\begin{align*}
    LD\tilde\Xi_n &= \hat E\left[lD\tilde\Xi_n\right]\\
    &=\hat E\left[\sum_{k=1}^{11}\tilde s_k(x,w^*,\delta)lD\tilde\Xi_{n,k}\right],
\end{align*}
where $lD\tilde\Xi_n$ are as stated in the Lemma.
\end{proof}
\begin{lemma}\label{Omegaprop Xi}
Suppose the conditions of Lemma \ref{decomposition} hold. (i) Then for each $(y, x, w^*)\in\mathbb{S}_{\tau}$, $E[LD\tilde\Xi_n]=0$, and if Assumption \ref{finite variance} also holds, then $E[LD\tilde\Xi_n^2]=n^{-1}\tilde\Omega(y,x,w^*,h)$, where $\tilde\Omega(y,x,w^*,h)\equiv E\left[lD\tilde\Xi_n^2\right]<\infty$.\\
\indent (ii) If Assumption \ref{finite variance} also holds, and if for each $(y, x,w^*)\in\mathbb{S}_{\tau}$, $\tilde\Omega(y,x,w^*,h_n)>0$ for all $n$ sufficiently large, then for each $(y,x, w^*)\in\mathbb{S}_{\tau}$
\begin{align*}
    n^{1/2}\left(\tilde\Omega(y,x,w^*,h_n)\right)^{-1/2}LD\tilde\Xi_n\xrightarrow{d}N(0,1).
\end{align*}
\end{lemma}
\begin{proof}[\proofname\ of Lemma \ref{Omegaprop Xi}]
(i) The fact that $E\left[LD\tilde\Xi_n\right]=0$ follows directly from Theorem \ref{Omegaprop}. Next I show that for all $(y,x,w^*)\in\mathbb{S}_\tau$, $\tilde\Omega(y,x,w^*,h)<\infty$. It follows by Cauchy Schwartz that 
\begin{align*}
    \tilde\Omega(y,x,w^*,h_n)&\equiv E\left[lD\tilde\Xi_n^2\right]\\
    &\preceq \sum_{k=1}^{11}\tilde s_k^2(y,x,w^*)E\left[lD\tilde\Xi_{n,k}^2\right].
\end{align*}
Note that 
\begin{align*}
    E\left[lD\tilde\Xi_{n,1}^2\right]&=E\left[\left(\int l_{0,0,0}(y,x,w^*,h_n)dy\right)^2\right]\\
    &\leq E\left[\int \left(l_{0,0,0}(y,x,w^*,h_n)\right)^2dy\right]\\
    &= \int E\left[ \left(l_{0,0,0}(y,x,w^*,h_n)\right)^2\right]dy
\end{align*}
where the second line follows from Jensen's Inequality and the third line follows from Tonelli's Theorem. By Theorem \ref{Omegaprop}, $\int E\left[ \left(l_{0,0,0}(y,x,w^*,h)\right)^2\right]dy<\infty$ for each $h$. Conclusions for other values of $k\in\{1,2,3,...,11\}$ follows similarly. Then by Assumption \ref{cpctcontdiff}, I have that $\tilde\Omega(y,x,w^*,h)<\infty$ for all values $(y,x,w^*)\in\mathbb{S}_\tau$.\\
\indent (ii) To show asymptotic normality, I apply Lemma \ref{CLTgeneral} to $lD\tilde\Xi_n=\sum_{k=1}^{11}\tilde s_k(x,w^*,\delta)lD\tilde\Xi_{n,k}$. By previous argument, $\tilde\Omega(y,x,w^*,h)<\infty$ for all values $(y,x,w^*)\in\mathbb{S}_\tau$. I've assumed that for $n$ sufficiently large, $\tilde\Omega(y,x,w^*,h_n)>0$. The conclusion follows.
\end{proof}
Finally, I state the asymptotic normality result for the estimator of $\rho(y,x)$ under a high-level assumption:
\begin{assumption}\label{n^1/2Omega^-1/2RDXi}(High Level)
For given $\lambda_1, \lambda_{2,1}, \lambda_{2,2}\in\{0,1\}$ and given $j\in\{1,2\}$, $h_n\rightarrow 0$ at a rate such that for each $(y,x,w^*)\in \mathbb{S}_{\tau}$ such that $\tilde\Omega(y,x,w^*,h_n)>0$ for all $n$ sufficiently large, I have $n^{1/2}\left(\tilde\Omega(y,x,w^*,h_n)\right)^{-1/2}\left|BD\tilde\Xi_n\right|\xrightarrow{p}0$ and $n^{1/2}\left(\tilde\Omega(y,x,w^*,h_n)\right)^{-1/2}\left|RD\tilde\Xi_n\right|\xrightarrow{p}0$.
\end{assumption}

\begin{theorem}\label{clt}
If the conditions of Theorem \ref{Omegaprop} and Assumption \ref{n^1/2Omega^-1/2BR}, \ref{n^1/2Omega^-1/2RDXi} hold, then for each $(y,x,w^*)\in \mathbb{S}_{\tau}$ such that $\tilde\Omega(y,x,w^*,h_n)>0$ for all $n$ sufficiently large I have
\begin{align*}
    n^{1/2}\left(\tilde\Omega (y,x,w^*,h_n)\right)^{-1/2} \left(\widehat{\rho(y,x)}-\rho(y,x)\right)\xrightarrow{d}N(0,1).
\end{align*}
\end{theorem}

\bibliographystyle{econ}
\bibliography{newproject}
\end{document}